%% file: DGNT09.tex
\documentclass[11pt]{article}

\date{}
\usepackage[T1]{fontenc}
\usepackage[latin9]{inputenc}
\usepackage{float}
\usepackage{amsmath}
\usepackage{graphicx}
\usepackage{amssymb}
\usepackage{verbatim}

\makeatletter

\floatstyle{ruled}
\newfloat{algorithm}{tbp}{loa}
\floatname{algorithm}{Algorithm}

\usepackage{geometry}

\usepackage{float}

\usepackage{framed}

\usepackage{color}

\definecolor{shadecolor}{rgb}{1,0,0}

\makeatletter

\floatstyle{ruled}
\newfloat{algorithm}{tbp}{loa}
\floatname{algorithm}{Algorithm}

\usepackage{theorem}

\newtheorem{theorem}{Theorem}

\newtheorem{lemma}{Lemma}

\newenvironment{proof}{{\noindent\bf Proof. } }{{\hfill $\Box$}}

\def\PROG#1{${\cal #1}$}

\title{Self-Stabilizing Byzantine Asynchronous Unison}

\author{Swan Dubois$^{1,3}$
   \and
   Maria Gradinariu Potop-Butucaru$^{1,4}$
   \and Mikhail Nesterenko$^{2,5}$
   \and S\'ebastien Tixeuil$^{1,6}$
  }

\begin{document}

\maketitle

\footnotetext[1]{The author is with Universit\'e Pierre \& Marie Curie 
    and \mbox{INRIA}, France.}
\footnotetext[2]{The author is with Kent State University, USA.}
\footnotetext[3]{swan.dubois@lip6.fr}
\footnotetext[4]{maria.gradinariu@lip6.fr}
\footnotetext[5]{mikhail@cs.kent.edu}
\footnotetext[6]{sebastien.tixeuil@lip6.fr}

\abstract{We explore asynchronous unison in the presence of systemic
  transient and permanent Byzantine faults in shared memory. We
  observe that the problem is not solvable under less than strongly
  fair scheduler or for system topologies with maximum node degree
  greater than two. We present a self-stabilizing Byzantine-tolerant
  solution to asynchronous unison for chain and ring topologies. Our
  algorithm has minimum possible containment radius and optimal
  stabilization time.  }

\em

\section{Introduction} 

\emph{Asynchronous unison}~\cite{M91j} requires processors to maintain
synchronization between their counters called
\emph{clocks}. Specifically, each processor has to increment its clock
indefinitely while the clock \emph{drift} from its neighbors should
not exceed $1$.  Asynchronous unison is a fundamental building block
for a number of principal tasks in distributed systems such as
distributed snapshots~\cite{CL85j} and synchronization
\cite{A85j,AKMPV07j}.

A practical large-scale distributed system must counter a variety of
transient and permanent faults. A systemic transient fault may perturb
the configuration of the system and leave it in the arbitrary
configuration. \emph{Self-stabilization}~\cite{D74j,D00b,T09bc} is a
versatile technique for transient fault forward recovery.
\emph{Byzantine fault}~\cite{LSP82j} is the most generic permanent
fault model: a faulty processor may behave arbitrarily. However,
designing distributed systems that handle both transient and permanent
faults proved to be rather difficult~\cite{DD05c,DW04j,NA02c}. Some of
the difficulty is due to the inability of the system to counter
Byzantine behavior by relying on the information encoded in the global
system configuration: a transient fault may place the system in an
arbitrary configuration.

In this context considering joint Byzantine and systemic transient
fault tolerance for asynchronous unison appears futile.  Indeed, the
Byzantine processor may keep setting its clock to an arbitrary value
while the clocks of the correct processors are completely out of
synchrony. Hence, we are happy to report that the problem is
solvable. In this paper we present a shared-memory Byzantine-tolerant
self-stabilizing asynchronous unison algorithm that operates chain and
ring system topologies. The algorithm operates under a strongly fair
scheduler. We show that the problem is unsolvable for any other
topology or for less stringent scheduler. Our algorithm achieves
minimal fault-containment radius: each correct processor eventually
synchronizes with its correct neighbors.  We prove our algorithm
correct and demonstrate that its stabilization time is asymptotically
optimal.

\ \\ \textbf{Related work.} The impetus of this work is the study by
Dubois et al~\cite{DPT09ca}. They consider joint tolerance to crash
faults and systemic transient faults. The key observation that enables
this avenue of research is that the definition of asynchronous unison
does not preclude the correct processors from decrementing their
clocks. This allows the processors to synchronize and maintain unison
even while their neighbors may crash or behave arbitrarily.

There are several pure self-stabilizing solutions to the unison
problem~\cite{BPV04c,BPV05c,CFG92c,GH90j}. None of those tolerate
Byzantine faults.

Classic Byzantine fault tolerance focuses on masking the fault. There
are self-stabilizing Byzantine-tolerant clock synchronization
algorithms for completely connected synchronous systems both
probabilistic~\cite{BDH08c,DW04j} and
deterministic~\cite{DH07cb,HDD06c}. The probabilistic and
deterministic solutions tolerate up to one-third and one-fourth of
faulty processors respectively.

Another approach to joint transient and Byzantine tolerance is
\emph{containment}. For tasks whose correctness can be checked
locally, such as vertex coloring, link coloring or dining
philosophers, the fault may be isolated within a region of the
system. \emph{Strict-stabilization} guarantees that there exists a
containment radius outside of which the processors are not affected by
the fault~\cite{MT07j,NA02c,SOM05c}. Yet some problems are not local
and do not admit strict stabilization. However, the tolerance
requirements may weakened to
\emph{strong-stabilization}~\cite{MT06cb,MT08ca} which allows the
processors arbitrarily far from the faulty processor to be
affected. The faulty processor can affect the correct processors only
a finite number of times.  Strong-stabilization enables solution to
several problems, such as tree orientation and tree construction.

\section{Model, Definitions and Notation}

\paragraph{Program syntax and semantics.} A distributed system consists
of $n$ processors that form a communication graph. The processors are
nodes in this graph. The edges of this graph are pairs of processors
that can communicate with each other. Such pairs are \emph{neighbors}.
A \emph{distance} between two processors is the length of the shortest
path between them in this communication graph. Each processor contains
variables and rules. A variable ranges over a fixed domain of values.
A rule is of the form $\langle label \rangle : \langle guard \rangle
\longrightarrow \langle command \rangle$.  A \emph{guard} is a boolean
predicate over processor variables. A \emph{command} is a sequence of
assignment statements.  Processor $p$ may mention its variables
anywhere in its guards and commands.  That is, $p$ can read and update
its variables. However, $p$ may not mention the variables of its
neighbors on the left-hand-sides of the assignment statements of its
commands. That is, $p$ may only read the variables of its neighbors.

A processor is either \emph{correct} or \emph{faulty}.  In this paper
we consider \emph{crash faults} and \emph{Byzantine faults}. A crashed
processor stop the execution of its rules for the remainder of the
run.  A processor affected by Byzantine fault disregards its program
and it may write arbitrary values to variables. Note that, in a given
state, a Byzantine processor exhibits the same state to all its
neighbors.  When the fault type is not explicitly mentioned, the fault
is Byzantine.

An assignment of values to all variables of the system is
\emph{configuration}. A rule whose guard is \textbf{true} in some
system configuration is \emph{enabled} in this configuration, the rule
is \emph{disabled} otherwise. An atomic execution of a subset of
enabled rules transitions the system from one configuration to
another. This transition is a \emph{step}. Note that a faulty
processor is assumed to always have an enabled rule and its step
consists of writing arbitrary values to its variables.  A \emph{run}
of a distributed system is a maximal sequence of such transitions. By
maximality we mean that the sequence is either infinite or ends in a
state where none of the rules are enabled.

\paragraph{Schedulers.}
A \emph{scheduler} (also called \emph{daemon}) is a restriction on the
runs to be considered. The schedulers differ by execution semantics
and by fairness. The scheduler is \emph{synchronous} if in every run
each step contains the execution of every enabled rule. The scheduler
is asynchronous otherwise. There are several types of asynchronous
schedulers.  In the runs of \emph{distributed} (also called
\emph{powerset}) scheduler, a step may contain the execution of an
arbitrary subset of enabled rules. This is the lest restrictive
scheduler.  In the runs of a \emph{central} scheduler, every step
contains the execution of exactly one enabled rule. In the runs of
\emph{locally central} scheduler, the step may contain the execution
of multiple enabled rules as long as none of the rules belong to
neighbor processors.  Central and locally central schedulers are
equivalent. That is, they define the same set of runs. In this paper
we consider these two types of schedulers.

With respect to fairness, the schedulers are classified as
follows. The most restrictive is a \emph{strongly fair scheduler}. In
every run of this scheduler, a rule is executed infinitely often if it
is enabled in infinitely many configurations of the run. Note that the
strongly fair scheduler requires that the rule is executed even if it
continuously keeps being enabled and disabled throughout the run.  A
less restrictive is \emph{weakly fair scheduler}. In every run of this
scheduler, a rule is executed infinitely often if it is enabled in all
but finitely many configurations of the run. That is, the rule has to
be executed only if it is continuously enabled. An \emph{unfair
  scheduler} places no fairness restrictions on the runs of the
distributed system. Faulty processors are not subject to scheduling
restrictions of any of the schedulers: a faulty processor may take no
steps during a run or it may take an infinitely many steps.

\paragraph{Predicates and specifications.} 
A predicate is a boolean function over program configurations. A
configuration \emph{conforms} to some predicate $R$, if $R$ evaluates
to \textbf{true} in this configuration. The configuration
\emph{violates} the predicate otherwise.  Predicate $R$ is
\emph{closed} in a certain program \PROG{P}, if every configuration of
a run of \PROG{P} conforms to $R$ provided that the program starts
from a configuration conforming to $R$. Note that if a program
configuration conforms to $R$ and, after the execution of any step of
\PROG{P}, the resultant configuration also conforms to $R$, then $R$
is closed in \PROG{P}.

A \emph{processor specification} for a processor $p$ defines a set of
configuration sequences. These sequences are formed by variables of
some subset of processors in the system. This subset always includes
$p$ itself. A \emph{problem specification}, or just \emph{problem},
defines specifications for each processor of the system. A problem
specification in the presence of faults defines specifications for
correct processors only.  Program \PROG{P} \emph{solves} problem
\PROG{S} under a certain scheduler if every run of \PROG{P} satisfies
the specifications defined by \PROG{S}.  A closed predicate $I$ is an
\emph{invariant} of program \PROG{P} with respect to problem \PROG{S}
if every run of \PROG{P} that starts in a state conforming to $I$
satisfies \PROG{S}. An $f$-fault $d$-distance invariant $I_{fd}$ is a
particular invariant of \PROG{P} such that if the system has no more
than $f$ processors then in every run that starts in a configuration
conforming to $I_{fd}$, each processor in the distance of at least $d$
away from the fault satisfies the problem \PROG{S}. That is, only
correct processors at distance $d$ or higher have to satisfy the
specification.

A program \PROG{P} is \emph{self-stabilizing} to specification
\PROG{S} if every run of \PROG{P} that starts in an arbitrary
configuration contains a configuration conforming to an invariant of
\PROG{P}.  A program \PROG{P} is \emph{strictly-stabilizing} for $f$
faults and distance $d$, denoted $(f,d)$-\emph{strictly-stabilizing},
to problem \PROG{S} if \PROG{P} converges to an $f$-fault $d$-distance
invariant $I_{fd}$.

\paragraph{Unison specification.} 
Consider the system of processors each of which has a natural number
variable $c$ called \emph{clock}. The clock \emph{drift} between two
processors is the difference between their clock values. Two neighbor
processor are \emph{in unison} if their drift is no more than one.

\emph{Asynchronous unison} specifies that, for every processor $p$,
every program run has to comply with the following two properties.
\begin{description}
\item{\em Safety:} in every configuration, processor $p$ is in unison
  with its neighbors;
\item{\em Liveness:} the clock of processor $p$ is incremented
  infinitely often.
\end{description}

A program that solves the asynchronous unison problem is
\emph{minimal} if the only variable that each processor has it its
clock.

% A program that solves the asynchronous unison problem is
% \emph{priority} if for every 
%
%
% if whenever the clock of a process is either equal or
% one unit behind its neighbors, there exists a finite fragment of
% execution during which the process is continuously activated, its
% neighbors are never activated, and the first modification of the clock
% of the process is an incrementation.

\section{Impossibility Results and Model Justification}

Dubois et al~\cite{DPT09ca} established a number of impossibility
results for asynchronous unison and crash faults. These results are
immediately applicable to Byzantine faults as a Byzantine process may
emulate the crash fault by never executing a step. We summarize their
results in the below theorem.

\begin{theorem}[\cite{DPT09ca}]\label{trm:crash}
There does not exist a minimal $(f,d)$-strictly-stabilizing solution to
the asynchronous unison problem in shared memory for any distance $d
\geq 0$ if the communication graph of the distributed system contains
processors of degree greater than two or if the number of faults is
greater than one or if the scheduler is either unfair or weakly fair.
\end{theorem}

The intuition behind the impossibility results is as follows. If the
system contains a processor $p$ with at least three neighbors, the
neighbors can cycle through their states such that all three are
always in unison with $p$ yet $p$ cannot update its clock without
breaking unison with at least one neighbor. If the system allows two
faults, then the faulty processors may contain such clock values so
far apart that if the correct processors stay in unison with the
faulty ones then they are not able to synchronize with each other. If
the execution scheduler is either unfair or weakly fair then, one
correct processors may cycle through its unison states such that its
neighbor is never given an opportunity to update its clock.

The results of Theorem~\ref{trm:crash} leave the following execution
model that is still open for solutions: system topology with maximum
degree at most two (i.e. a chain or a ring), at most one fault, and a
strongly fair scheduler. We pursue solutions for this particular model
in the remainder of the paper.

\begin{comment}
\begin{table*}
	\centering
	\begin{tabular}{|c|c||c|c|c|c|c|}
	\cline{3-7}
	  \multicolumn{2}{c||}{} &  Unfair & \multicolumn{2}{c|}{Weakly fair} & \multicolumn{2}{c|}{Strongly fair}\tabularnewline
	\cline{4-7}
	  \multicolumn{2}{c||}{} &   & Minimal & Priority & Minimal & Priority \tabularnewline
	\hline
	\hline
		  $f=1$ & $\Delta\geq 3$  &  Imp. & Imp. & Imp. & Imp. \cite{DPT09ca} & Imp. \cite{DPT09ca}    \tabularnewline
	\cline{2-2}\cline{6-7}
  	  & $\Delta\leq 2$  & \cite{DPT09ca} & \cite{DPT09ca} & \cite{DPT09ca}  & \multicolumn{2}{c|}{Pos. (Th. \ref{trm:stabilizationChain}, \ref{trm:specChain}, \ref{trm:stabilizationRing} \& \ref{trm:specRing})}  \tabularnewline
	\cline{1-2}\cline{3-7}
	  \multicolumn{2}{|c||}{$f\geq2$}  & \multicolumn{5}{c|}{Imp. \cite{DPT09ca}}\tabularnewline
	\hline
	\end{tabular}
\caption{Summary of results}
\label{table1}
\end{table*}
\end{comment}

\begin{figure}
\begin{tabbing}
1234567890123\=12345678901234567\=12345\=12345\=12345\=12345\=\kill
\textbf{processor}\ $p$ \\
\textbf{constants}\>  $l, r$: \  left and right neighbors of $p$\\
\> $dg_p$: degree of $p$\\
\textbf{variable} \>$c_p:$ \ natural number, clock value of $p$ \\
\ \\
\textbf{rules} \\
end processor rules \\
\emph{leftEndUp}:\>$(dg_p=1)\wedge(c_p \leq c_r) \longrightarrow c_p := c_r+1$ \\
\emph{leftEndDown}:\>$(dg_p=1)\wedge(c_p > c_r) \longrightarrow c_p := c_r-1$ \\
\emph{rightEndUp} and \emph{rightEndDown} are similar \\
\ \\
middle processor operation rules \\
\emph{middleLeftUp}:\>$(dg_p=2)\wedge(c_p = c_l \vee c_p = c_l-1) \wedge (c_p \leq c_r) \longrightarrow c_p := c_p+1$ \\
\emph{middleLeftDown}:\>$(dg_p=2)\wedge(c_p = c_l \vee c_p = c_l+1) \wedge (c_p > c_r) \longrightarrow c_p := c_p-1$ \\
\emph{middleRightUp} and \emph{middleRightDown} are similar \\
\ \\
middle processor synchronization rules\\
\emph{syncUp}:\>$(dg_p=2)\wedge(c_p < c_l - 1 ) \wedge (c_p < c_r - 1) \longrightarrow c_p := min\{c_l,c_r\}$ \\
\emph{syncDown}:\>$(dg_p=2)\wedge(c_p > c_l + 1) \wedge (c_p > c_r + 1) \longrightarrow c_p := max\{c_l,c_r\}$ \\
\end{tabbing}
\caption{\PROG{SSU}: $(1,0)$-strict-stabilizing asynchronous unison algorithm
% (minimal and priority) 
for chains and rings.}\label{fig:ssu}
\end{figure}

\section{\PROG{SSU}: A Strict-Stabilizing Unison for Chains and Rings}

In this section we present the $(1,0)$-strictly-stabilizing minimal
priority algorithm unison algorithm, prove its correctness and
evaluate its stabilization performance.

\subsection{Algorithm Description}
%
% don't know where to put it yet, MN
%
% A simple observation of the algorithm
% allows us to say that \PROG{SSU} is minimal and priority.

The algorithm can operate on either chain or ring system topologies.
For the description of the algorithm, let us introduce some
topological terminology.  A \emph{middle} processor has two
neighbors. An \emph{end} processor has only one. In a ring every
processor is a middle processor. A chain has two end processors. We
consider the system of processors to be laid out horizontally left to
right. We, therefore, speak of left and right neighbor for a processor
and left and right ends of a chain.

Recall that \emph{drift} between two processors $p$ and $q$ is the
difference between their clock values.  Two processors $p$ and $q$ are
\emph{in unison} if the drift between them is no more than $1$.  An
\emph{island} is a segment of correct processors such that for each
processor $p$, if its neighbor $q$ is also in this island, then $p$
and $q$ are in unison. A processor with no in-unison neighbors is
assumed to be a single-processor island. Note that a faulty processor
never belongs to an island. The \emph{width} of an island is the
number of processors in this island.

The main idea of the algorithm is as follows. Processors form islands
of processors with synchronized clocks. The algorithm is designed such
that the clocks of the processors with adjacent islands drift closer
to each other and the islands eventually merge.  If a faulty processor
restricts the drift of one such island, for example by never changing
its clock, the other islands still drift and synchronize with the
affected island.

\paragraph{Operation description.}
A detailed description of \PROG{SSU} is shown in Figure~\ref{fig:ssu}.
Specifically, \PROG{SSU} operates as follows. Each processor $p$
maintains a single variable $c_p$ where it stores its current clock
value. That is, our algorithm is minimal.

We grouped the processor rules into end processor rules and middle
processor rules. Middle processor rules are further grouped into:
operation --- executed when the processor is in unison with at least
one of its neighbors, and synchronization --- executed otherwise.

At least one rule is always enabled at an end processor. Depending on
the clock value of its neighbor, the left end processor either
increments or decrements its own clock using rules \emph{leftEndUp}
and \emph{leftEndDown}. The operation of the right end processor is
similar.

Let us describe the rules of a middle processor.  If processor $p$ is
in unison with its left neighbor, $p$ can adjust $c_p$ to match its
right neighbor using rules \emph{middleLeftUp} or
\emph{middleLeftDown}. The execution of neither rule breaks the unison
of $p$ and its left neighbor. Similar adjustment is done for the left
neighbor using \emph{middleRightUp} and \emph{middleRightDown}.  Note
that if $p$ is in unison with both of its neighbors and $c_l$ and
$c_r$ differ by 2, none of these rules of $p$ are enabled as any
changes of $c_p$ break the unison with a neighbor of $p$.

If $p$ is in unison with neither of its neighbors, and the clocks of
the two neighbors are either both greater or both less than the clock
of $p$, the processor synchronizes its clock with one of the neighbors
using rule \emph{syncDown} or \emph{syncUp}.

\paragraph{Example operation.} The operation of our algorithm is best 
understood with an example. Figure~\ref{fig:Exemple1} shows the
operation of \PROG{SSU} on a chain without a permanent
fault. Figure~\ref{fig:Exemple2} illustrates the operation of
\PROG{SSU} on a chain with a faulty
processor. Figures~\ref{fig:Exemple3} and~\ref{fig:Exemple4} show the
operation of \PROG{SSU} on rings respectively without and with a
faulty processor.

\begin{figure*}
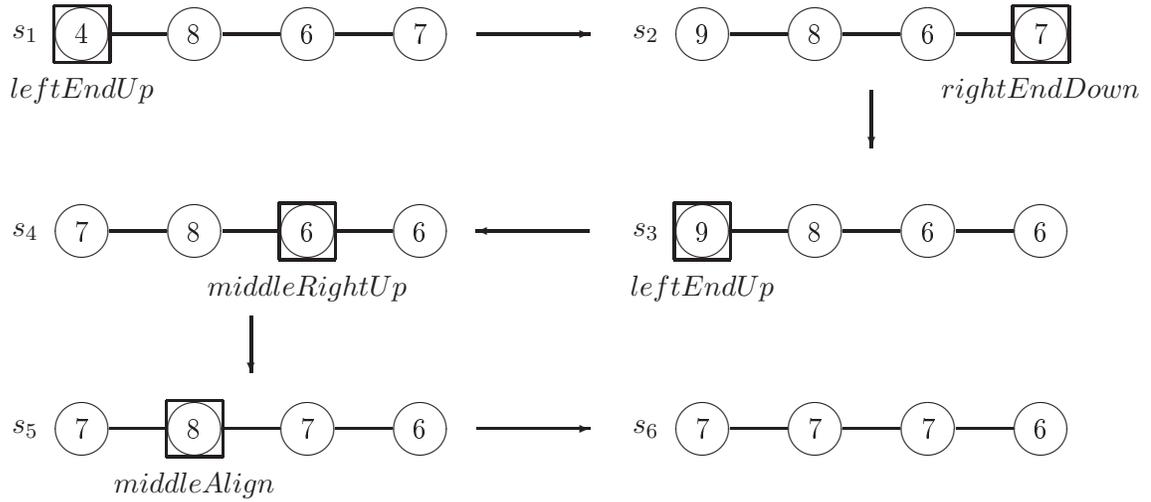

\noindent \begin{centering} \include{Exemple1}
  \par\end{centering}\caption{\label{fig:Exemple1} An example
    operation sequence of \PROG{SSU} on a chain with no
    faults. Numbers represent clock values. Squared processor has an
    enabled rule to be executed.}
\end{figure*}

\begin{figure*}
\noindent \begin{centering} \include{Exemple2}
  \par\end{centering}\caption{\label{fig:Exemple2} An example
    operation sequence of \PROG{SSU} on a chain with a faulty
    processor. Numbers are processor clock values. The faulty
    processor is in double circle. Squared processor has an enabled
    rule to be executed.}
\end{figure*}

\begin{figure*}
\noindent \begin{centering} \include{Exemple3}
  \par\end{centering}\caption{\label{fig:Exemple3} An example
    operation sequence of \PROG{SSU} on a ring with no faults. Numbers
    represent clock values. Squared processor has an enabled rule to
    be executed.}
\end{figure*}

\begin{figure*}
\noindent \begin{centering} \include{Exemple4}
  \par\end{centering}\caption{\label{fig:Exemple4} An example
    operation sequence of \PROG{SSU} on a chain with a faulty
    processor. Numbers are processor clock values. The faulty
    processor is in double circle. Squared processor has an enabled
    rule to be executed.}
\end{figure*}

\subsection{Correctness Proof}\label{sub:correctnessChain}

\paragraph{Chains.}
For chains it is sufficient to consider the operation of the algorithm
for the case where the faulty processor is at the end of the
chain. Indeed, if the faulty processor is in the middle of the chain,
the synchronization of the two segments of correct processors is
independent of each other. Thus, without loss of generality, we assume
that if there exists a faulty processor in the system, it is always
the right end processor.

\begin{lemma}\label{lem:islandClosureChain}
If a run of \PROG{SSU} on a chain starts from a configuration where
two processors $p$ and $q$ belong to the same island, then the two
processors belong to the same island in every configuration of this
run.
\end{lemma}

In other words, Lemma~\ref{lem:islandClosureChain} states that an island
is never broken. The validity of the lemma can be easily ascertained
by the examination of the algorithm's rules as a processor never
de-synchronizes from its in-unison neighbor.

\begin{lemma}\label{lem:islandInfiniteChain}
In every run of \PROG{SSU} on a chain, each processor in the leftmost
island takes an infinite number of steps.
\end{lemma}

\begin{proof} The proof is by induction on the width of the island.
In every configuration, the left end processor has either
\emph{leftEndUp} or \emph{leftEndDown} enabled. Due to the strongly
fair scheduler, this processor takes an infinite number of steps in
every run.

Assume that the left neighbor $l$ of processor $p$ that belongs to the
leftmost island takes an infinite number of steps in the
run. According to Lemma~\ref{lem:islandClosureChain}, $l$ and $p$ are
in unison in every configuration of this run. That is, $l$ and $p$
transition between the three sets of states: $c_l = c_p+1$, $c_l =
c_p$ and $c_l= c_p-1$. See Figure~\ref{fig:islandInfiniteChain} for
illustration. Observe that, regardless of the clock value of the right
neighbor of $p$, if $c_l=c_p$ then $p$ has either \emph{middleLeftUp}
or \emph{middleLeftDown} rule enabled. If $p$ executes this rule, the
system goes either in the state where $c_l = c_p+1$ or $c_l = c_p-1$.
Since $l$ executes infinitely many steps in the run then a
configuration where $c_l=c_p$ repeats infinitely often. That is, one
of $p$'s rules are enabled infinitely often in this run.  Since the
scheduler is strongly fair, $p$ executes infinitely many steps.
\end{proof}

\begin{figure}
\begin{center}
\includegraphics[scale=0.25]{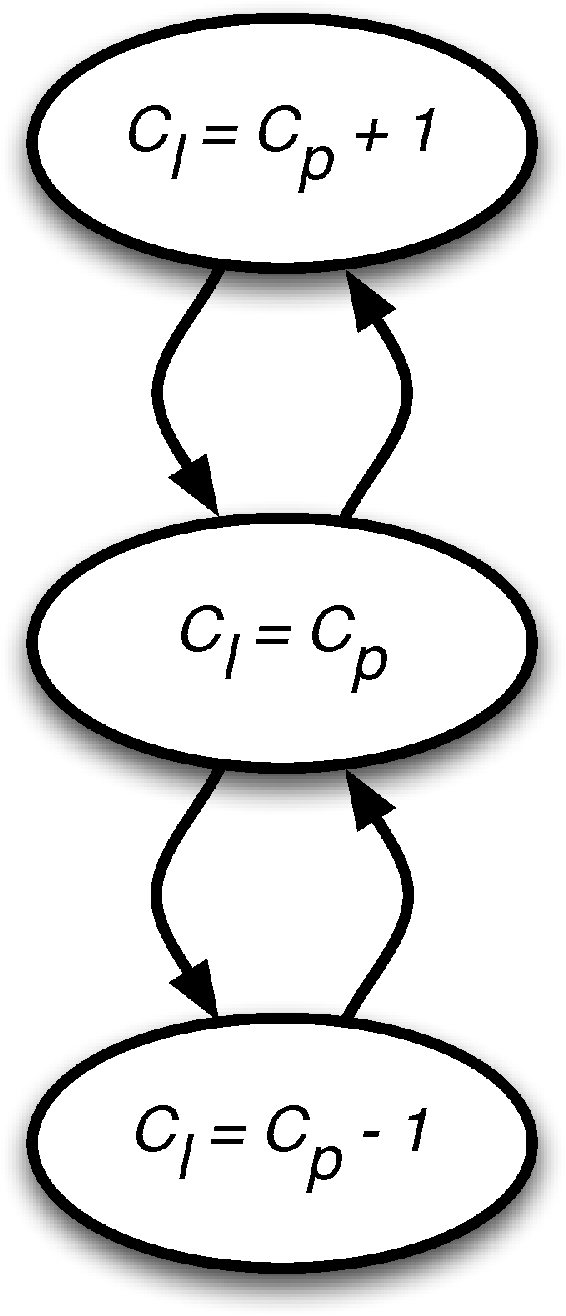}
\end{center}
\caption{The transitions of in-unison neighbor processors $l$ and
  $p$. An illustration for the proof of
  Lemma~\ref{lem:islandInfiniteChain}.}\label{fig:islandInfiniteChain}
\end{figure}

\begin{lemma}\label{lem:islandsMergeChain}
If a run of \PROG{SSU} on a chain starts from a configuration where
processor $p$ belongs to the leftmost island while its right correct
neighbor $r$ does not, then this run contains a configuration where
both $p$ and $r$ belong to the same island.
\end{lemma}

In other words, Lemma~\ref{lem:islandsMergeChain} claims that every two
adjacent islands eventually merge.

\ \\
\begin{proof} 
We prove the lemma by demonstrating that the drift between $p$ and $r$
decreases to zero in every run of \PROG{SSU}.  Let us consider the
rules of $r$. The execution of any rule by $r$ can only decrease the
drift between the two processors. The execution of the rules by $p$
always decreases the drift as well.  According to
Lemma~\ref{lem:islandInfiniteChain}, $p$ takes infinitely many steps
in this run. This means that this run contains a configuration where
the drift between $p$ and $r$ is zero.
\end{proof}

Define the following predicate:

\[\begin{array}{ccl}
INV &\equiv & \text{each correct processor is in unison}\\ 
& & \text{with its correct neighbors}
\end{array}
\]

\begin{theorem}\label{trm:stabilizationChain}
Algorithm~\PROG{SSU} on chains stabilizes to $INV$.
\end{theorem}

\begin{proof}(\emph{sketch})
If every correct processor is in unison with its neighbors, all
correct processors belong to a single island. The closure of \emph{INV}
follows from Lemma~\ref{lem:islandClosureChain}.  Note that
Lemma~\ref{lem:islandsMergeChain} guarantees that the two leftmost
islands eventually merge. The convergence if \PROG{SSU} to \emph{INV} can
be proven by induction on the number of islands in the initial
configuration.
\end{proof}

\begin{theorem}\label{trm:specChain}
Predicate $INV$ is an $(1,0)$-invariant of \PROG{SSU} on chains
with respect to the asynchronous unison problem.
\end{theorem}

In other words, Theorem~\ref{trm:specChain} states that every run of
\PROG{SSU} starting from a configuration conforming to \emph{INV}
satisfies the specification of asynchronous unison.

\ \\
\begin{proof} 
The safety property of the asynchronous unison follows immediately
from the closure of \emph{INV}. Let us consider the liveness property. Once
in unison the only operation that a processor can execute on its clock
is increment or decrement. According to
Lemma~\ref{lem:islandInfiniteChain}, every correct processor of the system
takes an infinite number of steps. Since the clock values are natural
numbers, each processor is bound to execute an infinite number clock
increments. Hence the liveness.
\end{proof}

\paragraph{Rings.} Since there are no end processors on a ring, we 
only have to consider the middle processor rules.

\begin{lemma}\label{lem:islandClosureRing}
If a run of \PROG{SSU} on a ring starts from a configuration where two
processors $p$ and $q$ belong to the same island, then the two
processors belong to the same island in every configuration of this
run.
\end{lemma}

The above lemma is proven similarly to
Lemma~\ref{lem:islandClosureChain}.

\begin{lemma}\label{lem:islandInfiniteRing}
In every run of \PROG{SSU} on a ring, there is an island where every
processor takes an infinite number of steps.
\end{lemma}

\begin{proof}(\emph{sketch}) 
Observe that in every configuration 
of \PROG{SSU} on a ring,
there is at least one correct processor whose clock holds the largest
or the smallest value in the system. This processor has a rule
enabled. Since we consider a strongly fair scheduler, there are
infinitely many steps executed by correct processors in every run of
\PROG{SSU}. Since there are finitely many correct processors, at least
one correct processor takes infinitely many steps. Let us consider the
island to which this processor belongs. The rest of the lemma is
proven by induction on the width of this island similar to
Lemma~\ref{lem:islandInfiniteChain}.
\end{proof}

\begin{lemma}\label{lem:islandsMergeRing}
If a run of \PROG{SSU} starts from a configuration where there is more
than one island, then this run contains a configuration where some two
islands merge.
\end{lemma}

\begin{proof}(\emph{sketch}) 
Let us consider the initial configuration of \PROG{SSU} on a ring with
more than one island.  According to
Lemma~\ref{lem:islandInfiniteRing}, there is at least one island in
this configuration where every processor takes an infinite number of
steps. Assume, without loss of generality, that this island has an
adjacent island to the right. An argument similar to the one employed
in the proof of Lemma~\ref{lem:islandsMergeChain} demonstrates that
these islands eventually merge.
\end{proof}

The below two theorems are proven similarly to their equivalents for
the chain topology.

\begin{theorem}\label{trm:stabilizationRing}
Algorithm~\PROG{SSU} on rings stabilizes to $INV$.
\end{theorem}

\begin{theorem}\label{trm:specRing}
Predicate $INV$ is an $(1,0)$-invariant of \PROG{SSU} on rings
with respect to the asynchronous unison problem.
\end{theorem}

\subsection{Stabilization Time}

In this section, we compute the stabilization time of \PROG{SSU}. We
estimate the stabilization time in the number of asynchronous
rounds. In general, this notion is somewhat tricky to define for
strongly fair scheduler, at the actions of processors may become
disabled and then enabled an arbitrary many times before
execution. However, this definition simplifies for the case of
\PROG{SSU} as every correct processor takes an infinite number of
steps. We define an \emph{asynchronous round} to be the smallest
segment of a run of the algorithm where every correct process executes
a step.

\paragraph{Upper bound of \PROG{SSU}.} First, we show that \PROG{SSU} needs at most $L$ rounds to stabilize where $L$ is the largest clock drift between correct processors in the system.

\begin{theorem}\label{trm:upperBound}
The stabilization time of \PROG{SSU} is in $O(L)$ rounds both on chains and rings  where $L$ is the maximum clock drift between two correct neighbors in the initial configuration.
\end{theorem}

\begin{proof}
Assume that there exists an execution $\omega$ such that there exists at least two distinct islands $I_1$ and $I_2$ at the end of the round $L_\omega$ (where $L_\omega$ is the maximum clock drift between two correct neighbors in the initial configuration of $\omega$). Note that $L_\omega \geq 2$. Otherwise, any processor is in unison with its neighbor in the initial configuration and Lemma \ref{lem:islandClosureChain} or \ref{lem:islandClosureRing} implies $I_1$ and $I_2$ are never distinct. 

Let $p$ and $q$ be two neighbor processors such that $p\in I_1$ and $q\in I_2$. Without loss of generality, we can assume that $c_q<c_p$ in the initial configuration of $\omega$. By construction, we have $c_p-c_q\leq L_\omega$.

While $I_1$ and $I_2$ are distinct, according to the proof of Lemma \ref{lem:islandsMergeChain} or \ref{lem:islandsMergeRing}, the following property holds: $c_q<c_p$.

In the case where the system is a chain, note that $p$ and $q$ are not end processors. Otherwise, $p$ and $q$ are in unison at the end of the first round since the end processor synchronizes its clock with the one of its neighbor at its first activation and this contradicts the construction of $\omega$ and the fact that $L_\omega \geq 2$.

Now, we can observe that any activation of $p$ by a middle processor operation or synchronization rule can only decrease the clock value of $p$ by at least one.
 
Following the definition of asynchronous round, there is at least one activation of $p$ during each round of $\omega$. Then, we can conclude that, at the end of the round $i$ ($1\leq i\leq L_\omega$), we have: $c_p-c_q\leq L_\omega -i$. 

We can deduce that $p$ and $q$ are necessarily in unison at the end of the round $L_\omega -1$ which contradicts the construction of $\omega$. Then, the stabilization time of \PROG{SSU} is in $O(L)$ rounds both on chains and rings. Hence the result.
\end{proof}

\paragraph{Lower bound on chains.} Then, we show that any $(1,0)$-strictly-stabilizing deterministic minimal asynchronous unison on a chain needs at least $L$ rounds to stabilize where  $L$ is the largest clock drift between correct processors in the system.

In the following lemmas, \PROG{A} denotes any $(1,0)$-strictly-stabilizing deterministic minimal asynchronous unison on a chain under a central strongly fair scheduler.

\begin{lemma}\label{lem:middleClosureChain}
When a middle processor is in unison with only one of its neighbors, any enabled rule of \PROG{A} for this processor maintains this unison.
\end{lemma}

\begin{proof}
Assume that there exists a set of clock values $\{a,b,c\}$ (with $|a-b|\leq 1$ and $|b-c|\geq 2$) such that a middle processor $p$ is enabled by a rule $R$ of \PROG{A} when $c_p=b$ and neighbors clock are respectively $a$ and $c$ and that $R$ modifies $c_p$ into a value $b'$ (with $|a-b'|\geq 2$).

\begin{figure}
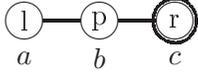

\noindent \begin{centering} \include{ExMiddleClosure}
  \par\end{centering}\caption{\label{fig:exMiddleClosure}Configuration used in proof of Lemma \ref{lem:middleClosureChain}}
\end{figure}

Then, consider the following initial configuration: $V=\{l,p,r\}$, $E=\{\{l,p\},\{p,r\}\}$, $r$ is Byzantine and $c_l=a$, $c_p=b$, $c_r=c$ (see Figure \ref{fig:exMiddleClosure}). We can observe that this configuration satisfies $INV$. By construction, $p$ is enabled by $R$ in this configuration (recall that \PROG{A} is minimal and deterministic). If the scheduler chooses $p$, then we obtain a configuration which does not satisfy $INV$. Hence, \PROG{A} does not respect the closure of the safety property of asynchronous unison. This is contradictory with its construction.
\end{proof}

\begin{lemma}\label{lem:middleActivatedChain}
When a middle processor $p$ is in unison with only one of its neighbors (denote by $q$ the other neighbor of $p$), the following property holds: in any execution starting from this configuration in which $q$ remains not synchronized with $p$, $p$ moves its clock closer to the clock of $q$ in a finite time.
\end{lemma}

\begin{proof}
Assume that there exists a set  of clock values $\{a,b,c\}$ (with $|a-b|\leq 1$ and $|b-c|\geq 2$) such that there exists an execution $\omega$ starting from a configuration (in which $c_p=b$ and neighbors clock are respectively $a$ and $c$ -- denote by $q$ the processor such that $c_q=c$) in which $q$ remains not synchronized with $p$ and in which $p$ never moves its clock closer to the clock of $q$.

\begin{figure}
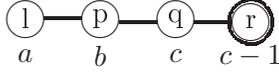

\noindent \begin{centering} \include{ExMiddleActivated}
  \par\end{centering}\caption{\label{fig:exMiddleActivated}Configuration used in proof of Lemma \ref{lem:middleActivatedChain}}
\end{figure}

We deal with the case where $b>c$ (the case where $b<c$ is similar). Then, consider the following initial configuration $s_0$: $V=\{l,p,q,r\}$, $E=\{\{l,p\},\{p,q\},\{q,r\}\}$, $r$ is Byzantine and $c_l=a$, $c_p=b$, $c_q=c$, $c_r=c-1$ (see Figure \ref{fig:exMiddleActivated}). If $r$ acts as a crashed processor, its clock value remains constant. Then, by Lemma \ref{lem:middleClosureChain}, we have $c_q\in\{c,c-1,c-2\}$ in any state of any execution starting from $s_0$. Hence, $p$ can not distinguish this execution from $\omega$ (recall that \PROG{A} is minimal and deterministic). Consequently, there exists an execution starting from $s_0$ such that $c_p\geq b$ and $c_q\leq c$ in any state. This contradicts the convergence property of \PROG{A}.
\end{proof}

\begin{lemma}\label{lem:endEnabledChain}
When an end processor is in unison with its neighbor, there exists an enabled rule of \PROG{A} for this processor.
\end{lemma}

\begin{proof}
Assume that there exists a set of clock values $\{a,b\}$ (with $|a-b|\leq 1$) such that an end processor $p$ is not enabled by any rule of \PROG{A} when $c_p=a$ and its neighbor clock is $b$.

\begin{figure}
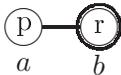

\noindent \begin{centering} \include{ExEndEnabled}
  \par\end{centering}\caption{\label{fig:exEndEnabled}Configuration used in proof of Lemma \ref{lem:endEnabledChain}}
\end{figure}

Then, consider the following initial configuration: $V=\{p,r\}$, $E=\{\{p,r\}\}$, $r$ is Byzantine and $c_p=a$, $c_r=b$ (see Figure \ref{fig:exEndEnabled}). By construction, $p$ is not enabled in this configuration (recall that \PROG{A} is minimal and deterministic). Assume now that $r$ acts as a crashed processor. Then, we can observe that $p$ is never enabled in this execution, that contradicts the liveness property of $(1,0)$-strictly-stabilizing asynchronous unison.
\end{proof}

If we consider the execution described in the proof of Lemma \ref{lem:endEnabledChain}, we can observe that $p$ is infinitely often activated (by fairness assumption) and that its clock is always in the set $\{b-1,b,b+1\}$ (by closure of \PROG{A}). Since \PROG{A} is minimal and deterministic, we can deduce that values of $c_p$ over this execution follow a given cycle. We characterize now \PROG{A} by this cycle. More formally, we say that:

\begin{enumerate}
\item \PROG{A} is of type \textbf{1} if the cycle is $b,b+1,b,b+1,\ldots$.
\item \PROG{A} is of type \textbf{2} if the cycle is $b,b-1,b,b-1,\ldots$.
\item \PROG{A} is of type \textbf{3} if the cycle is $b,b+1,b-1,b,b+1,b-1,\ldots$.
\end{enumerate}

Notice that the protocol \PROG{SSU} is of type \textbf{1}.

\begin{theorem}\label{trm:lowerBoundChain}
The stabilization time of any $(1,0)$-strictly-stabilizing deterministic minimal asynchronous unison on chains is in $\Omega(L)$ where $L$ is the maximum clock drift between two correct neighbors in the initial configuration.
\end{theorem}

\begin{proof}
Assume that \PROG{A} is a $(1,0)$-strictly-stabilizing deterministic minimal asynchronous unison on chains.

We provide the proof of this theorem in the case where \PROG{A} is of type \textbf{1} since other cases are similar.

Let $a,t$ be natural numbers. Consider the following initial configuration $s^0$: $V=\{p,q,r,s\}$, $E=\{\{p,q\},\{q,r\},\{r,s\}\}$, $s$ is Byzantine and $c_p=a+2t$, $c_q=a+2t$, $c_r=a$, $c_s=a$ (see Figure \ref{fig:exOmega}). Hence, we have a maximal clock drift of $L=2t$.

\begin{figure}
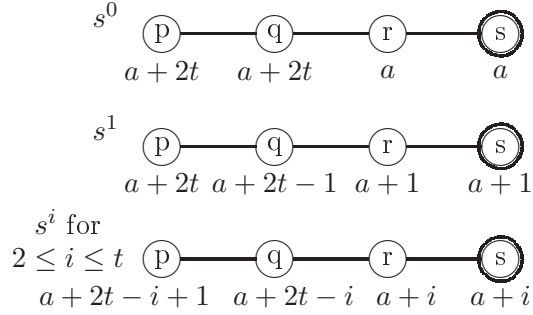

\noindent \begin{centering} \include{ExOmega}
  \par\end{centering}\caption{\label{fig:exOmega}Configurations used in proof of Theorem \ref{trm:lowerBoundChain}}
\end{figure}

Note that $p$ is enabled to take the value $a+2t+1$ in $s^0$ (by Lemma \ref{lem:endEnabledChain} and the fact that \PROG{A} is minimal and of type \textbf{1}). By Lemmas \ref{lem:middleActivatedChain}, \ref{lem:middleClosureChain}, and the fact that \PROG{A} is minimal, we can deduce that $q$ is enabled to take the value $a+2t-1$ only when $c_p=a+2t$. Similar reasoning holds for $r$ which is enabled to take the value $a+1$ when $c_s=a$.

Then, the following execution of \PROG{A} is possible: $p$ is activated and takes value $a+2t+1$, $p$ is activated and takes value $a+2t$ ($p$ is enabled by Lemma \ref{lem:endEnabledChain} and the new value is determined by the type of \PROG{A}), $q$ is activated and takes value $a+2t-1$, $r$ is activated and takes value $a+1$ and $s$ takes the value $a+1$ (recall that $s$ is byzantine). We obtain the configuration $s^1$ depicted in Figure \ref{fig:exOmega}.

We can observe that the first round $R_1$ of our execution ends in $s^1$ and that we have now a maximal clock drift of $a+2(t-1)$.

By the same reasoning, we can construct a sequence of $t-1$ rounds $R_{i}=s^{i-1}\ldots s^i$ ($2\leq i\leq t$) as follows: $p$ is activated and takes value $a+2t+1-i$, $q$ is activated and takes value $a+2t-i$, $r$ is activated and takes value $a+i$ and $s$ takes the value $i$. We obtain the configuration $s^i$ at the end of round $R_i$ ($2\leq i\leq t$) depicted in Figure \ref{fig:exOmega}. At the end of round $R_i$ ($2\leq i\leq t$), we have a maximal clock drift of $2(t-i)$.

We can conclude that, at the end of the round $R_{t-1}$, the maximal clock drift is $2$ whereas, at the end of the round $R_t$, the maximal clock drift is $1$ (since we have $c_p-c_q=1$ and $c_q-c_r=0$). By construction of $t$, we can conclude that \PROG{A} needs $\Omega(L)$ rounds to stabilize.
\end{proof}

\paragraph{Lower bound on rings.} Then, we show that any $(1,0)$-strictly-stabilizing deterministic minimal asynchronous unison on a chain needs at least $L$ rounds to stabilize where  $L$ is the largest clock drift between correct processors in the system.

In the following lemmas, \PROG{A} denotes any $(1,0)$-strictly-stabilizing deterministic minimal asynchronous unison on a ring under a central strongly fair scheduler.

\begin{lemma}\label{lem:middleClosureRing}
When a processor is in unison with only one of its neighbors, any enabled rule of \PROG{A} for this processor maintains this unison.
\end{lemma}

\begin{proof}
The proof of Lemma \ref{lem:middleClosureChain} directly applies here if we consider the following system: $V=\{p,q,r\}$ and $E=\{\{p,q\},\{q,r\},\{r,p\}\}$.
\end{proof}

\begin{lemma}\label{lem:middleActivatedRing}
When a processor $p$ is in unison with only one of its neighbors (denote by $q$ the other neighbor of $p$), the following property holds: in any execution starting from this configuration in which $q$ remains not synchronized with $p$, $p$ moves its clock closer to the clock of $q$ in a finite time.
\end{lemma}

\begin{proof}
The proof of Lemma \ref{lem:middleActivatedChain} directly applies here if we consider the following system: $V=\{p,q,r,s\}$ and $E=\{\{p,q\},\{q,r\},\{r,s\},\{s,p\}\}$.
\end{proof}

\begin{theorem}\label{trm:lowerBoundRing}
The stabilization time of any $(1,0)$-strictly-stabilizing deterministic minimal asynchronous unison on rings is in $\Omega(L)$ where $L$ is the maximum clock drift between two correct neighbors in the initial configuration.
\end{theorem}

\begin{proof}
Assume that \PROG{A} is a $(1,0)$-strictly-stabilizing deterministic minimal asynchronous unison on rings.

Let $a,t$ be natural numbers. Consider the following initial configuration $s^0$: $V=\{p,q,r,s,t\}$, $E=\{\{p,q\},\{q,r\},\{r,s\},\{s,t\},\{t,p\}\}$, $r$ is Byzantine and $c_p=c_t=a+2t$, $c_q=c_s=c_r=a$ (see Figure \ref{fig:exOmega2}). Hence, we have a maximal clock drift of $L=2t$.

\begin{figure}
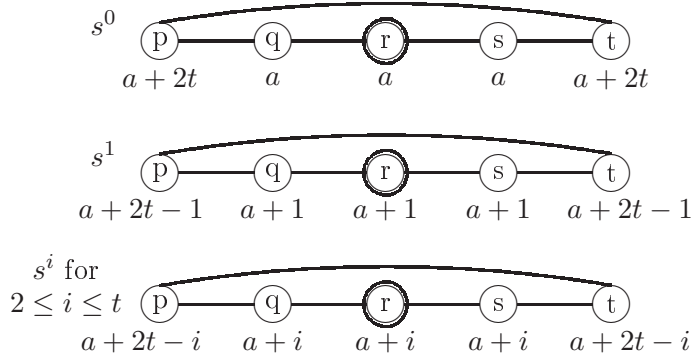

\noindent \begin{centering} \include{ExOmega2}
  \par\end{centering}\caption{\label{fig:exOmega2}Configurations used in proof of Theorem \ref{trm:lowerBoundRing}}
\end{figure}

Note that $p$ and $t$ are enabled to take the value $a+2t-1$ in $s^0$ (by Lemmas \ref{lem:middleActivatedRing} and \ref{lem:middleClosureRing} and the fact that \PROG{A} is minimal).By similar reasoning, we can deduce that $q$ and $s$ are enabled to take the value $a+1$.

Then, the following execution of \PROG{A} is possible: $p$ is activated and takes value $a+2t-1$, $t$ is activated and takes value $a+2t-1$, $q$ is activated and takes value $a+1$, $s$ is activated and takes value $a+1$ and $s$ takes the value $a+1$ (recall that $s$ is byzantine). We obtain the configuration $s^1$ depicted in Figure \ref{fig:exOmega2}.

We can observe that the first round $R_1$ of our execution ends in $s^1$ and that we have now a maximal clock drift of $a+2(t-1)$.

By the same reasoning, we can construct a sequence of $t-1$ rounds $R_{i}=s^{i-1}\ldots s^i$ ($2\leq i\leq t$) as follows:  $p$ is activated and takes value $a+2t-i$, $t$ is activated and takes value $a+2t-i$, $q$ is activated and takes value $a+i$, $s$ is activated and takes value $a+i$ and $s$ takes the value $a+i$ (recall that $s$ is byzantine). We obtain the configuration $s^i$ at the end of round $R_i$ ($2\leq i\leq t$) depicted in Figure \ref{fig:exOmega2}. At the end of round $R_i$ ($2\leq i\leq t$), we have a maximal clock drift of $2(t-i)$.

We can conclude that, at the end of the round $R_{t-1}$, the maximal clock drift is $2$ whereas, at the end of the round $R_t$, the maximal clock drift is $0$. By construction of $t$, we can conclude that \PROG{A} needs $\Omega(L)$ rounds to stabilize.
\end{proof}

\paragraph{Conclusion.} Let us review our conclusions so far. Theorem~\ref{trm:upperBound} proves
that the stabilization complexity of \PROG{SSU} is in $O(L)$ rounds while
Theorems~\ref{trm:lowerBoundChain} and \ref{trm:lowerBoundRing} show that any $(1,0)$-strict-stabilizing
algorithm requires at least that many rounds to stabilize. The
following theorem summarizes these results.

\begin{theorem}The stabilization complexity of \PROG{SSU} is optimal. 
It stabilizes in $\Theta(L)$ asynchronous rounds where L is the
largest drift between correct processors.
\end{theorem}

\section{Conclusion}

In this paper we explored joint tolerance to Byzantine and systemic
transient faults for the asynchronous unison problem in shared
memory. The presence of algorithms that tolerate both fault classes
poses the question for further study: what are the properties of such
algorithms in more concrete execution models of finer atomicity such
as shared registers or message-passing. Lower atomicity models tend to
empower faulty processors. Indeed, in shared register model, the
Byzantine processor on a ring may report differing clock values to its
right and left neighbors complicating fault recovery. In our future
work we would like to pursue this research question.

\bibliographystyle{plain}
%\bibliography{../../../biblio/biblio}
\bibliography{biblio}

\end{document}

%% file: Exemple1.tex
\ifx\JPicScale\undefined\def\JPicScale{0.75}\fi
\unitlength \JPicScale mm
\begin{picture}(190,90)(0,0)
\linethickness{0.3mm}
\put(15,85){\circle{10}}

\linethickness{0.3mm}
\put(35,85){\circle{10}}

\linethickness{0.3mm}
\put(55,85){\circle{10}}

\linethickness{0.3mm}
\put(75,85){\circle{10}}

\linethickness{0.3mm}
\put(185,85){\circle{10}}

\linethickness{0.3mm}
\put(145,85){\circle{10}}

\linethickness{0.3mm}
\put(125,85){\circle{10}}

\linethickness{0.3mm}
\put(20,85){\line(1,0){10}}
\linethickness{0.3mm}
\put(40,85){\line(1,0){10}}
\linethickness{0.3mm}
\put(60,85){\line(1,0){10}}
\linethickness{0.3mm}
\put(130,85){\line(1,0){10}}
\linethickness{0.3mm}
\put(150,85){\line(1,0){10}}
\linethickness{0.3mm}
\put(170,85){\line(1,0){10}}
\linethickness{0.3mm}
\put(185,50){\circle{10}}

\linethickness{0.3mm}
\put(165,50){\circle{10}}

\linethickness{0.3mm}
\put(145,50){\circle{10}}

\linethickness{0.3mm}
\put(125,50){\circle{10}}

\linethickness{0.3mm}
\put(130,50){\line(1,0){10}}
\linethickness{0.3mm}
\put(150,50){\line(1,0){10}}
\linethickness{0.3mm}
\put(170,50){\line(1,0){10}}
\linethickness{0.3mm}
\put(185,15){\circle{10}}

\linethickness{0.3mm}
\put(165,15){\circle{10}}

\linethickness{0.3mm}
\put(145,15){\circle{10}}

\linethickness{0.3mm}
\put(125,15){\circle{10}}

\linethickness{0.3mm}
\put(130,15){\line(1,0){10}}
\linethickness{0.3mm}
\put(150,15){\line(1,0){10}}
\linethickness{0.3mm}
\put(170,15){\line(1,0){10}}
\linethickness{0.3mm}
\put(75,50){\circle{10}}

\linethickness{0.3mm}
\put(55,50){\circle{10}}

\linethickness{0.3mm}
\put(35,50){\circle{10}}

\linethickness{0.3mm}
\put(15,50){\circle{10}}

\linethickness{0.3mm}
\put(20,50){\line(1,0){10}}
\linethickness{0.3mm}
\put(40,50){\line(1,0){10}}
\linethickness{0.3mm}
\put(60,50){\line(1,0){10}}
\linethickness{0.3mm}
\put(75,15){\circle{10}}

\linethickness{0.3mm}
\put(55,15){\circle{10}}

\linethickness{0.3mm}
\put(15,15){\circle{10}}

\linethickness{0.3mm}
\put(20,15){\line(1,0){10}}
\linethickness{0.3mm}
\put(40,15){\line(1,0){10}}
\linethickness{0.3mm}
\put(60,15){\line(1,0){10}}
\linethickness{0.3mm}
\put(85,85){\line(1,0){20}}
\put(105,85){\vector(1,0){0.12}}
\linethickness{0.3mm}
\put(85,50){\line(1,0){20}}
\put(85,50){\vector(-1,0){0.12}}
\linethickness{0.3mm}
\put(155,65){\line(0,1){10}}
\put(155,65){\vector(0,-1){0.12}}
\linethickness{0.3mm}
\put(45,25){\line(0,1){10}}
\put(45,25){\vector(0,-1){0.12}}
\linethickness{0.3mm}
\put(85,15){\line(1,0){20}}
\put(105,15){\vector(1,0){0.12}}
\put(15,85){\makebox(0,0)[cc]{4}}

\put(55,85){\makebox(0,0)[cc]{6}}

\put(75,85){\makebox(0,0)[cc]{7}}

\put(165,85){\makebox(0,0)[cc]{6}}

\put(185,85){\makebox(0,0)[cc]{7}}

\put(5,85){\makebox(0,0)[cc]{$s_1$}}

\put(115,85){\makebox(0,0)[cc]{$s_2$}}

\put(115,50){\makebox(0,0)[cc]{$s_3$}}

\put(5,50){\makebox(0,0)[cc]{$s_4$}}

\put(5,15){\makebox(0,0)[cc]{$s_5$}}

\put(115,15){\makebox(0,0)[cc]{$s_6$}}

\linethickness{0.3mm}
\put(35,15){\circle{10}}

\linethickness{0.3mm}
\put(10,90){\line(1,0){10}}
\put(10,80){\line(0,1){10}}
\put(20,80){\line(0,1){10}}
\put(10,80){\line(1,0){10}}
\linethickness{0.3mm}
\put(120,55){\line(1,0){10}}
\put(120,45){\line(0,1){10}}
\put(130,45){\line(0,1){10}}
\put(120,45){\line(1,0){10}}
\put(15,50){\makebox(0,0)[cc]{}}

\linethickness{0.3mm}
\put(165,85){\circle{10}}

\linethickness{0.3mm}
\put(180,90){\line(1,0){10}}
\put(180,80){\line(0,1){10}}
\put(190,80){\line(0,1){10}}
\put(180,80){\line(1,0){10}}
\put(165,50){\makebox(0,0)[cc]{6}}

\put(185,50){\makebox(0,0)[cc]{6}}

\put(75,50){\makebox(0,0)[cc]{6}}

\put(55,50){\makebox(0,0)[cc]{6}}

\put(35,85){\makebox(0,0)[cc]{8}}

\put(145,85){\makebox(0,0)[cc]{8}}

\put(125,85){\makebox(0,0)[cc]{9}}

\put(145,50){\makebox(0,0)[cc]{8}}

\put(125,50){\makebox(0,0)[cc]{9}}

\put(35,50){\makebox(0,0)[cc]{8}}

\put(15,50){\makebox(0,0)[cc]{7}}

\linethickness{0.3mm}
\put(50,55){\line(1,0){10}}
\put(50,45){\line(0,1){10}}
\put(60,45){\line(0,1){10}}
\put(50,45){\line(1,0){10}}
\put(55,15){\makebox(0,0)[cc]{7}}

\put(75,15){\makebox(0,0)[cc]{6}}

\put(35,15){\makebox(0,0)[cc]{8}}

\put(15,15){\makebox(0,0)[cc]{7}}

\linethickness{0.3mm}
\put(30,20){\line(1,0){10}}
\put(30,10){\line(0,1){10}}
\put(40,10){\line(0,1){10}}
\put(30,10){\line(1,0){10}}
\put(125,15){\makebox(0,0)[cc]{7}}

\put(145,15){\makebox(0,0)[cc]{7}}

\put(165,15){\makebox(0,0)[cc]{7}}

\put(185,15){\makebox(0,0)[cc]{6}}

\put(115,30){\makebox(0,0)[cc]{}}

\put(15,75){\makebox(0,0)[cc]{$leftEndUp$}}

\put(185,75){\makebox(0,0)[cc]{$rightEndDown$}}

\put(125,40){\makebox(0,0)[cc]{$leftEndUp$}}

\put(55,40){\makebox(0,0)[cc]{$middleRightUp$}}

\put(35,5){\makebox(0,0)[cc]{$middleAlign$}}

\end{picture}

%% file: Exemple2.tex
\ifx\JPicScale\undefined\def\JPicScale{0.75}\fi
\unitlength \JPicScale mm
\begin{picture}(190.78,90.93)(0,0)
\linethickness{0.3mm}
\put(15,85){\circle{10}}

\linethickness{0.3mm}
\put(35,85){\circle{10}}

\linethickness{0.3mm}
\put(55,85){\circle{10}}

\linethickness{0.3mm}
\put(75,85){\circle{10}}

\linethickness{0.3mm}
\put(185,85){\circle{10}}

\linethickness{0.3mm}
\put(165,85){\circle{10}}

\linethickness{0.3mm}
\put(145,85){\circle{10}}

\linethickness{0.3mm}
\put(125,85){\circle{10}}

\linethickness{0.3mm}
\put(20,85){\line(1,0){10}}
\linethickness{0.3mm}
\put(40,85){\line(1,0){10}}
\linethickness{0.3mm}
\put(60,85){\line(1,0){10}}
\linethickness{0.3mm}
\put(130,85){\line(1,0){10}}
\linethickness{0.3mm}
\put(150,85){\line(1,0){10}}
\linethickness{0.3mm}
\put(170,85){\line(1,0){10}}
\linethickness{0.3mm}
\put(185,50){\circle{10}}

\linethickness{0.3mm}
\put(145,50){\circle{10}}

\linethickness{0.3mm}
\put(125,50){\circle{10}}

\linethickness{0.3mm}
\put(130,50){\line(1,0){10}}
\linethickness{0.3mm}
\put(150,50){\line(1,0){10}}
\linethickness{0.3mm}
\put(170,50){\line(1,0){10}}
\linethickness{0.3mm}
\put(185,15){\circle{10}}

\linethickness{0.3mm}
\put(165,15){\circle{10}}

\linethickness{0.3mm}
\put(145,15){\circle{10}}

\linethickness{0.3mm}
\put(125,15){\circle{10}}

\linethickness{0.3mm}
\put(130,15){\line(1,0){10}}
\linethickness{0.3mm}
\put(150,15){\line(1,0){10}}
\linethickness{0.3mm}
\put(170,15){\line(1,0){10}}
\linethickness{0.3mm}
\put(75,50){\circle{10}}

\linethickness{0.3mm}
\put(55,50){\circle{10}}

\linethickness{0.3mm}
\put(35,50){\circle{10}}

\linethickness{0.3mm}
\put(15,50){\circle{10}}

\linethickness{0.3mm}
\put(20,50){\line(1,0){10}}
\linethickness{0.3mm}
\put(40,50){\line(1,0){10}}
\linethickness{0.3mm}
\put(60,50){\line(1,0){10}}
\linethickness{0.3mm}
\put(75,15){\circle{10}}

\linethickness{0.3mm}
\put(55,15){\circle{10}}

\linethickness{0.3mm}
\put(35,15){\circle{10}}

\linethickness{0.3mm}
\put(15,15){\circle{10}}

\linethickness{0.3mm}
\put(20,15){\line(1,0){10}}
\linethickness{0.3mm}
\put(40,15){\line(1,0){10}}
\linethickness{0.3mm}
\put(60,15){\line(1,0){10}}
\linethickness{0.3mm}
\put(85,85){\line(1,0){20}}
\put(105,85){\vector(1,0){0.12}}
\linethickness{0.3mm}
\put(85,50){\line(1,0){20}}
\put(85,50){\vector(-1,0){0.12}}
\linethickness{0.3mm}
\put(155,65){\line(0,1){10}}
\put(155,65){\vector(0,-1){0.12}}
\linethickness{0.3mm}
\put(45,25){\line(0,1){10}}
\put(45,25){\vector(0,-1){0.12}}
\linethickness{0.3mm}
\put(85,15){\line(1,0){20}}
\put(105,15){\vector(1,0){0.12}}
\linethickness{0.3mm}
\put(80.78,84.75){\line(0,1){0.49}}
\multiput(80.74,84.27)(0.04,0.49){1}{\line(0,1){0.49}}
\multiput(80.66,83.78)(0.08,0.48){1}{\line(0,1){0.48}}
\multiput(80.54,83.31)(0.12,0.48){1}{\line(0,1){0.48}}
\multiput(80.39,82.84)(0.16,0.46){1}{\line(0,1){0.46}}
\multiput(80.2,82.4)(0.1,0.22){2}{\line(0,1){0.22}}
\multiput(79.97,81.96)(0.11,0.22){2}{\line(0,1){0.22}}
\multiput(79.71,81.55)(0.13,0.21){2}{\line(0,1){0.21}}
\multiput(79.41,81.17)(0.15,0.19){2}{\line(0,1){0.19}}
\multiput(79.09,80.81)(0.11,0.12){3}{\line(0,1){0.12}}
\multiput(78.74,80.47)(0.12,0.11){3}{\line(1,0){0.12}}
\multiput(78.36,80.17)(0.13,0.1){3}{\line(1,0){0.13}}
\multiput(77.96,79.91)(0.2,0.13){2}{\line(1,0){0.2}}
\multiput(77.54,79.67)(0.21,0.12){2}{\line(1,0){0.21}}
\multiput(77.1,79.48)(0.22,0.1){2}{\line(1,0){0.22}}
\multiput(76.65,79.32)(0.45,0.16){1}{\line(1,0){0.45}}
\multiput(76.19,79.2)(0.46,0.12){1}{\line(1,0){0.46}}
\multiput(75.72,79.12)(0.47,0.08){1}{\line(1,0){0.47}}
\multiput(75.24,79.08)(0.48,0.04){1}{\line(1,0){0.48}}
\put(74.76,79.08){\line(1,0){0.48}}
\multiput(74.28,79.12)(0.48,-0.04){1}{\line(1,0){0.48}}
\multiput(73.81,79.2)(0.47,-0.08){1}{\line(1,0){0.47}}
\multiput(73.35,79.32)(0.46,-0.12){1}{\line(1,0){0.46}}
\multiput(72.9,79.48)(0.45,-0.16){1}{\line(1,0){0.45}}
\multiput(72.46,79.67)(0.22,-0.1){2}{\line(1,0){0.22}}
\multiput(72.04,79.91)(0.21,-0.12){2}{\line(1,0){0.21}}
\multiput(71.64,80.17)(0.2,-0.13){2}{\line(1,0){0.2}}
\multiput(71.26,80.47)(0.13,-0.1){3}{\line(1,0){0.13}}
\multiput(70.91,80.81)(0.12,-0.11){3}{\line(1,0){0.12}}
\multiput(70.59,81.17)(0.11,-0.12){3}{\line(0,-1){0.12}}
\multiput(70.29,81.55)(0.15,-0.19){2}{\line(0,-1){0.19}}
\multiput(70.03,81.96)(0.13,-0.21){2}{\line(0,-1){0.21}}
\multiput(69.8,82.4)(0.11,-0.22){2}{\line(0,-1){0.22}}
\multiput(69.61,82.84)(0.1,-0.22){2}{\line(0,-1){0.22}}
\multiput(69.46,83.31)(0.16,-0.46){1}{\line(0,-1){0.46}}
\multiput(69.34,83.78)(0.12,-0.48){1}{\line(0,-1){0.48}}
\multiput(69.26,84.27)(0.08,-0.48){1}{\line(0,-1){0.48}}
\multiput(69.22,84.75)(0.04,-0.49){1}{\line(0,-1){0.49}}
\put(69.22,84.75){\line(0,1){0.49}}
\multiput(69.22,85.25)(0.04,0.49){1}{\line(0,1){0.49}}
\multiput(69.26,85.73)(0.08,0.48){1}{\line(0,1){0.48}}
\multiput(69.34,86.22)(0.12,0.48){1}{\line(0,1){0.48}}
\multiput(69.46,86.69)(0.16,0.46){1}{\line(0,1){0.46}}
\multiput(69.61,87.16)(0.1,0.22){2}{\line(0,1){0.22}}
\multiput(69.8,87.6)(0.11,0.22){2}{\line(0,1){0.22}}
\multiput(70.03,88.04)(0.13,0.21){2}{\line(0,1){0.21}}
\multiput(70.29,88.45)(0.15,0.19){2}{\line(0,1){0.19}}
\multiput(70.59,88.83)(0.11,0.12){3}{\line(0,1){0.12}}
\multiput(70.91,89.19)(0.12,0.11){3}{\line(1,0){0.12}}
\multiput(71.26,89.53)(0.13,0.1){3}{\line(1,0){0.13}}
\multiput(71.64,89.83)(0.2,0.13){2}{\line(1,0){0.2}}
\multiput(72.04,90.09)(0.21,0.12){2}{\line(1,0){0.21}}
\multiput(72.46,90.33)(0.22,0.1){2}{\line(1,0){0.22}}
\multiput(72.9,90.52)(0.45,0.16){1}{\line(1,0){0.45}}
\multiput(73.35,90.68)(0.46,0.12){1}{\line(1,0){0.46}}
\multiput(73.81,90.8)(0.47,0.08){1}{\line(1,0){0.47}}
\multiput(74.28,90.88)(0.48,0.04){1}{\line(1,0){0.48}}
\put(74.76,90.92){\line(1,0){0.48}}
\multiput(75.24,90.92)(0.48,-0.04){1}{\line(1,0){0.48}}
\multiput(75.72,90.88)(0.47,-0.08){1}{\line(1,0){0.47}}
\multiput(76.19,90.8)(0.46,-0.12){1}{\line(1,0){0.46}}
\multiput(76.65,90.68)(0.45,-0.16){1}{\line(1,0){0.45}}
\multiput(77.1,90.52)(0.22,-0.1){2}{\line(1,0){0.22}}
\multiput(77.54,90.33)(0.21,-0.12){2}{\line(1,0){0.21}}
\multiput(77.96,90.09)(0.2,-0.13){2}{\line(1,0){0.2}}
\multiput(78.36,89.83)(0.13,-0.1){3}{\line(1,0){0.13}}
\multiput(78.74,89.53)(0.12,-0.11){3}{\line(1,0){0.12}}
\multiput(79.09,89.19)(0.11,-0.12){3}{\line(0,-1){0.12}}
\multiput(79.41,88.83)(0.15,-0.19){2}{\line(0,-1){0.19}}
\multiput(79.71,88.45)(0.13,-0.21){2}{\line(0,-1){0.21}}
\multiput(79.97,88.04)(0.11,-0.22){2}{\line(0,-1){0.22}}
\multiput(80.2,87.6)(0.1,-0.22){2}{\line(0,-1){0.22}}
\multiput(80.39,87.16)(0.16,-0.46){1}{\line(0,-1){0.46}}
\multiput(80.54,86.69)(0.12,-0.48){1}{\line(0,-1){0.48}}
\multiput(80.66,86.22)(0.08,-0.48){1}{\line(0,-1){0.48}}
\multiput(80.74,85.73)(0.04,-0.49){1}{\line(0,-1){0.49}}

\linethickness{0.3mm}
\put(190.78,84.75){\line(0,1){0.49}}
\multiput(190.74,84.27)(0.04,0.49){1}{\line(0,1){0.49}}
\multiput(190.66,83.78)(0.08,0.48){1}{\line(0,1){0.48}}
\multiput(190.54,83.31)(0.12,0.48){1}{\line(0,1){0.48}}
\multiput(190.39,82.84)(0.16,0.46){1}{\line(0,1){0.46}}
\multiput(190.2,82.4)(0.1,0.22){2}{\line(0,1){0.22}}
\multiput(189.97,81.96)(0.11,0.22){2}{\line(0,1){0.22}}
\multiput(189.71,81.55)(0.13,0.21){2}{\line(0,1){0.21}}
\multiput(189.41,81.17)(0.15,0.19){2}{\line(0,1){0.19}}
\multiput(189.09,80.81)(0.11,0.12){3}{\line(0,1){0.12}}
\multiput(188.74,80.47)(0.12,0.11){3}{\line(1,0){0.12}}
\multiput(188.36,80.17)(0.13,0.1){3}{\line(1,0){0.13}}
\multiput(187.96,79.91)(0.2,0.13){2}{\line(1,0){0.2}}
\multiput(187.54,79.67)(0.21,0.12){2}{\line(1,0){0.21}}
\multiput(187.1,79.48)(0.22,0.1){2}{\line(1,0){0.22}}
\multiput(186.65,79.32)(0.45,0.16){1}{\line(1,0){0.45}}
\multiput(186.19,79.2)(0.46,0.12){1}{\line(1,0){0.46}}
\multiput(185.72,79.12)(0.47,0.08){1}{\line(1,0){0.47}}
\multiput(185.24,79.08)(0.48,0.04){1}{\line(1,0){0.48}}
\put(184.76,79.08){\line(1,0){0.48}}
\multiput(184.28,79.12)(0.48,-0.04){1}{\line(1,0){0.48}}
\multiput(183.81,79.2)(0.47,-0.08){1}{\line(1,0){0.47}}
\multiput(183.35,79.32)(0.46,-0.12){1}{\line(1,0){0.46}}
\multiput(182.9,79.48)(0.45,-0.16){1}{\line(1,0){0.45}}
\multiput(182.46,79.67)(0.22,-0.1){2}{\line(1,0){0.22}}
\multiput(182.04,79.91)(0.21,-0.12){2}{\line(1,0){0.21}}
\multiput(181.64,80.17)(0.2,-0.13){2}{\line(1,0){0.2}}
\multiput(181.26,80.47)(0.13,-0.1){3}{\line(1,0){0.13}}
\multiput(180.91,80.81)(0.12,-0.11){3}{\line(1,0){0.12}}
\multiput(180.59,81.17)(0.11,-0.12){3}{\line(0,-1){0.12}}
\multiput(180.29,81.55)(0.15,-0.19){2}{\line(0,-1){0.19}}
\multiput(180.03,81.96)(0.13,-0.21){2}{\line(0,-1){0.21}}
\multiput(179.8,82.4)(0.11,-0.22){2}{\line(0,-1){0.22}}
\multiput(179.61,82.84)(0.1,-0.22){2}{\line(0,-1){0.22}}
\multiput(179.46,83.31)(0.16,-0.46){1}{\line(0,-1){0.46}}
\multiput(179.34,83.78)(0.12,-0.48){1}{\line(0,-1){0.48}}
\multiput(179.26,84.27)(0.08,-0.48){1}{\line(0,-1){0.48}}
\multiput(179.22,84.75)(0.04,-0.49){1}{\line(0,-1){0.49}}
\put(179.22,84.75){\line(0,1){0.49}}
\multiput(179.22,85.25)(0.04,0.49){1}{\line(0,1){0.49}}
\multiput(179.26,85.73)(0.08,0.48){1}{\line(0,1){0.48}}
\multiput(179.34,86.22)(0.12,0.48){1}{\line(0,1){0.48}}
\multiput(179.46,86.69)(0.16,0.46){1}{\line(0,1){0.46}}
\multiput(179.61,87.16)(0.1,0.22){2}{\line(0,1){0.22}}
\multiput(179.8,87.6)(0.11,0.22){2}{\line(0,1){0.22}}
\multiput(180.03,88.04)(0.13,0.21){2}{\line(0,1){0.21}}
\multiput(180.29,88.45)(0.15,0.19){2}{\line(0,1){0.19}}
\multiput(180.59,88.83)(0.11,0.12){3}{\line(0,1){0.12}}
\multiput(180.91,89.19)(0.12,0.11){3}{\line(1,0){0.12}}
\multiput(181.26,89.53)(0.13,0.1){3}{\line(1,0){0.13}}
\multiput(181.64,89.83)(0.2,0.13){2}{\line(1,0){0.2}}
\multiput(182.04,90.09)(0.21,0.12){2}{\line(1,0){0.21}}
\multiput(182.46,90.33)(0.22,0.1){2}{\line(1,0){0.22}}
\multiput(182.9,90.52)(0.45,0.16){1}{\line(1,0){0.45}}
\multiput(183.35,90.68)(0.46,0.12){1}{\line(1,0){0.46}}
\multiput(183.81,90.8)(0.47,0.08){1}{\line(1,0){0.47}}
\multiput(184.28,90.88)(0.48,0.04){1}{\line(1,0){0.48}}
\put(184.76,90.92){\line(1,0){0.48}}
\multiput(185.24,90.92)(0.48,-0.04){1}{\line(1,0){0.48}}
\multiput(185.72,90.88)(0.47,-0.08){1}{\line(1,0){0.47}}
\multiput(186.19,90.8)(0.46,-0.12){1}{\line(1,0){0.46}}
\multiput(186.65,90.68)(0.45,-0.16){1}{\line(1,0){0.45}}
\multiput(187.1,90.52)(0.22,-0.1){2}{\line(1,0){0.22}}
\multiput(187.54,90.33)(0.21,-0.12){2}{\line(1,0){0.21}}
\multiput(187.96,90.09)(0.2,-0.13){2}{\line(1,0){0.2}}
\multiput(188.36,89.83)(0.13,-0.1){3}{\line(1,0){0.13}}
\multiput(188.74,89.53)(0.12,-0.11){3}{\line(1,0){0.12}}
\multiput(189.09,89.19)(0.11,-0.12){3}{\line(0,-1){0.12}}
\multiput(189.41,88.83)(0.15,-0.19){2}{\line(0,-1){0.19}}
\multiput(189.71,88.45)(0.13,-0.21){2}{\line(0,-1){0.21}}
\multiput(189.97,88.04)(0.11,-0.22){2}{\line(0,-1){0.22}}
\multiput(190.2,87.6)(0.1,-0.22){2}{\line(0,-1){0.22}}
\multiput(190.39,87.16)(0.16,-0.46){1}{\line(0,-1){0.46}}
\multiput(190.54,86.69)(0.12,-0.48){1}{\line(0,-1){0.48}}
\multiput(190.66,86.22)(0.08,-0.48){1}{\line(0,-1){0.48}}
\multiput(190.74,85.73)(0.04,-0.49){1}{\line(0,-1){0.49}}

\linethickness{0.3mm}
\put(190.78,49.75){\line(0,1){0.49}}
\multiput(190.74,49.27)(0.04,0.49){1}{\line(0,1){0.49}}
\multiput(190.66,48.78)(0.08,0.48){1}{\line(0,1){0.48}}
\multiput(190.54,48.31)(0.12,0.48){1}{\line(0,1){0.48}}
\multiput(190.39,47.84)(0.16,0.46){1}{\line(0,1){0.46}}
\multiput(190.2,47.4)(0.1,0.22){2}{\line(0,1){0.22}}
\multiput(189.97,46.96)(0.11,0.22){2}{\line(0,1){0.22}}
\multiput(189.71,46.55)(0.13,0.21){2}{\line(0,1){0.21}}
\multiput(189.41,46.17)(0.15,0.19){2}{\line(0,1){0.19}}
\multiput(189.09,45.81)(0.11,0.12){3}{\line(0,1){0.12}}
\multiput(188.74,45.47)(0.12,0.11){3}{\line(1,0){0.12}}
\multiput(188.36,45.17)(0.13,0.1){3}{\line(1,0){0.13}}
\multiput(187.96,44.91)(0.2,0.13){2}{\line(1,0){0.2}}
\multiput(187.54,44.67)(0.21,0.12){2}{\line(1,0){0.21}}
\multiput(187.1,44.48)(0.22,0.1){2}{\line(1,0){0.22}}
\multiput(186.65,44.32)(0.45,0.16){1}{\line(1,0){0.45}}
\multiput(186.19,44.2)(0.46,0.12){1}{\line(1,0){0.46}}
\multiput(185.72,44.12)(0.47,0.08){1}{\line(1,0){0.47}}
\multiput(185.24,44.08)(0.48,0.04){1}{\line(1,0){0.48}}
\put(184.76,44.08){\line(1,0){0.48}}
\multiput(184.28,44.12)(0.48,-0.04){1}{\line(1,0){0.48}}
\multiput(183.81,44.2)(0.47,-0.08){1}{\line(1,0){0.47}}
\multiput(183.35,44.32)(0.46,-0.12){1}{\line(1,0){0.46}}
\multiput(182.9,44.48)(0.45,-0.16){1}{\line(1,0){0.45}}
\multiput(182.46,44.67)(0.22,-0.1){2}{\line(1,0){0.22}}
\multiput(182.04,44.91)(0.21,-0.12){2}{\line(1,0){0.21}}
\multiput(181.64,45.17)(0.2,-0.13){2}{\line(1,0){0.2}}
\multiput(181.26,45.47)(0.13,-0.1){3}{\line(1,0){0.13}}
\multiput(180.91,45.81)(0.12,-0.11){3}{\line(1,0){0.12}}
\multiput(180.59,46.17)(0.11,-0.12){3}{\line(0,-1){0.12}}
\multiput(180.29,46.55)(0.15,-0.19){2}{\line(0,-1){0.19}}
\multiput(180.03,46.96)(0.13,-0.21){2}{\line(0,-1){0.21}}
\multiput(179.8,47.4)(0.11,-0.22){2}{\line(0,-1){0.22}}
\multiput(179.61,47.84)(0.1,-0.22){2}{\line(0,-1){0.22}}
\multiput(179.46,48.31)(0.16,-0.46){1}{\line(0,-1){0.46}}
\multiput(179.34,48.78)(0.12,-0.48){1}{\line(0,-1){0.48}}
\multiput(179.26,49.27)(0.08,-0.48){1}{\line(0,-1){0.48}}
\multiput(179.22,49.75)(0.04,-0.49){1}{\line(0,-1){0.49}}
\put(179.22,49.75){\line(0,1){0.49}}
\multiput(179.22,50.25)(0.04,0.49){1}{\line(0,1){0.49}}
\multiput(179.26,50.73)(0.08,0.48){1}{\line(0,1){0.48}}
\multiput(179.34,51.22)(0.12,0.48){1}{\line(0,1){0.48}}
\multiput(179.46,51.69)(0.16,0.46){1}{\line(0,1){0.46}}
\multiput(179.61,52.16)(0.1,0.22){2}{\line(0,1){0.22}}
\multiput(179.8,52.6)(0.11,0.22){2}{\line(0,1){0.22}}
\multiput(180.03,53.04)(0.13,0.21){2}{\line(0,1){0.21}}
\multiput(180.29,53.45)(0.15,0.19){2}{\line(0,1){0.19}}
\multiput(180.59,53.83)(0.11,0.12){3}{\line(0,1){0.12}}
\multiput(180.91,54.19)(0.12,0.11){3}{\line(1,0){0.12}}
\multiput(181.26,54.53)(0.13,0.1){3}{\line(1,0){0.13}}
\multiput(181.64,54.83)(0.2,0.13){2}{\line(1,0){0.2}}
\multiput(182.04,55.09)(0.21,0.12){2}{\line(1,0){0.21}}
\multiput(182.46,55.33)(0.22,0.1){2}{\line(1,0){0.22}}
\multiput(182.9,55.52)(0.45,0.16){1}{\line(1,0){0.45}}
\multiput(183.35,55.68)(0.46,0.12){1}{\line(1,0){0.46}}
\multiput(183.81,55.8)(0.47,0.08){1}{\line(1,0){0.47}}
\multiput(184.28,55.88)(0.48,0.04){1}{\line(1,0){0.48}}
\put(184.76,55.92){\line(1,0){0.48}}
\multiput(185.24,55.92)(0.48,-0.04){1}{\line(1,0){0.48}}
\multiput(185.72,55.88)(0.47,-0.08){1}{\line(1,0){0.47}}
\multiput(186.19,55.8)(0.46,-0.12){1}{\line(1,0){0.46}}
\multiput(186.65,55.68)(0.45,-0.16){1}{\line(1,0){0.45}}
\multiput(187.1,55.52)(0.22,-0.1){2}{\line(1,0){0.22}}
\multiput(187.54,55.33)(0.21,-0.12){2}{\line(1,0){0.21}}
\multiput(187.96,55.09)(0.2,-0.13){2}{\line(1,0){0.2}}
\multiput(188.36,54.83)(0.13,-0.1){3}{\line(1,0){0.13}}
\multiput(188.74,54.53)(0.12,-0.11){3}{\line(1,0){0.12}}
\multiput(189.09,54.19)(0.11,-0.12){3}{\line(0,-1){0.12}}
\multiput(189.41,53.83)(0.15,-0.19){2}{\line(0,-1){0.19}}
\multiput(189.71,53.45)(0.13,-0.21){2}{\line(0,-1){0.21}}
\multiput(189.97,53.04)(0.11,-0.22){2}{\line(0,-1){0.22}}
\multiput(190.2,52.6)(0.1,-0.22){2}{\line(0,-1){0.22}}
\multiput(190.39,52.16)(0.16,-0.46){1}{\line(0,-1){0.46}}
\multiput(190.54,51.69)(0.12,-0.48){1}{\line(0,-1){0.48}}
\multiput(190.66,51.22)(0.08,-0.48){1}{\line(0,-1){0.48}}
\multiput(190.74,50.73)(0.04,-0.49){1}{\line(0,-1){0.49}}

\linethickness{0.3mm}
\put(80.78,49.75){\line(0,1){0.49}}
\multiput(80.74,49.27)(0.04,0.49){1}{\line(0,1){0.49}}
\multiput(80.66,48.78)(0.08,0.48){1}{\line(0,1){0.48}}
\multiput(80.54,48.31)(0.12,0.48){1}{\line(0,1){0.48}}
\multiput(80.39,47.84)(0.16,0.46){1}{\line(0,1){0.46}}
\multiput(80.2,47.4)(0.1,0.22){2}{\line(0,1){0.22}}
\multiput(79.97,46.96)(0.11,0.22){2}{\line(0,1){0.22}}
\multiput(79.71,46.55)(0.13,0.21){2}{\line(0,1){0.21}}
\multiput(79.41,46.17)(0.15,0.19){2}{\line(0,1){0.19}}
\multiput(79.09,45.81)(0.11,0.12){3}{\line(0,1){0.12}}
\multiput(78.74,45.47)(0.12,0.11){3}{\line(1,0){0.12}}
\multiput(78.36,45.17)(0.13,0.1){3}{\line(1,0){0.13}}
\multiput(77.96,44.91)(0.2,0.13){2}{\line(1,0){0.2}}
\multiput(77.54,44.67)(0.21,0.12){2}{\line(1,0){0.21}}
\multiput(77.1,44.48)(0.22,0.1){2}{\line(1,0){0.22}}
\multiput(76.65,44.32)(0.45,0.16){1}{\line(1,0){0.45}}
\multiput(76.19,44.2)(0.46,0.12){1}{\line(1,0){0.46}}
\multiput(75.72,44.12)(0.47,0.08){1}{\line(1,0){0.47}}
\multiput(75.24,44.08)(0.48,0.04){1}{\line(1,0){0.48}}
\put(74.76,44.08){\line(1,0){0.48}}
\multiput(74.28,44.12)(0.48,-0.04){1}{\line(1,0){0.48}}
\multiput(73.81,44.2)(0.47,-0.08){1}{\line(1,0){0.47}}
\multiput(73.35,44.32)(0.46,-0.12){1}{\line(1,0){0.46}}
\multiput(72.9,44.48)(0.45,-0.16){1}{\line(1,0){0.45}}
\multiput(72.46,44.67)(0.22,-0.1){2}{\line(1,0){0.22}}
\multiput(72.04,44.91)(0.21,-0.12){2}{\line(1,0){0.21}}
\multiput(71.64,45.17)(0.2,-0.13){2}{\line(1,0){0.2}}
\multiput(71.26,45.47)(0.13,-0.1){3}{\line(1,0){0.13}}
\multiput(70.91,45.81)(0.12,-0.11){3}{\line(1,0){0.12}}
\multiput(70.59,46.17)(0.11,-0.12){3}{\line(0,-1){0.12}}
\multiput(70.29,46.55)(0.15,-0.19){2}{\line(0,-1){0.19}}
\multiput(70.03,46.96)(0.13,-0.21){2}{\line(0,-1){0.21}}
\multiput(69.8,47.4)(0.11,-0.22){2}{\line(0,-1){0.22}}
\multiput(69.61,47.84)(0.1,-0.22){2}{\line(0,-1){0.22}}
\multiput(69.46,48.31)(0.16,-0.46){1}{\line(0,-1){0.46}}
\multiput(69.34,48.78)(0.12,-0.48){1}{\line(0,-1){0.48}}
\multiput(69.26,49.27)(0.08,-0.48){1}{\line(0,-1){0.48}}
\multiput(69.22,49.75)(0.04,-0.49){1}{\line(0,-1){0.49}}
\put(69.22,49.75){\line(0,1){0.49}}
\multiput(69.22,50.25)(0.04,0.49){1}{\line(0,1){0.49}}
\multiput(69.26,50.73)(0.08,0.48){1}{\line(0,1){0.48}}
\multiput(69.34,51.22)(0.12,0.48){1}{\line(0,1){0.48}}
\multiput(69.46,51.69)(0.16,0.46){1}{\line(0,1){0.46}}
\multiput(69.61,52.16)(0.1,0.22){2}{\line(0,1){0.22}}
\multiput(69.8,52.6)(0.11,0.22){2}{\line(0,1){0.22}}
\multiput(70.03,53.04)(0.13,0.21){2}{\line(0,1){0.21}}
\multiput(70.29,53.45)(0.15,0.19){2}{\line(0,1){0.19}}
\multiput(70.59,53.83)(0.11,0.12){3}{\line(0,1){0.12}}
\multiput(70.91,54.19)(0.12,0.11){3}{\line(1,0){0.12}}
\multiput(71.26,54.53)(0.13,0.1){3}{\line(1,0){0.13}}
\multiput(71.64,54.83)(0.2,0.13){2}{\line(1,0){0.2}}
\multiput(72.04,55.09)(0.21,0.12){2}{\line(1,0){0.21}}
\multiput(72.46,55.33)(0.22,0.1){2}{\line(1,0){0.22}}
\multiput(72.9,55.52)(0.45,0.16){1}{\line(1,0){0.45}}
\multiput(73.35,55.68)(0.46,0.12){1}{\line(1,0){0.46}}
\multiput(73.81,55.8)(0.47,0.08){1}{\line(1,0){0.47}}
\multiput(74.28,55.88)(0.48,0.04){1}{\line(1,0){0.48}}
\put(74.76,55.92){\line(1,0){0.48}}
\multiput(75.24,55.92)(0.48,-0.04){1}{\line(1,0){0.48}}
\multiput(75.72,55.88)(0.47,-0.08){1}{\line(1,0){0.47}}
\multiput(76.19,55.8)(0.46,-0.12){1}{\line(1,0){0.46}}
\multiput(76.65,55.68)(0.45,-0.16){1}{\line(1,0){0.45}}
\multiput(77.1,55.52)(0.22,-0.1){2}{\line(1,0){0.22}}
\multiput(77.54,55.33)(0.21,-0.12){2}{\line(1,0){0.21}}
\multiput(77.96,55.09)(0.2,-0.13){2}{\line(1,0){0.2}}
\multiput(78.36,54.83)(0.13,-0.1){3}{\line(1,0){0.13}}
\multiput(78.74,54.53)(0.12,-0.11){3}{\line(1,0){0.12}}
\multiput(79.09,54.19)(0.11,-0.12){3}{\line(0,-1){0.12}}
\multiput(79.41,53.83)(0.15,-0.19){2}{\line(0,-1){0.19}}
\multiput(79.71,53.45)(0.13,-0.21){2}{\line(0,-1){0.21}}
\multiput(79.97,53.04)(0.11,-0.22){2}{\line(0,-1){0.22}}
\multiput(80.2,52.6)(0.1,-0.22){2}{\line(0,-1){0.22}}
\multiput(80.39,52.16)(0.16,-0.46){1}{\line(0,-1){0.46}}
\multiput(80.54,51.69)(0.12,-0.48){1}{\line(0,-1){0.48}}
\multiput(80.66,51.22)(0.08,-0.48){1}{\line(0,-1){0.48}}
\multiput(80.74,50.73)(0.04,-0.49){1}{\line(0,-1){0.49}}

\linethickness{0.3mm}
\put(190.78,14.75){\line(0,1){0.49}}
\multiput(190.74,14.27)(0.04,0.49){1}{\line(0,1){0.49}}
\multiput(190.66,13.78)(0.08,0.48){1}{\line(0,1){0.48}}
\multiput(190.54,13.31)(0.12,0.48){1}{\line(0,1){0.48}}
\multiput(190.39,12.84)(0.16,0.46){1}{\line(0,1){0.46}}
\multiput(190.2,12.4)(0.1,0.22){2}{\line(0,1){0.22}}
\multiput(189.97,11.96)(0.11,0.22){2}{\line(0,1){0.22}}
\multiput(189.71,11.55)(0.13,0.21){2}{\line(0,1){0.21}}
\multiput(189.41,11.17)(0.15,0.19){2}{\line(0,1){0.19}}
\multiput(189.09,10.81)(0.11,0.12){3}{\line(0,1){0.12}}
\multiput(188.74,10.47)(0.12,0.11){3}{\line(1,0){0.12}}
\multiput(188.36,10.17)(0.13,0.1){3}{\line(1,0){0.13}}
\multiput(187.96,9.91)(0.2,0.13){2}{\line(1,0){0.2}}
\multiput(187.54,9.67)(0.21,0.12){2}{\line(1,0){0.21}}
\multiput(187.1,9.48)(0.22,0.1){2}{\line(1,0){0.22}}
\multiput(186.65,9.32)(0.45,0.16){1}{\line(1,0){0.45}}
\multiput(186.19,9.2)(0.46,0.12){1}{\line(1,0){0.46}}
\multiput(185.72,9.12)(0.47,0.08){1}{\line(1,0){0.47}}
\multiput(185.24,9.08)(0.48,0.04){1}{\line(1,0){0.48}}
\put(184.76,9.08){\line(1,0){0.48}}
\multiput(184.28,9.12)(0.48,-0.04){1}{\line(1,0){0.48}}
\multiput(183.81,9.2)(0.47,-0.08){1}{\line(1,0){0.47}}
\multiput(183.35,9.32)(0.46,-0.12){1}{\line(1,0){0.46}}
\multiput(182.9,9.48)(0.45,-0.16){1}{\line(1,0){0.45}}
\multiput(182.46,9.67)(0.22,-0.1){2}{\line(1,0){0.22}}
\multiput(182.04,9.91)(0.21,-0.12){2}{\line(1,0){0.21}}
\multiput(181.64,10.17)(0.2,-0.13){2}{\line(1,0){0.2}}
\multiput(181.26,10.47)(0.13,-0.1){3}{\line(1,0){0.13}}
\multiput(180.91,10.81)(0.12,-0.11){3}{\line(1,0){0.12}}
\multiput(180.59,11.17)(0.11,-0.12){3}{\line(0,-1){0.12}}
\multiput(180.29,11.55)(0.15,-0.19){2}{\line(0,-1){0.19}}
\multiput(180.03,11.96)(0.13,-0.21){2}{\line(0,-1){0.21}}
\multiput(179.8,12.4)(0.11,-0.22){2}{\line(0,-1){0.22}}
\multiput(179.61,12.84)(0.1,-0.22){2}{\line(0,-1){0.22}}
\multiput(179.46,13.31)(0.16,-0.46){1}{\line(0,-1){0.46}}
\multiput(179.34,13.78)(0.12,-0.48){1}{\line(0,-1){0.48}}
\multiput(179.26,14.27)(0.08,-0.48){1}{\line(0,-1){0.48}}
\multiput(179.22,14.75)(0.04,-0.49){1}{\line(0,-1){0.49}}
\put(179.22,14.75){\line(0,1){0.49}}
\multiput(179.22,15.25)(0.04,0.49){1}{\line(0,1){0.49}}
\multiput(179.26,15.73)(0.08,0.48){1}{\line(0,1){0.48}}
\multiput(179.34,16.22)(0.12,0.48){1}{\line(0,1){0.48}}
\multiput(179.46,16.69)(0.16,0.46){1}{\line(0,1){0.46}}
\multiput(179.61,17.16)(0.1,0.22){2}{\line(0,1){0.22}}
\multiput(179.8,17.6)(0.11,0.22){2}{\line(0,1){0.22}}
\multiput(180.03,18.04)(0.13,0.21){2}{\line(0,1){0.21}}
\multiput(180.29,18.45)(0.15,0.19){2}{\line(0,1){0.19}}
\multiput(180.59,18.83)(0.11,0.12){3}{\line(0,1){0.12}}
\multiput(180.91,19.19)(0.12,0.11){3}{\line(1,0){0.12}}
\multiput(181.26,19.53)(0.13,0.1){3}{\line(1,0){0.13}}
\multiput(181.64,19.83)(0.2,0.13){2}{\line(1,0){0.2}}
\multiput(182.04,20.09)(0.21,0.12){2}{\line(1,0){0.21}}
\multiput(182.46,20.33)(0.22,0.1){2}{\line(1,0){0.22}}
\multiput(182.9,20.52)(0.45,0.16){1}{\line(1,0){0.45}}
\multiput(183.35,20.68)(0.46,0.12){1}{\line(1,0){0.46}}
\multiput(183.81,20.8)(0.47,0.08){1}{\line(1,0){0.47}}
\multiput(184.28,20.88)(0.48,0.04){1}{\line(1,0){0.48}}
\put(184.76,20.92){\line(1,0){0.48}}
\multiput(185.24,20.92)(0.48,-0.04){1}{\line(1,0){0.48}}
\multiput(185.72,20.88)(0.47,-0.08){1}{\line(1,0){0.47}}
\multiput(186.19,20.8)(0.46,-0.12){1}{\line(1,0){0.46}}
\multiput(186.65,20.68)(0.45,-0.16){1}{\line(1,0){0.45}}
\multiput(187.1,20.52)(0.22,-0.1){2}{\line(1,0){0.22}}
\multiput(187.54,20.33)(0.21,-0.12){2}{\line(1,0){0.21}}
\multiput(187.96,20.09)(0.2,-0.13){2}{\line(1,0){0.2}}
\multiput(188.36,19.83)(0.13,-0.1){3}{\line(1,0){0.13}}
\multiput(188.74,19.53)(0.12,-0.11){3}{\line(1,0){0.12}}
\multiput(189.09,19.19)(0.11,-0.12){3}{\line(0,-1){0.12}}
\multiput(189.41,18.83)(0.15,-0.19){2}{\line(0,-1){0.19}}
\multiput(189.71,18.45)(0.13,-0.21){2}{\line(0,-1){0.21}}
\multiput(189.97,18.04)(0.11,-0.22){2}{\line(0,-1){0.22}}
\multiput(190.2,17.6)(0.1,-0.22){2}{\line(0,-1){0.22}}
\multiput(190.39,17.16)(0.16,-0.46){1}{\line(0,-1){0.46}}
\multiput(190.54,16.69)(0.12,-0.48){1}{\line(0,-1){0.48}}
\multiput(190.66,16.22)(0.08,-0.48){1}{\line(0,-1){0.48}}
\multiput(190.74,15.73)(0.04,-0.49){1}{\line(0,-1){0.49}}

\linethickness{0.3mm}
\put(80.78,14.75){\line(0,1){0.49}}
\multiput(80.74,14.27)(0.04,0.49){1}{\line(0,1){0.49}}
\multiput(80.66,13.78)(0.08,0.48){1}{\line(0,1){0.48}}
\multiput(80.54,13.31)(0.12,0.48){1}{\line(0,1){0.48}}
\multiput(80.39,12.84)(0.16,0.46){1}{\line(0,1){0.46}}
\multiput(80.2,12.4)(0.1,0.22){2}{\line(0,1){0.22}}
\multiput(79.97,11.96)(0.11,0.22){2}{\line(0,1){0.22}}
\multiput(79.71,11.55)(0.13,0.21){2}{\line(0,1){0.21}}
\multiput(79.41,11.17)(0.15,0.19){2}{\line(0,1){0.19}}
\multiput(79.09,10.81)(0.11,0.12){3}{\line(0,1){0.12}}
\multiput(78.74,10.47)(0.12,0.11){3}{\line(1,0){0.12}}
\multiput(78.36,10.17)(0.13,0.1){3}{\line(1,0){0.13}}
\multiput(77.96,9.91)(0.2,0.13){2}{\line(1,0){0.2}}
\multiput(77.54,9.67)(0.21,0.12){2}{\line(1,0){0.21}}
\multiput(77.1,9.48)(0.22,0.1){2}{\line(1,0){0.22}}
\multiput(76.65,9.32)(0.45,0.16){1}{\line(1,0){0.45}}
\multiput(76.19,9.2)(0.46,0.12){1}{\line(1,0){0.46}}
\multiput(75.72,9.12)(0.47,0.08){1}{\line(1,0){0.47}}
\multiput(75.24,9.08)(0.48,0.04){1}{\line(1,0){0.48}}
\put(74.76,9.08){\line(1,0){0.48}}
\multiput(74.28,9.12)(0.48,-0.04){1}{\line(1,0){0.48}}
\multiput(73.81,9.2)(0.47,-0.08){1}{\line(1,0){0.47}}
\multiput(73.35,9.32)(0.46,-0.12){1}{\line(1,0){0.46}}
\multiput(72.9,9.48)(0.45,-0.16){1}{\line(1,0){0.45}}
\multiput(72.46,9.67)(0.22,-0.1){2}{\line(1,0){0.22}}
\multiput(72.04,9.91)(0.21,-0.12){2}{\line(1,0){0.21}}
\multiput(71.64,10.17)(0.2,-0.13){2}{\line(1,0){0.2}}
\multiput(71.26,10.47)(0.13,-0.1){3}{\line(1,0){0.13}}
\multiput(70.91,10.81)(0.12,-0.11){3}{\line(1,0){0.12}}
\multiput(70.59,11.17)(0.11,-0.12){3}{\line(0,-1){0.12}}
\multiput(70.29,11.55)(0.15,-0.19){2}{\line(0,-1){0.19}}
\multiput(70.03,11.96)(0.13,-0.21){2}{\line(0,-1){0.21}}
\multiput(69.8,12.4)(0.11,-0.22){2}{\line(0,-1){0.22}}
\multiput(69.61,12.84)(0.1,-0.22){2}{\line(0,-1){0.22}}
\multiput(69.46,13.31)(0.16,-0.46){1}{\line(0,-1){0.46}}
\multiput(69.34,13.78)(0.12,-0.48){1}{\line(0,-1){0.48}}
\multiput(69.26,14.27)(0.08,-0.48){1}{\line(0,-1){0.48}}
\multiput(69.22,14.75)(0.04,-0.49){1}{\line(0,-1){0.49}}
\put(69.22,14.75){\line(0,1){0.49}}
\multiput(69.22,15.25)(0.04,0.49){1}{\line(0,1){0.49}}
\multiput(69.26,15.73)(0.08,0.48){1}{\line(0,1){0.48}}
\multiput(69.34,16.22)(0.12,0.48){1}{\line(0,1){0.48}}
\multiput(69.46,16.69)(0.16,0.46){1}{\line(0,1){0.46}}
\multiput(69.61,17.16)(0.1,0.22){2}{\line(0,1){0.22}}
\multiput(69.8,17.6)(0.11,0.22){2}{\line(0,1){0.22}}
\multiput(70.03,18.04)(0.13,0.21){2}{\line(0,1){0.21}}
\multiput(70.29,18.45)(0.15,0.19){2}{\line(0,1){0.19}}
\multiput(70.59,18.83)(0.11,0.12){3}{\line(0,1){0.12}}
\multiput(70.91,19.19)(0.12,0.11){3}{\line(1,0){0.12}}
\multiput(71.26,19.53)(0.13,0.1){3}{\line(1,0){0.13}}
\multiput(71.64,19.83)(0.2,0.13){2}{\line(1,0){0.2}}
\multiput(72.04,20.09)(0.21,0.12){2}{\line(1,0){0.21}}
\multiput(72.46,20.33)(0.22,0.1){2}{\line(1,0){0.22}}
\multiput(72.9,20.52)(0.45,0.16){1}{\line(1,0){0.45}}
\multiput(73.35,20.68)(0.46,0.12){1}{\line(1,0){0.46}}
\multiput(73.81,20.8)(0.47,0.08){1}{\line(1,0){0.47}}
\multiput(74.28,20.88)(0.48,0.04){1}{\line(1,0){0.48}}
\put(74.76,20.92){\line(1,0){0.48}}
\multiput(75.24,20.92)(0.48,-0.04){1}{\line(1,0){0.48}}
\multiput(75.72,20.88)(0.47,-0.08){1}{\line(1,0){0.47}}
\multiput(76.19,20.8)(0.46,-0.12){1}{\line(1,0){0.46}}
\multiput(76.65,20.68)(0.45,-0.16){1}{\line(1,0){0.45}}
\multiput(77.1,20.52)(0.22,-0.1){2}{\line(1,0){0.22}}
\multiput(77.54,20.33)(0.21,-0.12){2}{\line(1,0){0.21}}
\multiput(77.96,20.09)(0.2,-0.13){2}{\line(1,0){0.2}}
\multiput(78.36,19.83)(0.13,-0.1){3}{\line(1,0){0.13}}
\multiput(78.74,19.53)(0.12,-0.11){3}{\line(1,0){0.12}}
\multiput(79.09,19.19)(0.11,-0.12){3}{\line(0,-1){0.12}}
\multiput(79.41,18.83)(0.15,-0.19){2}{\line(0,-1){0.19}}
\multiput(79.71,18.45)(0.13,-0.21){2}{\line(0,-1){0.21}}
\multiput(79.97,18.04)(0.11,-0.22){2}{\line(0,-1){0.22}}
\multiput(80.2,17.6)(0.1,-0.22){2}{\line(0,-1){0.22}}
\multiput(80.39,17.16)(0.16,-0.46){1}{\line(0,-1){0.46}}
\multiput(80.54,16.69)(0.12,-0.48){1}{\line(0,-1){0.48}}
\multiput(80.66,16.22)(0.08,-0.48){1}{\line(0,-1){0.48}}
\multiput(80.74,15.73)(0.04,-0.49){1}{\line(0,-1){0.49}}

\linethickness{0.3mm}
\put(165,50){\circle{10}}

\put(55,85){\makebox(0,0)[cc]{5}}

\put(75,85){\makebox(0,0)[cc]{10}}

\put(35,85){\makebox(0,0)[cc]{8}}

\put(15,85){\makebox(0,0)[cc]{4}}

\linethickness{0.3mm}
\put(10,90){\line(1,0){10}}
\put(10,80){\line(0,1){10}}
\put(20,80){\line(0,1){10}}
\put(10,80){\line(1,0){10}}
\put(145,85){\makebox(0,0)[cc]{8}}

\put(165,85){\makebox(0,0)[cc]{5}}

\put(185,85){\makebox(0,0)[cc]{10}}

\linethickness{0.3mm}
\put(120,90){\line(1,0){10}}
\put(120,80){\line(0,1){10}}
\put(130,80){\line(0,1){10}}
\put(120,80){\line(1,0){10}}
\put(145,50){\makebox(0,0)[cc]{8}}

\put(165,50){\makebox(0,0)[cc]{5}}

\put(185,50){\makebox(0,0)[cc]{6}}

\put(125,50){\makebox(0,0)[cc]{8}}

\put(15,50){\makebox(0,0)[cc]{8}}

\put(35,50){\makebox(0,0)[cc]{7}}

\put(55,50){\makebox(0,0)[cc]{5}}

\put(75,50){\makebox(0,0)[cc]{0}}

\linethickness{0.3mm}
\put(140,55){\line(1,0){10}}
\put(140,45){\line(0,1){10}}
\put(150,45){\line(0,1){10}}
\put(140,45){\line(1,0){10}}
\linethickness{0.3mm}
\put(10,55){\line(1,0){10}}
\put(10,45){\line(0,1){10}}
\put(20,45){\line(0,1){10}}
\put(10,45){\line(1,0){10}}
\put(35,15){\makebox(0,0)[cc]{7}}

\put(55,15){\makebox(0,0)[cc]{5}}

\put(75,15){\makebox(0,0)[cc]{0}}

\put(125,15){\makebox(0,0)[cc]{7}}

\put(165,15){\makebox(0,0)[cc]{5}}

\put(185,15){\makebox(0,0)[cc]{9}}

\linethickness{0.3mm}
\put(30,20){\line(1,0){10}}
\put(30,10){\line(0,1){10}}
\put(40,10){\line(0,1){10}}
\put(30,10){\line(1,0){10}}
\put(90,65){\makebox(0,0)[cc]{}}

\put(15,15){\makebox(0,0)[cc]{7}}

\put(145,15){\makebox(0,0)[cc]{6}}

\put(15,75){\makebox(0,0)[cc]{$leftEndUp$}}

\put(125,85){\makebox(0,0)[cc]{9}}

\put(125,75){\makebox(0,0)[cc]{$leftEndDown$}}

\put(145,40){\makebox(0,0)[cc]{$middleLeftDown$}}

\put(15,40){\makebox(0,0)[cc]{$leftEndDown$}}

\put(35,5){\makebox(0,0)[cc]{$middleLeftDown$}}

\put(5,85){\makebox(0,0)[cc]{$s_0$}}

\put(115,85){\makebox(0,0)[cc]{$s_1$}}

\put(115,50){\makebox(0,0)[cc]{$s_2$}}

\put(5,50){\makebox(0,0)[cc]{$s_3$}}

\put(5,15){\makebox(0,0)[cc]{$s_4$}}

\put(115,15){\makebox(0,0)[cc]{$s_5$}}

\end{picture}

%% file: Exemple3.tex
\ifx\JPicScale\undefined\def\JPicScale{0.75}\fi
\unitlength \JPicScale mm
\begin{picture}(190,90)(0,0)
\linethickness{0.3mm}
\put(15,85){\circle{10}}

\linethickness{0.3mm}
\put(35,85){\circle{10}}

\linethickness{0.3mm}
\put(55,85){\circle{10}}

\linethickness{0.3mm}
\put(75,85){\circle{10}}

\linethickness{0.3mm}
\put(145,85){\circle{10}}

\linethickness{0.3mm}
\put(125,85){\circle{10}}

\linethickness{0.3mm}
\put(20,85){\line(1,0){10}}
\linethickness{0.3mm}
\put(40,85){\line(1,0){10}}
\linethickness{0.3mm}
\put(60,85){\line(1,0){10}}
\linethickness{0.3mm}
\put(130,85){\line(1,0){10}}
\linethickness{0.3mm}
\put(150,85){\line(1,0){10}}
\linethickness{0.3mm}
\put(170,85){\line(1,0){10}}
\linethickness{0.3mm}
\put(185,50){\circle{10}}

\linethickness{0.3mm}
\put(165,50){\circle{10}}

\linethickness{0.3mm}
\put(145,50){\circle{10}}

\linethickness{0.3mm}
\put(125,50){\circle{10}}

\linethickness{0.3mm}
\put(130,50){\line(1,0){10}}
\linethickness{0.3mm}
\put(150,50){\line(1,0){10}}
\linethickness{0.3mm}
\put(170,50){\line(1,0){10}}
\linethickness{0.3mm}
\put(185,15){\circle{10}}

\linethickness{0.3mm}
\put(165,15){\circle{10}}

\linethickness{0.3mm}
\put(145,15){\circle{10}}

\linethickness{0.3mm}
\put(125,15){\circle{10}}

\linethickness{0.3mm}
\put(130,15){\line(1,0){10}}
\linethickness{0.3mm}
\put(150,15){\line(1,0){10}}
\linethickness{0.3mm}
\put(170,15){\line(1,0){10}}
\linethickness{0.3mm}
\put(75,50){\circle{10}}

\linethickness{0.3mm}
\put(55,50){\circle{10}}

\linethickness{0.3mm}
\put(35,50){\circle{10}}

\linethickness{0.3mm}
\put(15,50){\circle{10}}

\linethickness{0.3mm}
\put(20,50){\line(1,0){10}}
\linethickness{0.3mm}
\put(40,50){\line(1,0){10}}
\linethickness{0.3mm}
\put(60,50){\line(1,0){10}}
\linethickness{0.3mm}
\put(75,15){\circle{10}}

\linethickness{0.3mm}
\put(55,15){\circle{10}}

\linethickness{0.3mm}
\put(15,15){\circle{10}}

\linethickness{0.3mm}
\put(20,15){\line(1,0){10}}
\linethickness{0.3mm}
\put(40,15){\line(1,0){10}}
\linethickness{0.3mm}
\put(60,15){\line(1,0){10}}
\linethickness{0.3mm}
\put(85,85){\line(1,0){20}}
\put(105,85){\vector(1,0){0.12}}
\linethickness{0.3mm}
\put(85,50){\line(1,0){20}}
\put(85,50){\vector(-1,0){0.12}}
\linethickness{0.3mm}
\put(155,65){\line(0,1){5}}
\put(155,65){\vector(0,-1){0.12}}
\linethickness{0.3mm}
\put(45,30){\line(0,1){10}}
\put(45,30){\vector(0,-1){0.12}}
\linethickness{0.3mm}
\put(85,15){\line(1,0){20}}
\put(105,15){\vector(1,0){0.12}}
\put(5,85){\makebox(0,0)[cc]{$s_1$}}

\put(115,85){\makebox(0,0)[cc]{$s_2$}}

\put(115,50){\makebox(0,0)[cc]{$s_3$}}

\put(5,50){\makebox(0,0)[cc]{$s_4$}}

\put(5,15){\makebox(0,0)[cc]{$s_5$}}

\put(115,15){\makebox(0,0)[cc]{$s_6$}}

\linethickness{0.3mm}
\put(35,15){\circle{10}}

\put(15,50){\makebox(0,0)[cc]{}}

\linethickness{0.3mm}
\put(165,85){\circle{10}}

\put(115,30){\makebox(0,0)[cc]{}}

\linethickness{0.3mm}
\multiput(74.53,90.16)(0.47,-0.16){1}{\line(1,0){0.47}}
\multiput(74.05,90.32)(0.48,-0.16){1}{\line(1,0){0.48}}
\multiput(73.57,90.48)(0.48,-0.16){1}{\line(1,0){0.48}}
\multiput(73.1,90.63)(0.48,-0.15){1}{\line(1,0){0.48}}
\multiput(72.62,90.78)(0.48,-0.15){1}{\line(1,0){0.48}}
\multiput(72.14,90.93)(0.48,-0.15){1}{\line(1,0){0.48}}
\multiput(71.66,91.07)(0.48,-0.15){1}{\line(1,0){0.48}}
\multiput(71.18,91.22)(0.48,-0.14){1}{\line(1,0){0.48}}
\multiput(70.7,91.36)(0.48,-0.14){1}{\line(1,0){0.48}}
\multiput(70.22,91.5)(0.48,-0.14){1}{\line(1,0){0.48}}
\multiput(69.74,91.63)(0.48,-0.14){1}{\line(1,0){0.48}}
\multiput(69.26,91.76)(0.48,-0.13){1}{\line(1,0){0.48}}
\multiput(68.77,91.89)(0.48,-0.13){1}{\line(1,0){0.48}}
\multiput(68.29,92.02)(0.48,-0.13){1}{\line(1,0){0.48}}
\multiput(67.8,92.15)(0.49,-0.12){1}{\line(1,0){0.49}}
\multiput(67.32,92.27)(0.49,-0.12){1}{\line(1,0){0.49}}
\multiput(66.83,92.39)(0.49,-0.12){1}{\line(1,0){0.49}}
\multiput(66.34,92.5)(0.49,-0.12){1}{\line(1,0){0.49}}
\multiput(65.86,92.62)(0.49,-0.11){1}{\line(1,0){0.49}}
\multiput(65.37,92.73)(0.49,-0.11){1}{\line(1,0){0.49}}
\multiput(64.88,92.84)(0.49,-0.11){1}{\line(1,0){0.49}}
\multiput(64.39,92.95)(0.49,-0.11){1}{\line(1,0){0.49}}
\multiput(63.9,93.05)(0.49,-0.1){1}{\line(1,0){0.49}}
\multiput(63.41,93.15)(0.49,-0.1){1}{\line(1,0){0.49}}
\multiput(62.92,93.25)(0.49,-0.1){1}{\line(1,0){0.49}}
\multiput(62.43,93.34)(0.49,-0.1){1}{\line(1,0){0.49}}
\multiput(61.93,93.44)(0.49,-0.09){1}{\line(1,0){0.49}}
\multiput(61.44,93.53)(0.49,-0.09){1}{\line(1,0){0.49}}
\multiput(60.95,93.61)(0.49,-0.09){1}{\line(1,0){0.49}}
\multiput(60.45,93.7)(0.49,-0.09){1}{\line(1,0){0.49}}
\multiput(59.96,93.78)(0.49,-0.08){1}{\line(1,0){0.49}}
\multiput(59.47,93.86)(0.49,-0.08){1}{\line(1,0){0.49}}
\multiput(58.97,93.94)(0.49,-0.08){1}{\line(1,0){0.49}}
\multiput(58.48,94.01)(0.5,-0.07){1}{\line(1,0){0.5}}
\multiput(57.98,94.08)(0.5,-0.07){1}{\line(1,0){0.5}}
\multiput(57.48,94.15)(0.5,-0.07){1}{\line(1,0){0.5}}
\multiput(56.99,94.22)(0.5,-0.07){1}{\line(1,0){0.5}}
\multiput(56.49,94.28)(0.5,-0.06){1}{\line(1,0){0.5}}
\multiput(55.99,94.34)(0.5,-0.06){1}{\line(1,0){0.5}}
\multiput(55.5,94.4)(0.5,-0.06){1}{\line(1,0){0.5}}
\multiput(55,94.46)(0.5,-0.06){1}{\line(1,0){0.5}}
\multiput(54.5,94.51)(0.5,-0.05){1}{\line(1,0){0.5}}
\multiput(54,94.56)(0.5,-0.05){1}{\line(1,0){0.5}}
\multiput(53.5,94.61)(0.5,-0.05){1}{\line(1,0){0.5}}
\multiput(53,94.65)(0.5,-0.04){1}{\line(1,0){0.5}}
\multiput(52.5,94.7)(0.5,-0.04){1}{\line(1,0){0.5}}
\multiput(52.01,94.73)(0.5,-0.04){1}{\line(1,0){0.5}}
\multiput(51.51,94.77)(0.5,-0.04){1}{\line(1,0){0.5}}
\multiput(51.01,94.8)(0.5,-0.03){1}{\line(1,0){0.5}}
\multiput(50.51,94.84)(0.5,-0.03){1}{\line(1,0){0.5}}
\multiput(50.01,94.86)(0.5,-0.03){1}{\line(1,0){0.5}}
\multiput(49.51,94.89)(0.5,-0.03){1}{\line(1,0){0.5}}
\multiput(49.01,94.91)(0.5,-0.02){1}{\line(1,0){0.5}}
\multiput(48.51,94.93)(0.5,-0.02){1}{\line(1,0){0.5}}
\multiput(48,94.95)(0.5,-0.02){1}{\line(1,0){0.5}}
\multiput(47.5,94.97)(0.5,-0.01){1}{\line(1,0){0.5}}
\multiput(47,94.98)(0.5,-0.01){1}{\line(1,0){0.5}}
\multiput(46.5,94.99)(0.5,-0.01){1}{\line(1,0){0.5}}
\multiput(46,94.99)(0.5,-0.01){1}{\line(1,0){0.5}}
\multiput(45.5,95)(0.5,-0){1}{\line(1,0){0.5}}
\multiput(45,95)(0.5,-0){1}{\line(1,0){0.5}}
\multiput(44.5,95)(0.5,0){1}{\line(1,0){0.5}}
\multiput(44,94.99)(0.5,0){1}{\line(1,0){0.5}}
\multiput(43.5,94.99)(0.5,0.01){1}{\line(1,0){0.5}}
\multiput(43,94.98)(0.5,0.01){1}{\line(1,0){0.5}}
\multiput(42.5,94.97)(0.5,0.01){1}{\line(1,0){0.5}}
\multiput(42,94.95)(0.5,0.01){1}{\line(1,0){0.5}}
\multiput(41.49,94.93)(0.5,0.02){1}{\line(1,0){0.5}}
\multiput(40.99,94.91)(0.5,0.02){1}{\line(1,0){0.5}}
\multiput(40.49,94.89)(0.5,0.02){1}{\line(1,0){0.5}}
\multiput(39.99,94.86)(0.5,0.03){1}{\line(1,0){0.5}}
\multiput(39.49,94.84)(0.5,0.03){1}{\line(1,0){0.5}}
\multiput(38.99,94.8)(0.5,0.03){1}{\line(1,0){0.5}}
\multiput(38.49,94.77)(0.5,0.03){1}{\line(1,0){0.5}}
\multiput(37.99,94.73)(0.5,0.04){1}{\line(1,0){0.5}}
\multiput(37.5,94.7)(0.5,0.04){1}{\line(1,0){0.5}}
\multiput(37,94.65)(0.5,0.04){1}{\line(1,0){0.5}}
\multiput(36.5,94.61)(0.5,0.04){1}{\line(1,0){0.5}}
\multiput(36,94.56)(0.5,0.05){1}{\line(1,0){0.5}}
\multiput(35.5,94.51)(0.5,0.05){1}{\line(1,0){0.5}}
\multiput(35,94.46)(0.5,0.05){1}{\line(1,0){0.5}}
\multiput(34.5,94.4)(0.5,0.06){1}{\line(1,0){0.5}}
\multiput(34.01,94.34)(0.5,0.06){1}{\line(1,0){0.5}}
\multiput(33.51,94.28)(0.5,0.06){1}{\line(1,0){0.5}}
\multiput(33.01,94.22)(0.5,0.06){1}{\line(1,0){0.5}}
\multiput(32.52,94.15)(0.5,0.07){1}{\line(1,0){0.5}}
\multiput(32.02,94.08)(0.5,0.07){1}{\line(1,0){0.5}}
\multiput(31.52,94.01)(0.5,0.07){1}{\line(1,0){0.5}}
\multiput(31.03,93.94)(0.5,0.07){1}{\line(1,0){0.5}}
\multiput(30.53,93.86)(0.49,0.08){1}{\line(1,0){0.49}}
\multiput(30.04,93.78)(0.49,0.08){1}{\line(1,0){0.49}}
\multiput(29.55,93.7)(0.49,0.08){1}{\line(1,0){0.49}}
\multiput(29.05,93.61)(0.49,0.09){1}{\line(1,0){0.49}}
\multiput(28.56,93.53)(0.49,0.09){1}{\line(1,0){0.49}}
\multiput(28.07,93.44)(0.49,0.09){1}{\line(1,0){0.49}}
\multiput(27.57,93.34)(0.49,0.09){1}{\line(1,0){0.49}}
\multiput(27.08,93.25)(0.49,0.1){1}{\line(1,0){0.49}}
\multiput(26.59,93.15)(0.49,0.1){1}{\line(1,0){0.49}}
\multiput(26.1,93.05)(0.49,0.1){1}{\line(1,0){0.49}}
\multiput(25.61,92.95)(0.49,0.1){1}{\line(1,0){0.49}}
\multiput(25.12,92.84)(0.49,0.11){1}{\line(1,0){0.49}}
\multiput(24.63,92.73)(0.49,0.11){1}{\line(1,0){0.49}}
\multiput(24.14,92.62)(0.49,0.11){1}{\line(1,0){0.49}}
\multiput(23.66,92.5)(0.49,0.11){1}{\line(1,0){0.49}}
\multiput(23.17,92.39)(0.49,0.12){1}{\line(1,0){0.49}}
\multiput(22.68,92.27)(0.49,0.12){1}{\line(1,0){0.49}}
\multiput(22.2,92.15)(0.49,0.12){1}{\line(1,0){0.49}}
\multiput(21.71,92.02)(0.49,0.12){1}{\line(1,0){0.49}}
\multiput(21.23,91.89)(0.48,0.13){1}{\line(1,0){0.48}}
\multiput(20.74,91.76)(0.48,0.13){1}{\line(1,0){0.48}}
\multiput(20.26,91.63)(0.48,0.13){1}{\line(1,0){0.48}}
\multiput(19.78,91.5)(0.48,0.14){1}{\line(1,0){0.48}}
\multiput(19.3,91.36)(0.48,0.14){1}{\line(1,0){0.48}}
\multiput(18.82,91.22)(0.48,0.14){1}{\line(1,0){0.48}}
\multiput(18.34,91.07)(0.48,0.14){1}{\line(1,0){0.48}}
\multiput(17.86,90.93)(0.48,0.15){1}{\line(1,0){0.48}}
\multiput(17.38,90.78)(0.48,0.15){1}{\line(1,0){0.48}}
\multiput(16.9,90.63)(0.48,0.15){1}{\line(1,0){0.48}}
\multiput(16.43,90.48)(0.48,0.15){1}{\line(1,0){0.48}}
\multiput(15.95,90.32)(0.48,0.16){1}{\line(1,0){0.48}}
\multiput(15.47,90.16)(0.48,0.16){1}{\line(1,0){0.48}}
\multiput(15,90)(0.47,0.16){1}{\line(1,0){0.47}}

\linethickness{0.3mm}
\multiput(184.53,90.16)(0.47,-0.16){1}{\line(1,0){0.47}}
\multiput(184.05,90.32)(0.48,-0.16){1}{\line(1,0){0.48}}
\multiput(183.57,90.48)(0.48,-0.16){1}{\line(1,0){0.48}}
\multiput(183.1,90.63)(0.48,-0.15){1}{\line(1,0){0.48}}
\multiput(182.62,90.78)(0.48,-0.15){1}{\line(1,0){0.48}}
\multiput(182.14,90.93)(0.48,-0.15){1}{\line(1,0){0.48}}
\multiput(181.66,91.07)(0.48,-0.15){1}{\line(1,0){0.48}}
\multiput(181.18,91.22)(0.48,-0.14){1}{\line(1,0){0.48}}
\multiput(180.7,91.36)(0.48,-0.14){1}{\line(1,0){0.48}}
\multiput(180.22,91.5)(0.48,-0.14){1}{\line(1,0){0.48}}
\multiput(179.74,91.63)(0.48,-0.14){1}{\line(1,0){0.48}}
\multiput(179.26,91.76)(0.48,-0.13){1}{\line(1,0){0.48}}
\multiput(178.77,91.89)(0.48,-0.13){1}{\line(1,0){0.48}}
\multiput(178.29,92.02)(0.48,-0.13){1}{\line(1,0){0.48}}
\multiput(177.8,92.15)(0.49,-0.12){1}{\line(1,0){0.49}}
\multiput(177.32,92.27)(0.49,-0.12){1}{\line(1,0){0.49}}
\multiput(176.83,92.39)(0.49,-0.12){1}{\line(1,0){0.49}}
\multiput(176.34,92.5)(0.49,-0.12){1}{\line(1,0){0.49}}
\multiput(175.86,92.62)(0.49,-0.11){1}{\line(1,0){0.49}}
\multiput(175.37,92.73)(0.49,-0.11){1}{\line(1,0){0.49}}
\multiput(174.88,92.84)(0.49,-0.11){1}{\line(1,0){0.49}}
\multiput(174.39,92.95)(0.49,-0.11){1}{\line(1,0){0.49}}
\multiput(173.9,93.05)(0.49,-0.1){1}{\line(1,0){0.49}}
\multiput(173.41,93.15)(0.49,-0.1){1}{\line(1,0){0.49}}
\multiput(172.92,93.25)(0.49,-0.1){1}{\line(1,0){0.49}}
\multiput(172.43,93.34)(0.49,-0.1){1}{\line(1,0){0.49}}
\multiput(171.93,93.44)(0.49,-0.09){1}{\line(1,0){0.49}}
\multiput(171.44,93.53)(0.49,-0.09){1}{\line(1,0){0.49}}
\multiput(170.95,93.61)(0.49,-0.09){1}{\line(1,0){0.49}}
\multiput(170.45,93.7)(0.49,-0.09){1}{\line(1,0){0.49}}
\multiput(169.96,93.78)(0.49,-0.08){1}{\line(1,0){0.49}}
\multiput(169.47,93.86)(0.49,-0.08){1}{\line(1,0){0.49}}
\multiput(168.97,93.94)(0.49,-0.08){1}{\line(1,0){0.49}}
\multiput(168.48,94.01)(0.5,-0.07){1}{\line(1,0){0.5}}
\multiput(167.98,94.08)(0.5,-0.07){1}{\line(1,0){0.5}}
\multiput(167.48,94.15)(0.5,-0.07){1}{\line(1,0){0.5}}
\multiput(166.99,94.22)(0.5,-0.07){1}{\line(1,0){0.5}}
\multiput(166.49,94.28)(0.5,-0.06){1}{\line(1,0){0.5}}
\multiput(165.99,94.34)(0.5,-0.06){1}{\line(1,0){0.5}}
\multiput(165.5,94.4)(0.5,-0.06){1}{\line(1,0){0.5}}
\multiput(165,94.46)(0.5,-0.06){1}{\line(1,0){0.5}}
\multiput(164.5,94.51)(0.5,-0.05){1}{\line(1,0){0.5}}
\multiput(164,94.56)(0.5,-0.05){1}{\line(1,0){0.5}}
\multiput(163.5,94.61)(0.5,-0.05){1}{\line(1,0){0.5}}
\multiput(163,94.65)(0.5,-0.04){1}{\line(1,0){0.5}}
\multiput(162.5,94.7)(0.5,-0.04){1}{\line(1,0){0.5}}
\multiput(162.01,94.73)(0.5,-0.04){1}{\line(1,0){0.5}}
\multiput(161.51,94.77)(0.5,-0.04){1}{\line(1,0){0.5}}
\multiput(161.01,94.8)(0.5,-0.03){1}{\line(1,0){0.5}}
\multiput(160.51,94.84)(0.5,-0.03){1}{\line(1,0){0.5}}
\multiput(160.01,94.86)(0.5,-0.03){1}{\line(1,0){0.5}}
\multiput(159.51,94.89)(0.5,-0.03){1}{\line(1,0){0.5}}
\multiput(159.01,94.91)(0.5,-0.02){1}{\line(1,0){0.5}}
\multiput(158.51,94.93)(0.5,-0.02){1}{\line(1,0){0.5}}
\multiput(158,94.95)(0.5,-0.02){1}{\line(1,0){0.5}}
\multiput(157.5,94.97)(0.5,-0.01){1}{\line(1,0){0.5}}
\multiput(157,94.98)(0.5,-0.01){1}{\line(1,0){0.5}}
\multiput(156.5,94.99)(0.5,-0.01){1}{\line(1,0){0.5}}
\multiput(156,94.99)(0.5,-0.01){1}{\line(1,0){0.5}}
\multiput(155.5,95)(0.5,-0){1}{\line(1,0){0.5}}
\multiput(155,95)(0.5,-0){1}{\line(1,0){0.5}}
\multiput(154.5,95)(0.5,0){1}{\line(1,0){0.5}}
\multiput(154,94.99)(0.5,0){1}{\line(1,0){0.5}}
\multiput(153.5,94.99)(0.5,0.01){1}{\line(1,0){0.5}}
\multiput(153,94.98)(0.5,0.01){1}{\line(1,0){0.5}}
\multiput(152.5,94.97)(0.5,0.01){1}{\line(1,0){0.5}}
\multiput(152,94.95)(0.5,0.01){1}{\line(1,0){0.5}}
\multiput(151.49,94.93)(0.5,0.02){1}{\line(1,0){0.5}}
\multiput(150.99,94.91)(0.5,0.02){1}{\line(1,0){0.5}}
\multiput(150.49,94.89)(0.5,0.02){1}{\line(1,0){0.5}}
\multiput(149.99,94.86)(0.5,0.03){1}{\line(1,0){0.5}}
\multiput(149.49,94.84)(0.5,0.03){1}{\line(1,0){0.5}}
\multiput(148.99,94.8)(0.5,0.03){1}{\line(1,0){0.5}}
\multiput(148.49,94.77)(0.5,0.03){1}{\line(1,0){0.5}}
\multiput(147.99,94.73)(0.5,0.04){1}{\line(1,0){0.5}}
\multiput(147.5,94.7)(0.5,0.04){1}{\line(1,0){0.5}}
\multiput(147,94.65)(0.5,0.04){1}{\line(1,0){0.5}}
\multiput(146.5,94.61)(0.5,0.04){1}{\line(1,0){0.5}}
\multiput(146,94.56)(0.5,0.05){1}{\line(1,0){0.5}}
\multiput(145.5,94.51)(0.5,0.05){1}{\line(1,0){0.5}}
\multiput(145,94.46)(0.5,0.05){1}{\line(1,0){0.5}}
\multiput(144.5,94.4)(0.5,0.06){1}{\line(1,0){0.5}}
\multiput(144.01,94.34)(0.5,0.06){1}{\line(1,0){0.5}}
\multiput(143.51,94.28)(0.5,0.06){1}{\line(1,0){0.5}}
\multiput(143.01,94.22)(0.5,0.06){1}{\line(1,0){0.5}}
\multiput(142.52,94.15)(0.5,0.07){1}{\line(1,0){0.5}}
\multiput(142.02,94.08)(0.5,0.07){1}{\line(1,0){0.5}}
\multiput(141.52,94.01)(0.5,0.07){1}{\line(1,0){0.5}}
\multiput(141.03,93.94)(0.5,0.07){1}{\line(1,0){0.5}}
\multiput(140.53,93.86)(0.49,0.08){1}{\line(1,0){0.49}}
\multiput(140.04,93.78)(0.49,0.08){1}{\line(1,0){0.49}}
\multiput(139.55,93.7)(0.49,0.08){1}{\line(1,0){0.49}}
\multiput(139.05,93.61)(0.49,0.09){1}{\line(1,0){0.49}}
\multiput(138.56,93.53)(0.49,0.09){1}{\line(1,0){0.49}}
\multiput(138.07,93.44)(0.49,0.09){1}{\line(1,0){0.49}}
\multiput(137.57,93.34)(0.49,0.09){1}{\line(1,0){0.49}}
\multiput(137.08,93.25)(0.49,0.1){1}{\line(1,0){0.49}}
\multiput(136.59,93.15)(0.49,0.1){1}{\line(1,0){0.49}}
\multiput(136.1,93.05)(0.49,0.1){1}{\line(1,0){0.49}}
\multiput(135.61,92.95)(0.49,0.1){1}{\line(1,0){0.49}}
\multiput(135.12,92.84)(0.49,0.11){1}{\line(1,0){0.49}}
\multiput(134.63,92.73)(0.49,0.11){1}{\line(1,0){0.49}}
\multiput(134.14,92.62)(0.49,0.11){1}{\line(1,0){0.49}}
\multiput(133.66,92.5)(0.49,0.11){1}{\line(1,0){0.49}}
\multiput(133.17,92.39)(0.49,0.12){1}{\line(1,0){0.49}}
\multiput(132.68,92.27)(0.49,0.12){1}{\line(1,0){0.49}}
\multiput(132.2,92.15)(0.49,0.12){1}{\line(1,0){0.49}}
\multiput(131.71,92.02)(0.49,0.12){1}{\line(1,0){0.49}}
\multiput(131.23,91.89)(0.48,0.13){1}{\line(1,0){0.48}}
\multiput(130.74,91.76)(0.48,0.13){1}{\line(1,0){0.48}}
\multiput(130.26,91.63)(0.48,0.13){1}{\line(1,0){0.48}}
\multiput(129.78,91.5)(0.48,0.14){1}{\line(1,0){0.48}}
\multiput(129.3,91.36)(0.48,0.14){1}{\line(1,0){0.48}}
\multiput(128.82,91.22)(0.48,0.14){1}{\line(1,0){0.48}}
\multiput(128.34,91.07)(0.48,0.14){1}{\line(1,0){0.48}}
\multiput(127.86,90.93)(0.48,0.15){1}{\line(1,0){0.48}}
\multiput(127.38,90.78)(0.48,0.15){1}{\line(1,0){0.48}}
\multiput(126.9,90.63)(0.48,0.15){1}{\line(1,0){0.48}}
\multiput(126.43,90.48)(0.48,0.15){1}{\line(1,0){0.48}}
\multiput(125.95,90.32)(0.48,0.16){1}{\line(1,0){0.48}}
\multiput(125.47,90.16)(0.48,0.16){1}{\line(1,0){0.48}}
\multiput(125,90)(0.47,0.16){1}{\line(1,0){0.47}}

\linethickness{0.3mm}
\multiput(184.53,55.16)(0.47,-0.16){1}{\line(1,0){0.47}}
\multiput(184.05,55.32)(0.48,-0.16){1}{\line(1,0){0.48}}
\multiput(183.57,55.48)(0.48,-0.16){1}{\line(1,0){0.48}}
\multiput(183.1,55.63)(0.48,-0.15){1}{\line(1,0){0.48}}
\multiput(182.62,55.78)(0.48,-0.15){1}{\line(1,0){0.48}}
\multiput(182.14,55.93)(0.48,-0.15){1}{\line(1,0){0.48}}
\multiput(181.66,56.07)(0.48,-0.15){1}{\line(1,0){0.48}}
\multiput(181.18,56.22)(0.48,-0.14){1}{\line(1,0){0.48}}
\multiput(180.7,56.36)(0.48,-0.14){1}{\line(1,0){0.48}}
\multiput(180.22,56.5)(0.48,-0.14){1}{\line(1,0){0.48}}
\multiput(179.74,56.63)(0.48,-0.14){1}{\line(1,0){0.48}}
\multiput(179.26,56.76)(0.48,-0.13){1}{\line(1,0){0.48}}
\multiput(178.77,56.89)(0.48,-0.13){1}{\line(1,0){0.48}}
\multiput(178.29,57.02)(0.48,-0.13){1}{\line(1,0){0.48}}
\multiput(177.8,57.15)(0.49,-0.12){1}{\line(1,0){0.49}}
\multiput(177.32,57.27)(0.49,-0.12){1}{\line(1,0){0.49}}
\multiput(176.83,57.39)(0.49,-0.12){1}{\line(1,0){0.49}}
\multiput(176.34,57.5)(0.49,-0.12){1}{\line(1,0){0.49}}
\multiput(175.86,57.62)(0.49,-0.11){1}{\line(1,0){0.49}}
\multiput(175.37,57.73)(0.49,-0.11){1}{\line(1,0){0.49}}
\multiput(174.88,57.84)(0.49,-0.11){1}{\line(1,0){0.49}}
\multiput(174.39,57.95)(0.49,-0.11){1}{\line(1,0){0.49}}
\multiput(173.9,58.05)(0.49,-0.1){1}{\line(1,0){0.49}}
\multiput(173.41,58.15)(0.49,-0.1){1}{\line(1,0){0.49}}
\multiput(172.92,58.25)(0.49,-0.1){1}{\line(1,0){0.49}}
\multiput(172.43,58.34)(0.49,-0.1){1}{\line(1,0){0.49}}
\multiput(171.93,58.44)(0.49,-0.09){1}{\line(1,0){0.49}}
\multiput(171.44,58.53)(0.49,-0.09){1}{\line(1,0){0.49}}
\multiput(170.95,58.61)(0.49,-0.09){1}{\line(1,0){0.49}}
\multiput(170.45,58.7)(0.49,-0.09){1}{\line(1,0){0.49}}
\multiput(169.96,58.78)(0.49,-0.08){1}{\line(1,0){0.49}}
\multiput(169.47,58.86)(0.49,-0.08){1}{\line(1,0){0.49}}
\multiput(168.97,58.94)(0.49,-0.08){1}{\line(1,0){0.49}}
\multiput(168.48,59.01)(0.5,-0.07){1}{\line(1,0){0.5}}
\multiput(167.98,59.08)(0.5,-0.07){1}{\line(1,0){0.5}}
\multiput(167.48,59.15)(0.5,-0.07){1}{\line(1,0){0.5}}
\multiput(166.99,59.22)(0.5,-0.07){1}{\line(1,0){0.5}}
\multiput(166.49,59.28)(0.5,-0.06){1}{\line(1,0){0.5}}
\multiput(165.99,59.34)(0.5,-0.06){1}{\line(1,0){0.5}}
\multiput(165.5,59.4)(0.5,-0.06){1}{\line(1,0){0.5}}
\multiput(165,59.46)(0.5,-0.06){1}{\line(1,0){0.5}}
\multiput(164.5,59.51)(0.5,-0.05){1}{\line(1,0){0.5}}
\multiput(164,59.56)(0.5,-0.05){1}{\line(1,0){0.5}}
\multiput(163.5,59.61)(0.5,-0.05){1}{\line(1,0){0.5}}
\multiput(163,59.65)(0.5,-0.04){1}{\line(1,0){0.5}}
\multiput(162.5,59.7)(0.5,-0.04){1}{\line(1,0){0.5}}
\multiput(162.01,59.73)(0.5,-0.04){1}{\line(1,0){0.5}}
\multiput(161.51,59.77)(0.5,-0.04){1}{\line(1,0){0.5}}
\multiput(161.01,59.8)(0.5,-0.03){1}{\line(1,0){0.5}}
\multiput(160.51,59.84)(0.5,-0.03){1}{\line(1,0){0.5}}
\multiput(160.01,59.86)(0.5,-0.03){1}{\line(1,0){0.5}}
\multiput(159.51,59.89)(0.5,-0.03){1}{\line(1,0){0.5}}
\multiput(159.01,59.91)(0.5,-0.02){1}{\line(1,0){0.5}}
\multiput(158.51,59.93)(0.5,-0.02){1}{\line(1,0){0.5}}
\multiput(158,59.95)(0.5,-0.02){1}{\line(1,0){0.5}}
\multiput(157.5,59.97)(0.5,-0.01){1}{\line(1,0){0.5}}
\multiput(157,59.98)(0.5,-0.01){1}{\line(1,0){0.5}}
\multiput(156.5,59.99)(0.5,-0.01){1}{\line(1,0){0.5}}
\multiput(156,59.99)(0.5,-0.01){1}{\line(1,0){0.5}}
\multiput(155.5,60)(0.5,-0){1}{\line(1,0){0.5}}
\multiput(155,60)(0.5,-0){1}{\line(1,0){0.5}}
\multiput(154.5,60)(0.5,0){1}{\line(1,0){0.5}}
\multiput(154,59.99)(0.5,0){1}{\line(1,0){0.5}}
\multiput(153.5,59.99)(0.5,0.01){1}{\line(1,0){0.5}}
\multiput(153,59.98)(0.5,0.01){1}{\line(1,0){0.5}}
\multiput(152.5,59.97)(0.5,0.01){1}{\line(1,0){0.5}}
\multiput(152,59.95)(0.5,0.01){1}{\line(1,0){0.5}}
\multiput(151.49,59.93)(0.5,0.02){1}{\line(1,0){0.5}}
\multiput(150.99,59.91)(0.5,0.02){1}{\line(1,0){0.5}}
\multiput(150.49,59.89)(0.5,0.02){1}{\line(1,0){0.5}}
\multiput(149.99,59.86)(0.5,0.03){1}{\line(1,0){0.5}}
\multiput(149.49,59.84)(0.5,0.03){1}{\line(1,0){0.5}}
\multiput(148.99,59.8)(0.5,0.03){1}{\line(1,0){0.5}}
\multiput(148.49,59.77)(0.5,0.03){1}{\line(1,0){0.5}}
\multiput(147.99,59.73)(0.5,0.04){1}{\line(1,0){0.5}}
\multiput(147.5,59.7)(0.5,0.04){1}{\line(1,0){0.5}}
\multiput(147,59.65)(0.5,0.04){1}{\line(1,0){0.5}}
\multiput(146.5,59.61)(0.5,0.04){1}{\line(1,0){0.5}}
\multiput(146,59.56)(0.5,0.05){1}{\line(1,0){0.5}}
\multiput(145.5,59.51)(0.5,0.05){1}{\line(1,0){0.5}}
\multiput(145,59.46)(0.5,0.05){1}{\line(1,0){0.5}}
\multiput(144.5,59.4)(0.5,0.06){1}{\line(1,0){0.5}}
\multiput(144.01,59.34)(0.5,0.06){1}{\line(1,0){0.5}}
\multiput(143.51,59.28)(0.5,0.06){1}{\line(1,0){0.5}}
\multiput(143.01,59.22)(0.5,0.06){1}{\line(1,0){0.5}}
\multiput(142.52,59.15)(0.5,0.07){1}{\line(1,0){0.5}}
\multiput(142.02,59.08)(0.5,0.07){1}{\line(1,0){0.5}}
\multiput(141.52,59.01)(0.5,0.07){1}{\line(1,0){0.5}}
\multiput(141.03,58.94)(0.5,0.07){1}{\line(1,0){0.5}}
\multiput(140.53,58.86)(0.49,0.08){1}{\line(1,0){0.49}}
\multiput(140.04,58.78)(0.49,0.08){1}{\line(1,0){0.49}}
\multiput(139.55,58.7)(0.49,0.08){1}{\line(1,0){0.49}}
\multiput(139.05,58.61)(0.49,0.09){1}{\line(1,0){0.49}}
\multiput(138.56,58.53)(0.49,0.09){1}{\line(1,0){0.49}}
\multiput(138.07,58.44)(0.49,0.09){1}{\line(1,0){0.49}}
\multiput(137.57,58.34)(0.49,0.09){1}{\line(1,0){0.49}}
\multiput(137.08,58.25)(0.49,0.1){1}{\line(1,0){0.49}}
\multiput(136.59,58.15)(0.49,0.1){1}{\line(1,0){0.49}}
\multiput(136.1,58.05)(0.49,0.1){1}{\line(1,0){0.49}}
\multiput(135.61,57.95)(0.49,0.1){1}{\line(1,0){0.49}}
\multiput(135.12,57.84)(0.49,0.11){1}{\line(1,0){0.49}}
\multiput(134.63,57.73)(0.49,0.11){1}{\line(1,0){0.49}}
\multiput(134.14,57.62)(0.49,0.11){1}{\line(1,0){0.49}}
\multiput(133.66,57.5)(0.49,0.11){1}{\line(1,0){0.49}}
\multiput(133.17,57.39)(0.49,0.12){1}{\line(1,0){0.49}}
\multiput(132.68,57.27)(0.49,0.12){1}{\line(1,0){0.49}}
\multiput(132.2,57.15)(0.49,0.12){1}{\line(1,0){0.49}}
\multiput(131.71,57.02)(0.49,0.12){1}{\line(1,0){0.49}}
\multiput(131.23,56.89)(0.48,0.13){1}{\line(1,0){0.48}}
\multiput(130.74,56.76)(0.48,0.13){1}{\line(1,0){0.48}}
\multiput(130.26,56.63)(0.48,0.13){1}{\line(1,0){0.48}}
\multiput(129.78,56.5)(0.48,0.14){1}{\line(1,0){0.48}}
\multiput(129.3,56.36)(0.48,0.14){1}{\line(1,0){0.48}}
\multiput(128.82,56.22)(0.48,0.14){1}{\line(1,0){0.48}}
\multiput(128.34,56.07)(0.48,0.14){1}{\line(1,0){0.48}}
\multiput(127.86,55.93)(0.48,0.15){1}{\line(1,0){0.48}}
\multiput(127.38,55.78)(0.48,0.15){1}{\line(1,0){0.48}}
\multiput(126.9,55.63)(0.48,0.15){1}{\line(1,0){0.48}}
\multiput(126.43,55.48)(0.48,0.15){1}{\line(1,0){0.48}}
\multiput(125.95,55.32)(0.48,0.16){1}{\line(1,0){0.48}}
\multiput(125.47,55.16)(0.48,0.16){1}{\line(1,0){0.48}}
\multiput(125,55)(0.47,0.16){1}{\line(1,0){0.47}}

\linethickness{0.3mm}
\multiput(74.53,55.16)(0.47,-0.16){1}{\line(1,0){0.47}}
\multiput(74.05,55.32)(0.48,-0.16){1}{\line(1,0){0.48}}
\multiput(73.57,55.48)(0.48,-0.16){1}{\line(1,0){0.48}}
\multiput(73.1,55.63)(0.48,-0.15){1}{\line(1,0){0.48}}
\multiput(72.62,55.78)(0.48,-0.15){1}{\line(1,0){0.48}}
\multiput(72.14,55.93)(0.48,-0.15){1}{\line(1,0){0.48}}
\multiput(71.66,56.07)(0.48,-0.15){1}{\line(1,0){0.48}}
\multiput(71.18,56.22)(0.48,-0.14){1}{\line(1,0){0.48}}
\multiput(70.7,56.36)(0.48,-0.14){1}{\line(1,0){0.48}}
\multiput(70.22,56.5)(0.48,-0.14){1}{\line(1,0){0.48}}
\multiput(69.74,56.63)(0.48,-0.14){1}{\line(1,0){0.48}}
\multiput(69.26,56.76)(0.48,-0.13){1}{\line(1,0){0.48}}
\multiput(68.77,56.89)(0.48,-0.13){1}{\line(1,0){0.48}}
\multiput(68.29,57.02)(0.48,-0.13){1}{\line(1,0){0.48}}
\multiput(67.8,57.15)(0.49,-0.12){1}{\line(1,0){0.49}}
\multiput(67.32,57.27)(0.49,-0.12){1}{\line(1,0){0.49}}
\multiput(66.83,57.39)(0.49,-0.12){1}{\line(1,0){0.49}}
\multiput(66.34,57.5)(0.49,-0.12){1}{\line(1,0){0.49}}
\multiput(65.86,57.62)(0.49,-0.11){1}{\line(1,0){0.49}}
\multiput(65.37,57.73)(0.49,-0.11){1}{\line(1,0){0.49}}
\multiput(64.88,57.84)(0.49,-0.11){1}{\line(1,0){0.49}}
\multiput(64.39,57.95)(0.49,-0.11){1}{\line(1,0){0.49}}
\multiput(63.9,58.05)(0.49,-0.1){1}{\line(1,0){0.49}}
\multiput(63.41,58.15)(0.49,-0.1){1}{\line(1,0){0.49}}
\multiput(62.92,58.25)(0.49,-0.1){1}{\line(1,0){0.49}}
\multiput(62.43,58.34)(0.49,-0.1){1}{\line(1,0){0.49}}
\multiput(61.93,58.44)(0.49,-0.09){1}{\line(1,0){0.49}}
\multiput(61.44,58.53)(0.49,-0.09){1}{\line(1,0){0.49}}
\multiput(60.95,58.61)(0.49,-0.09){1}{\line(1,0){0.49}}
\multiput(60.45,58.7)(0.49,-0.09){1}{\line(1,0){0.49}}
\multiput(59.96,58.78)(0.49,-0.08){1}{\line(1,0){0.49}}
\multiput(59.47,58.86)(0.49,-0.08){1}{\line(1,0){0.49}}
\multiput(58.97,58.94)(0.49,-0.08){1}{\line(1,0){0.49}}
\multiput(58.48,59.01)(0.5,-0.07){1}{\line(1,0){0.5}}
\multiput(57.98,59.08)(0.5,-0.07){1}{\line(1,0){0.5}}
\multiput(57.48,59.15)(0.5,-0.07){1}{\line(1,0){0.5}}
\multiput(56.99,59.22)(0.5,-0.07){1}{\line(1,0){0.5}}
\multiput(56.49,59.28)(0.5,-0.06){1}{\line(1,0){0.5}}
\multiput(55.99,59.34)(0.5,-0.06){1}{\line(1,0){0.5}}
\multiput(55.5,59.4)(0.5,-0.06){1}{\line(1,0){0.5}}
\multiput(55,59.46)(0.5,-0.06){1}{\line(1,0){0.5}}
\multiput(54.5,59.51)(0.5,-0.05){1}{\line(1,0){0.5}}
\multiput(54,59.56)(0.5,-0.05){1}{\line(1,0){0.5}}
\multiput(53.5,59.61)(0.5,-0.05){1}{\line(1,0){0.5}}
\multiput(53,59.65)(0.5,-0.04){1}{\line(1,0){0.5}}
\multiput(52.5,59.7)(0.5,-0.04){1}{\line(1,0){0.5}}
\multiput(52.01,59.73)(0.5,-0.04){1}{\line(1,0){0.5}}
\multiput(51.51,59.77)(0.5,-0.04){1}{\line(1,0){0.5}}
\multiput(51.01,59.8)(0.5,-0.03){1}{\line(1,0){0.5}}
\multiput(50.51,59.84)(0.5,-0.03){1}{\line(1,0){0.5}}
\multiput(50.01,59.86)(0.5,-0.03){1}{\line(1,0){0.5}}
\multiput(49.51,59.89)(0.5,-0.03){1}{\line(1,0){0.5}}
\multiput(49.01,59.91)(0.5,-0.02){1}{\line(1,0){0.5}}
\multiput(48.51,59.93)(0.5,-0.02){1}{\line(1,0){0.5}}
\multiput(48,59.95)(0.5,-0.02){1}{\line(1,0){0.5}}
\multiput(47.5,59.97)(0.5,-0.01){1}{\line(1,0){0.5}}
\multiput(47,59.98)(0.5,-0.01){1}{\line(1,0){0.5}}
\multiput(46.5,59.99)(0.5,-0.01){1}{\line(1,0){0.5}}
\multiput(46,59.99)(0.5,-0.01){1}{\line(1,0){0.5}}
\multiput(45.5,60)(0.5,-0){1}{\line(1,0){0.5}}
\multiput(45,60)(0.5,-0){1}{\line(1,0){0.5}}
\multiput(44.5,60)(0.5,0){1}{\line(1,0){0.5}}
\multiput(44,59.99)(0.5,0){1}{\line(1,0){0.5}}
\multiput(43.5,59.99)(0.5,0.01){1}{\line(1,0){0.5}}
\multiput(43,59.98)(0.5,0.01){1}{\line(1,0){0.5}}
\multiput(42.5,59.97)(0.5,0.01){1}{\line(1,0){0.5}}
\multiput(42,59.95)(0.5,0.01){1}{\line(1,0){0.5}}
\multiput(41.49,59.93)(0.5,0.02){1}{\line(1,0){0.5}}
\multiput(40.99,59.91)(0.5,0.02){1}{\line(1,0){0.5}}
\multiput(40.49,59.89)(0.5,0.02){1}{\line(1,0){0.5}}
\multiput(39.99,59.86)(0.5,0.03){1}{\line(1,0){0.5}}
\multiput(39.49,59.84)(0.5,0.03){1}{\line(1,0){0.5}}
\multiput(38.99,59.8)(0.5,0.03){1}{\line(1,0){0.5}}
\multiput(38.49,59.77)(0.5,0.03){1}{\line(1,0){0.5}}
\multiput(37.99,59.73)(0.5,0.04){1}{\line(1,0){0.5}}
\multiput(37.5,59.7)(0.5,0.04){1}{\line(1,0){0.5}}
\multiput(37,59.65)(0.5,0.04){1}{\line(1,0){0.5}}
\multiput(36.5,59.61)(0.5,0.04){1}{\line(1,0){0.5}}
\multiput(36,59.56)(0.5,0.05){1}{\line(1,0){0.5}}
\multiput(35.5,59.51)(0.5,0.05){1}{\line(1,0){0.5}}
\multiput(35,59.46)(0.5,0.05){1}{\line(1,0){0.5}}
\multiput(34.5,59.4)(0.5,0.06){1}{\line(1,0){0.5}}
\multiput(34.01,59.34)(0.5,0.06){1}{\line(1,0){0.5}}
\multiput(33.51,59.28)(0.5,0.06){1}{\line(1,0){0.5}}
\multiput(33.01,59.22)(0.5,0.06){1}{\line(1,0){0.5}}
\multiput(32.52,59.15)(0.5,0.07){1}{\line(1,0){0.5}}
\multiput(32.02,59.08)(0.5,0.07){1}{\line(1,0){0.5}}
\multiput(31.52,59.01)(0.5,0.07){1}{\line(1,0){0.5}}
\multiput(31.03,58.94)(0.5,0.07){1}{\line(1,0){0.5}}
\multiput(30.53,58.86)(0.49,0.08){1}{\line(1,0){0.49}}
\multiput(30.04,58.78)(0.49,0.08){1}{\line(1,0){0.49}}
\multiput(29.55,58.7)(0.49,0.08){1}{\line(1,0){0.49}}
\multiput(29.05,58.61)(0.49,0.09){1}{\line(1,0){0.49}}
\multiput(28.56,58.53)(0.49,0.09){1}{\line(1,0){0.49}}
\multiput(28.07,58.44)(0.49,0.09){1}{\line(1,0){0.49}}
\multiput(27.57,58.34)(0.49,0.09){1}{\line(1,0){0.49}}
\multiput(27.08,58.25)(0.49,0.1){1}{\line(1,0){0.49}}
\multiput(26.59,58.15)(0.49,0.1){1}{\line(1,0){0.49}}
\multiput(26.1,58.05)(0.49,0.1){1}{\line(1,0){0.49}}
\multiput(25.61,57.95)(0.49,0.1){1}{\line(1,0){0.49}}
\multiput(25.12,57.84)(0.49,0.11){1}{\line(1,0){0.49}}
\multiput(24.63,57.73)(0.49,0.11){1}{\line(1,0){0.49}}
\multiput(24.14,57.62)(0.49,0.11){1}{\line(1,0){0.49}}
\multiput(23.66,57.5)(0.49,0.11){1}{\line(1,0){0.49}}
\multiput(23.17,57.39)(0.49,0.12){1}{\line(1,0){0.49}}
\multiput(22.68,57.27)(0.49,0.12){1}{\line(1,0){0.49}}
\multiput(22.2,57.15)(0.49,0.12){1}{\line(1,0){0.49}}
\multiput(21.71,57.02)(0.49,0.12){1}{\line(1,0){0.49}}
\multiput(21.23,56.89)(0.48,0.13){1}{\line(1,0){0.48}}
\multiput(20.74,56.76)(0.48,0.13){1}{\line(1,0){0.48}}
\multiput(20.26,56.63)(0.48,0.13){1}{\line(1,0){0.48}}
\multiput(19.78,56.5)(0.48,0.14){1}{\line(1,0){0.48}}
\multiput(19.3,56.36)(0.48,0.14){1}{\line(1,0){0.48}}
\multiput(18.82,56.22)(0.48,0.14){1}{\line(1,0){0.48}}
\multiput(18.34,56.07)(0.48,0.14){1}{\line(1,0){0.48}}
\multiput(17.86,55.93)(0.48,0.15){1}{\line(1,0){0.48}}
\multiput(17.38,55.78)(0.48,0.15){1}{\line(1,0){0.48}}
\multiput(16.9,55.63)(0.48,0.15){1}{\line(1,0){0.48}}
\multiput(16.43,55.48)(0.48,0.15){1}{\line(1,0){0.48}}
\multiput(15.95,55.32)(0.48,0.16){1}{\line(1,0){0.48}}
\multiput(15.47,55.16)(0.48,0.16){1}{\line(1,0){0.48}}
\multiput(15,55)(0.47,0.16){1}{\line(1,0){0.47}}

\linethickness{0.3mm}
\multiput(74.53,20.16)(0.47,-0.16){1}{\line(1,0){0.47}}
\multiput(74.05,20.32)(0.48,-0.16){1}{\line(1,0){0.48}}
\multiput(73.57,20.48)(0.48,-0.16){1}{\line(1,0){0.48}}
\multiput(73.1,20.63)(0.48,-0.15){1}{\line(1,0){0.48}}
\multiput(72.62,20.78)(0.48,-0.15){1}{\line(1,0){0.48}}
\multiput(72.14,20.93)(0.48,-0.15){1}{\line(1,0){0.48}}
\multiput(71.66,21.07)(0.48,-0.15){1}{\line(1,0){0.48}}
\multiput(71.18,21.22)(0.48,-0.14){1}{\line(1,0){0.48}}
\multiput(70.7,21.36)(0.48,-0.14){1}{\line(1,0){0.48}}
\multiput(70.22,21.5)(0.48,-0.14){1}{\line(1,0){0.48}}
\multiput(69.74,21.63)(0.48,-0.14){1}{\line(1,0){0.48}}
\multiput(69.26,21.76)(0.48,-0.13){1}{\line(1,0){0.48}}
\multiput(68.77,21.89)(0.48,-0.13){1}{\line(1,0){0.48}}
\multiput(68.29,22.02)(0.48,-0.13){1}{\line(1,0){0.48}}
\multiput(67.8,22.15)(0.49,-0.12){1}{\line(1,0){0.49}}
\multiput(67.32,22.27)(0.49,-0.12){1}{\line(1,0){0.49}}
\multiput(66.83,22.39)(0.49,-0.12){1}{\line(1,0){0.49}}
\multiput(66.34,22.5)(0.49,-0.12){1}{\line(1,0){0.49}}
\multiput(65.86,22.62)(0.49,-0.11){1}{\line(1,0){0.49}}
\multiput(65.37,22.73)(0.49,-0.11){1}{\line(1,0){0.49}}
\multiput(64.88,22.84)(0.49,-0.11){1}{\line(1,0){0.49}}
\multiput(64.39,22.95)(0.49,-0.11){1}{\line(1,0){0.49}}
\multiput(63.9,23.05)(0.49,-0.1){1}{\line(1,0){0.49}}
\multiput(63.41,23.15)(0.49,-0.1){1}{\line(1,0){0.49}}
\multiput(62.92,23.25)(0.49,-0.1){1}{\line(1,0){0.49}}
\multiput(62.43,23.34)(0.49,-0.1){1}{\line(1,0){0.49}}
\multiput(61.93,23.44)(0.49,-0.09){1}{\line(1,0){0.49}}
\multiput(61.44,23.53)(0.49,-0.09){1}{\line(1,0){0.49}}
\multiput(60.95,23.61)(0.49,-0.09){1}{\line(1,0){0.49}}
\multiput(60.45,23.7)(0.49,-0.09){1}{\line(1,0){0.49}}
\multiput(59.96,23.78)(0.49,-0.08){1}{\line(1,0){0.49}}
\multiput(59.47,23.86)(0.49,-0.08){1}{\line(1,0){0.49}}
\multiput(58.97,23.94)(0.49,-0.08){1}{\line(1,0){0.49}}
\multiput(58.48,24.01)(0.5,-0.07){1}{\line(1,0){0.5}}
\multiput(57.98,24.08)(0.5,-0.07){1}{\line(1,0){0.5}}
\multiput(57.48,24.15)(0.5,-0.07){1}{\line(1,0){0.5}}
\multiput(56.99,24.22)(0.5,-0.07){1}{\line(1,0){0.5}}
\multiput(56.49,24.28)(0.5,-0.06){1}{\line(1,0){0.5}}
\multiput(55.99,24.34)(0.5,-0.06){1}{\line(1,0){0.5}}
\multiput(55.5,24.4)(0.5,-0.06){1}{\line(1,0){0.5}}
\multiput(55,24.46)(0.5,-0.06){1}{\line(1,0){0.5}}
\multiput(54.5,24.51)(0.5,-0.05){1}{\line(1,0){0.5}}
\multiput(54,24.56)(0.5,-0.05){1}{\line(1,0){0.5}}
\multiput(53.5,24.61)(0.5,-0.05){1}{\line(1,0){0.5}}
\multiput(53,24.65)(0.5,-0.04){1}{\line(1,0){0.5}}
\multiput(52.5,24.7)(0.5,-0.04){1}{\line(1,0){0.5}}
\multiput(52.01,24.73)(0.5,-0.04){1}{\line(1,0){0.5}}
\multiput(51.51,24.77)(0.5,-0.04){1}{\line(1,0){0.5}}
\multiput(51.01,24.8)(0.5,-0.03){1}{\line(1,0){0.5}}
\multiput(50.51,24.84)(0.5,-0.03){1}{\line(1,0){0.5}}
\multiput(50.01,24.86)(0.5,-0.03){1}{\line(1,0){0.5}}
\multiput(49.51,24.89)(0.5,-0.03){1}{\line(1,0){0.5}}
\multiput(49.01,24.91)(0.5,-0.02){1}{\line(1,0){0.5}}
\multiput(48.51,24.93)(0.5,-0.02){1}{\line(1,0){0.5}}
\multiput(48,24.95)(0.5,-0.02){1}{\line(1,0){0.5}}
\multiput(47.5,24.97)(0.5,-0.01){1}{\line(1,0){0.5}}
\multiput(47,24.98)(0.5,-0.01){1}{\line(1,0){0.5}}
\multiput(46.5,24.99)(0.5,-0.01){1}{\line(1,0){0.5}}
\multiput(46,24.99)(0.5,-0.01){1}{\line(1,0){0.5}}
\multiput(45.5,25)(0.5,-0){1}{\line(1,0){0.5}}
\multiput(45,25)(0.5,-0){1}{\line(1,0){0.5}}
\multiput(44.5,25)(0.5,0){1}{\line(1,0){0.5}}
\multiput(44,24.99)(0.5,0){1}{\line(1,0){0.5}}
\multiput(43.5,24.99)(0.5,0.01){1}{\line(1,0){0.5}}
\multiput(43,24.98)(0.5,0.01){1}{\line(1,0){0.5}}
\multiput(42.5,24.97)(0.5,0.01){1}{\line(1,0){0.5}}
\multiput(42,24.95)(0.5,0.01){1}{\line(1,0){0.5}}
\multiput(41.49,24.93)(0.5,0.02){1}{\line(1,0){0.5}}
\multiput(40.99,24.91)(0.5,0.02){1}{\line(1,0){0.5}}
\multiput(40.49,24.89)(0.5,0.02){1}{\line(1,0){0.5}}
\multiput(39.99,24.86)(0.5,0.03){1}{\line(1,0){0.5}}
\multiput(39.49,24.84)(0.5,0.03){1}{\line(1,0){0.5}}
\multiput(38.99,24.8)(0.5,0.03){1}{\line(1,0){0.5}}
\multiput(38.49,24.77)(0.5,0.03){1}{\line(1,0){0.5}}
\multiput(37.99,24.73)(0.5,0.04){1}{\line(1,0){0.5}}
\multiput(37.5,24.7)(0.5,0.04){1}{\line(1,0){0.5}}
\multiput(37,24.65)(0.5,0.04){1}{\line(1,0){0.5}}
\multiput(36.5,24.61)(0.5,0.04){1}{\line(1,0){0.5}}
\multiput(36,24.56)(0.5,0.05){1}{\line(1,0){0.5}}
\multiput(35.5,24.51)(0.5,0.05){1}{\line(1,0){0.5}}
\multiput(35,24.46)(0.5,0.05){1}{\line(1,0){0.5}}
\multiput(34.5,24.4)(0.5,0.06){1}{\line(1,0){0.5}}
\multiput(34.01,24.34)(0.5,0.06){1}{\line(1,0){0.5}}
\multiput(33.51,24.28)(0.5,0.06){1}{\line(1,0){0.5}}
\multiput(33.01,24.22)(0.5,0.06){1}{\line(1,0){0.5}}
\multiput(32.52,24.15)(0.5,0.07){1}{\line(1,0){0.5}}
\multiput(32.02,24.08)(0.5,0.07){1}{\line(1,0){0.5}}
\multiput(31.52,24.01)(0.5,0.07){1}{\line(1,0){0.5}}
\multiput(31.03,23.94)(0.5,0.07){1}{\line(1,0){0.5}}
\multiput(30.53,23.86)(0.49,0.08){1}{\line(1,0){0.49}}
\multiput(30.04,23.78)(0.49,0.08){1}{\line(1,0){0.49}}
\multiput(29.55,23.7)(0.49,0.08){1}{\line(1,0){0.49}}
\multiput(29.05,23.61)(0.49,0.09){1}{\line(1,0){0.49}}
\multiput(28.56,23.53)(0.49,0.09){1}{\line(1,0){0.49}}
\multiput(28.07,23.44)(0.49,0.09){1}{\line(1,0){0.49}}
\multiput(27.57,23.34)(0.49,0.09){1}{\line(1,0){0.49}}
\multiput(27.08,23.25)(0.49,0.1){1}{\line(1,0){0.49}}
\multiput(26.59,23.15)(0.49,0.1){1}{\line(1,0){0.49}}
\multiput(26.1,23.05)(0.49,0.1){1}{\line(1,0){0.49}}
\multiput(25.61,22.95)(0.49,0.1){1}{\line(1,0){0.49}}
\multiput(25.12,22.84)(0.49,0.11){1}{\line(1,0){0.49}}
\multiput(24.63,22.73)(0.49,0.11){1}{\line(1,0){0.49}}
\multiput(24.14,22.62)(0.49,0.11){1}{\line(1,0){0.49}}
\multiput(23.66,22.5)(0.49,0.11){1}{\line(1,0){0.49}}
\multiput(23.17,22.39)(0.49,0.12){1}{\line(1,0){0.49}}
\multiput(22.68,22.27)(0.49,0.12){1}{\line(1,0){0.49}}
\multiput(22.2,22.15)(0.49,0.12){1}{\line(1,0){0.49}}
\multiput(21.71,22.02)(0.49,0.12){1}{\line(1,0){0.49}}
\multiput(21.23,21.89)(0.48,0.13){1}{\line(1,0){0.48}}
\multiput(20.74,21.76)(0.48,0.13){1}{\line(1,0){0.48}}
\multiput(20.26,21.63)(0.48,0.13){1}{\line(1,0){0.48}}
\multiput(19.78,21.5)(0.48,0.14){1}{\line(1,0){0.48}}
\multiput(19.3,21.36)(0.48,0.14){1}{\line(1,0){0.48}}
\multiput(18.82,21.22)(0.48,0.14){1}{\line(1,0){0.48}}
\multiput(18.34,21.07)(0.48,0.14){1}{\line(1,0){0.48}}
\multiput(17.86,20.93)(0.48,0.15){1}{\line(1,0){0.48}}
\multiput(17.38,20.78)(0.48,0.15){1}{\line(1,0){0.48}}
\multiput(16.9,20.63)(0.48,0.15){1}{\line(1,0){0.48}}
\multiput(16.43,20.48)(0.48,0.15){1}{\line(1,0){0.48}}
\multiput(15.95,20.32)(0.48,0.16){1}{\line(1,0){0.48}}
\multiput(15.47,20.16)(0.48,0.16){1}{\line(1,0){0.48}}
\multiput(15,20)(0.47,0.16){1}{\line(1,0){0.47}}

\linethickness{0.3mm}
\multiput(184.53,20.16)(0.47,-0.16){1}{\line(1,0){0.47}}
\multiput(184.05,20.32)(0.48,-0.16){1}{\line(1,0){0.48}}
\multiput(183.57,20.48)(0.48,-0.16){1}{\line(1,0){0.48}}
\multiput(183.1,20.63)(0.48,-0.15){1}{\line(1,0){0.48}}
\multiput(182.62,20.78)(0.48,-0.15){1}{\line(1,0){0.48}}
\multiput(182.14,20.93)(0.48,-0.15){1}{\line(1,0){0.48}}
\multiput(181.66,21.07)(0.48,-0.15){1}{\line(1,0){0.48}}
\multiput(181.18,21.22)(0.48,-0.14){1}{\line(1,0){0.48}}
\multiput(180.7,21.36)(0.48,-0.14){1}{\line(1,0){0.48}}
\multiput(180.22,21.5)(0.48,-0.14){1}{\line(1,0){0.48}}
\multiput(179.74,21.63)(0.48,-0.14){1}{\line(1,0){0.48}}
\multiput(179.26,21.76)(0.48,-0.13){1}{\line(1,0){0.48}}
\multiput(178.77,21.89)(0.48,-0.13){1}{\line(1,0){0.48}}
\multiput(178.29,22.02)(0.48,-0.13){1}{\line(1,0){0.48}}
\multiput(177.8,22.15)(0.49,-0.12){1}{\line(1,0){0.49}}
\multiput(177.32,22.27)(0.49,-0.12){1}{\line(1,0){0.49}}
\multiput(176.83,22.39)(0.49,-0.12){1}{\line(1,0){0.49}}
\multiput(176.34,22.5)(0.49,-0.12){1}{\line(1,0){0.49}}
\multiput(175.86,22.62)(0.49,-0.11){1}{\line(1,0){0.49}}
\multiput(175.37,22.73)(0.49,-0.11){1}{\line(1,0){0.49}}
\multiput(174.88,22.84)(0.49,-0.11){1}{\line(1,0){0.49}}
\multiput(174.39,22.95)(0.49,-0.11){1}{\line(1,0){0.49}}
\multiput(173.9,23.05)(0.49,-0.1){1}{\line(1,0){0.49}}
\multiput(173.41,23.15)(0.49,-0.1){1}{\line(1,0){0.49}}
\multiput(172.92,23.25)(0.49,-0.1){1}{\line(1,0){0.49}}
\multiput(172.43,23.34)(0.49,-0.1){1}{\line(1,0){0.49}}
\multiput(171.93,23.44)(0.49,-0.09){1}{\line(1,0){0.49}}
\multiput(171.44,23.53)(0.49,-0.09){1}{\line(1,0){0.49}}
\multiput(170.95,23.61)(0.49,-0.09){1}{\line(1,0){0.49}}
\multiput(170.45,23.7)(0.49,-0.09){1}{\line(1,0){0.49}}
\multiput(169.96,23.78)(0.49,-0.08){1}{\line(1,0){0.49}}
\multiput(169.47,23.86)(0.49,-0.08){1}{\line(1,0){0.49}}
\multiput(168.97,23.94)(0.49,-0.08){1}{\line(1,0){0.49}}
\multiput(168.48,24.01)(0.5,-0.07){1}{\line(1,0){0.5}}
\multiput(167.98,24.08)(0.5,-0.07){1}{\line(1,0){0.5}}
\multiput(167.48,24.15)(0.5,-0.07){1}{\line(1,0){0.5}}
\multiput(166.99,24.22)(0.5,-0.07){1}{\line(1,0){0.5}}
\multiput(166.49,24.28)(0.5,-0.06){1}{\line(1,0){0.5}}
\multiput(165.99,24.34)(0.5,-0.06){1}{\line(1,0){0.5}}
\multiput(165.5,24.4)(0.5,-0.06){1}{\line(1,0){0.5}}
\multiput(165,24.46)(0.5,-0.06){1}{\line(1,0){0.5}}
\multiput(164.5,24.51)(0.5,-0.05){1}{\line(1,0){0.5}}
\multiput(164,24.56)(0.5,-0.05){1}{\line(1,0){0.5}}
\multiput(163.5,24.61)(0.5,-0.05){1}{\line(1,0){0.5}}
\multiput(163,24.65)(0.5,-0.04){1}{\line(1,0){0.5}}
\multiput(162.5,24.7)(0.5,-0.04){1}{\line(1,0){0.5}}
\multiput(162.01,24.73)(0.5,-0.04){1}{\line(1,0){0.5}}
\multiput(161.51,24.77)(0.5,-0.04){1}{\line(1,0){0.5}}
\multiput(161.01,24.8)(0.5,-0.03){1}{\line(1,0){0.5}}
\multiput(160.51,24.84)(0.5,-0.03){1}{\line(1,0){0.5}}
\multiput(160.01,24.86)(0.5,-0.03){1}{\line(1,0){0.5}}
\multiput(159.51,24.89)(0.5,-0.03){1}{\line(1,0){0.5}}
\multiput(159.01,24.91)(0.5,-0.02){1}{\line(1,0){0.5}}
\multiput(158.51,24.93)(0.5,-0.02){1}{\line(1,0){0.5}}
\multiput(158,24.95)(0.5,-0.02){1}{\line(1,0){0.5}}
\multiput(157.5,24.97)(0.5,-0.01){1}{\line(1,0){0.5}}
\multiput(157,24.98)(0.5,-0.01){1}{\line(1,0){0.5}}
\multiput(156.5,24.99)(0.5,-0.01){1}{\line(1,0){0.5}}
\multiput(156,24.99)(0.5,-0.01){1}{\line(1,0){0.5}}
\multiput(155.5,25)(0.5,-0){1}{\line(1,0){0.5}}
\multiput(155,25)(0.5,-0){1}{\line(1,0){0.5}}
\multiput(154.5,25)(0.5,0){1}{\line(1,0){0.5}}
\multiput(154,24.99)(0.5,0){1}{\line(1,0){0.5}}
\multiput(153.5,24.99)(0.5,0.01){1}{\line(1,0){0.5}}
\multiput(153,24.98)(0.5,0.01){1}{\line(1,0){0.5}}
\multiput(152.5,24.97)(0.5,0.01){1}{\line(1,0){0.5}}
\multiput(152,24.95)(0.5,0.01){1}{\line(1,0){0.5}}
\multiput(151.49,24.93)(0.5,0.02){1}{\line(1,0){0.5}}
\multiput(150.99,24.91)(0.5,0.02){1}{\line(1,0){0.5}}
\multiput(150.49,24.89)(0.5,0.02){1}{\line(1,0){0.5}}
\multiput(149.99,24.86)(0.5,0.03){1}{\line(1,0){0.5}}
\multiput(149.49,24.84)(0.5,0.03){1}{\line(1,0){0.5}}
\multiput(148.99,24.8)(0.5,0.03){1}{\line(1,0){0.5}}
\multiput(148.49,24.77)(0.5,0.03){1}{\line(1,0){0.5}}
\multiput(147.99,24.73)(0.5,0.04){1}{\line(1,0){0.5}}
\multiput(147.5,24.7)(0.5,0.04){1}{\line(1,0){0.5}}
\multiput(147,24.65)(0.5,0.04){1}{\line(1,0){0.5}}
\multiput(146.5,24.61)(0.5,0.04){1}{\line(1,0){0.5}}
\multiput(146,24.56)(0.5,0.05){1}{\line(1,0){0.5}}
\multiput(145.5,24.51)(0.5,0.05){1}{\line(1,0){0.5}}
\multiput(145,24.46)(0.5,0.05){1}{\line(1,0){0.5}}
\multiput(144.5,24.4)(0.5,0.06){1}{\line(1,0){0.5}}
\multiput(144.01,24.34)(0.5,0.06){1}{\line(1,0){0.5}}
\multiput(143.51,24.28)(0.5,0.06){1}{\line(1,0){0.5}}
\multiput(143.01,24.22)(0.5,0.06){1}{\line(1,0){0.5}}
\multiput(142.52,24.15)(0.5,0.07){1}{\line(1,0){0.5}}
\multiput(142.02,24.08)(0.5,0.07){1}{\line(1,0){0.5}}
\multiput(141.52,24.01)(0.5,0.07){1}{\line(1,0){0.5}}
\multiput(141.03,23.94)(0.5,0.07){1}{\line(1,0){0.5}}
\multiput(140.53,23.86)(0.49,0.08){1}{\line(1,0){0.49}}
\multiput(140.04,23.78)(0.49,0.08){1}{\line(1,0){0.49}}
\multiput(139.55,23.7)(0.49,0.08){1}{\line(1,0){0.49}}
\multiput(139.05,23.61)(0.49,0.09){1}{\line(1,0){0.49}}
\multiput(138.56,23.53)(0.49,0.09){1}{\line(1,0){0.49}}
\multiput(138.07,23.44)(0.49,0.09){1}{\line(1,0){0.49}}
\multiput(137.57,23.34)(0.49,0.09){1}{\line(1,0){0.49}}
\multiput(137.08,23.25)(0.49,0.1){1}{\line(1,0){0.49}}
\multiput(136.59,23.15)(0.49,0.1){1}{\line(1,0){0.49}}
\multiput(136.1,23.05)(0.49,0.1){1}{\line(1,0){0.49}}
\multiput(135.61,22.95)(0.49,0.1){1}{\line(1,0){0.49}}
\multiput(135.12,22.84)(0.49,0.11){1}{\line(1,0){0.49}}
\multiput(134.63,22.73)(0.49,0.11){1}{\line(1,0){0.49}}
\multiput(134.14,22.62)(0.49,0.11){1}{\line(1,0){0.49}}
\multiput(133.66,22.5)(0.49,0.11){1}{\line(1,0){0.49}}
\multiput(133.17,22.39)(0.49,0.12){1}{\line(1,0){0.49}}
\multiput(132.68,22.27)(0.49,0.12){1}{\line(1,0){0.49}}
\multiput(132.2,22.15)(0.49,0.12){1}{\line(1,0){0.49}}
\multiput(131.71,22.02)(0.49,0.12){1}{\line(1,0){0.49}}
\multiput(131.23,21.89)(0.48,0.13){1}{\line(1,0){0.48}}
\multiput(130.74,21.76)(0.48,0.13){1}{\line(1,0){0.48}}
\multiput(130.26,21.63)(0.48,0.13){1}{\line(1,0){0.48}}
\multiput(129.78,21.5)(0.48,0.14){1}{\line(1,0){0.48}}
\multiput(129.3,21.36)(0.48,0.14){1}{\line(1,0){0.48}}
\multiput(128.82,21.22)(0.48,0.14){1}{\line(1,0){0.48}}
\multiput(128.34,21.07)(0.48,0.14){1}{\line(1,0){0.48}}
\multiput(127.86,20.93)(0.48,0.15){1}{\line(1,0){0.48}}
\multiput(127.38,20.78)(0.48,0.15){1}{\line(1,0){0.48}}
\multiput(126.9,20.63)(0.48,0.15){1}{\line(1,0){0.48}}
\multiput(126.43,20.48)(0.48,0.15){1}{\line(1,0){0.48}}
\multiput(125.95,20.32)(0.48,0.16){1}{\line(1,0){0.48}}
\multiput(125.47,20.16)(0.48,0.16){1}{\line(1,0){0.48}}
\multiput(125,20)(0.47,0.16){1}{\line(1,0){0.47}}

\linethickness{0.3mm}
\put(185,85){\circle{10}}

\put(15,85){\makebox(0,0)[cc]{7}}

\put(35,85){\makebox(0,0)[cc]{8}}

\put(55,85){\makebox(0,0)[cc]{12}}

\put(75,85){\makebox(0,0)[cc]{5}}

\linethickness{0.3mm}
\put(50,90){\line(1,0){10}}
\put(50,80){\line(0,1){10}}
\put(60,80){\line(0,1){10}}
\put(50,80){\line(1,0){10}}
\put(55,75){\makebox(0,0)[cc]{$syncDown$}}

\put(125,85){\makebox(0,0)[cc]{7}}

\put(145,85){\makebox(0,0)[cc]{8}}

\put(165,85){\makebox(0,0)[cc]{8}}

\put(185,85){\makebox(0,0)[cc]{5}}

\linethickness{0.3mm}
\put(160,90){\line(1,0){10}}
\put(160,80){\line(0,1){10}}
\put(170,80){\line(0,1){10}}
\put(160,80){\line(1,0){10}}
\put(125,50){\makebox(0,0)[cc]{7}}

\put(145,50){\makebox(0,0)[cc]{8}}

\put(165,50){\makebox(0,0)[cc]{7}}

\put(185,50){\makebox(0,0)[cc]{5}}

\put(165,75){\makebox(0,0)[cc]{$middleLeftDown$}}

\linethickness{0.3mm}
\put(140,55){\line(1,0){10}}
\put(140,45){\line(0,1){10}}
\put(150,45){\line(0,1){10}}
\put(140,45){\line(1,0){10}}
\put(145,40){\makebox(0,0)[cc]{$middleLeftDown$}}

\put(75,50){\makebox(0,0)[cc]{5}}

\put(55,50){\makebox(0,0)[cc]{7}}

\put(35,50){\makebox(0,0)[cc]{7}}

\put(15,50){\makebox(0,0)[cc]{7}}

\linethickness{0.3mm}
\put(10,55){\line(1,0){10}}
\put(10,45){\line(0,1){10}}
\put(20,45){\line(0,1){10}}
\put(10,45){\line(1,0){10}}
\put(15,40){\makebox(0,0)[cc]{$middleRightDown$}}

\put(15,15){\makebox(0,0)[cc]{6}}

\put(35,15){\makebox(0,0)[cc]{7}}

\put(55,15){\makebox(0,0)[cc]{7}}

\put(75,15){\makebox(0,0)[cc]{5}}

\linethickness{0.3mm}
\put(70,20){\line(1,0){10}}
\put(70,10){\line(0,1){10}}
\put(80,10){\line(0,1){10}}
\put(70,10){\line(1,0){10}}
\put(75,5){\makebox(0,0)[cc]{$middleLeftUp$}}

\put(125,15){\makebox(0,0)[cc]{6}}

\put(145,15){\makebox(0,0)[cc]{7}}

\put(165,15){\makebox(0,0)[cc]{7}}

\put(185,15){\makebox(0,0)[cc]{6}}

\end{picture}

%% file: ExMiddleClosure.tex
\ifx\JPicScale\undefined\def\JPicScale{0.5}\fi
\unitlength \JPicScale mm
\begin{picture}(50.78,20)(0,0)
\linethickness{0.3mm}
\put(5,14.07){\circle{10}}

\linethickness{0.3mm}
\put(25,14.07){\circle{10}}

\linethickness{0.3mm}
\put(45,14.07){\circle{10}}

\linethickness{0.3mm}
\put(10,14.07){\line(1,0){10}}
\linethickness{0.3mm}
\put(30,14.07){\line(1,0){10}}
\put(5,14.07){\makebox(0,0)[cc]{l}}

\put(25,14.07){\makebox(0,0)[cc]{p}}

\put(45,14.07){\makebox(0,0)[cc]{r}}

\put(5,4.07){\makebox(0,0)[cc]{$a$}}

\put(25,4.07){\makebox(0,0)[cc]{$b$}}

\put(45,4.07){\makebox(0,0)[cc]{$c$}}

\linethickness{0.3mm}
\put(50.78,13.82){\line(0,1){0.49}}
\multiput(50.74,13.34)(0.04,0.49){1}{\line(0,1){0.49}}
\multiput(50.66,12.85)(0.08,0.48){1}{\line(0,1){0.48}}
\multiput(50.54,12.38)(0.12,0.48){1}{\line(0,1){0.48}}
\multiput(50.38,11.91)(0.16,0.46){1}{\line(0,1){0.46}}
\multiput(50.19,11.47)(0.1,0.22){2}{\line(0,1){0.22}}
\multiput(49.97,11.03)(0.11,0.22){2}{\line(0,1){0.22}}
\multiput(49.7,10.62)(0.13,0.21){2}{\line(0,1){0.21}}
\multiput(49.41,10.24)(0.15,0.19){2}{\line(0,1){0.19}}
\multiput(49.09,9.88)(0.11,0.12){3}{\line(0,1){0.12}}
\multiput(48.74,9.54)(0.12,0.11){3}{\line(1,0){0.12}}
\multiput(48.36,9.24)(0.13,0.1){3}{\line(1,0){0.13}}
\multiput(47.96,8.98)(0.2,0.13){2}{\line(1,0){0.2}}
\multiput(47.54,8.74)(0.21,0.12){2}{\line(1,0){0.21}}
\multiput(47.1,8.55)(0.22,0.1){2}{\line(1,0){0.22}}
\multiput(46.65,8.39)(0.45,0.16){1}{\line(1,0){0.45}}
\multiput(46.19,8.27)(0.46,0.12){1}{\line(1,0){0.46}}
\multiput(45.71,8.19)(0.47,0.08){1}{\line(1,0){0.47}}
\multiput(45.24,8.15)(0.48,0.04){1}{\line(1,0){0.48}}
\put(44.76,8.15){\line(1,0){0.48}}
\multiput(44.29,8.19)(0.48,-0.04){1}{\line(1,0){0.48}}
\multiput(43.81,8.27)(0.47,-0.08){1}{\line(1,0){0.47}}
\multiput(43.35,8.39)(0.46,-0.12){1}{\line(1,0){0.46}}
\multiput(42.9,8.55)(0.45,-0.16){1}{\line(1,0){0.45}}
\multiput(42.46,8.74)(0.22,-0.1){2}{\line(1,0){0.22}}
\multiput(42.04,8.98)(0.21,-0.12){2}{\line(1,0){0.21}}
\multiput(41.64,9.24)(0.2,-0.13){2}{\line(1,0){0.2}}
\multiput(41.26,9.54)(0.13,-0.1){3}{\line(1,0){0.13}}
\multiput(40.91,9.88)(0.12,-0.11){3}{\line(1,0){0.12}}
\multiput(40.59,10.24)(0.11,-0.12){3}{\line(0,-1){0.12}}
\multiput(40.3,10.62)(0.15,-0.19){2}{\line(0,-1){0.19}}
\multiput(40.03,11.03)(0.13,-0.21){2}{\line(0,-1){0.21}}
\multiput(39.81,11.47)(0.11,-0.22){2}{\line(0,-1){0.22}}
\multiput(39.62,11.91)(0.1,-0.22){2}{\line(0,-1){0.22}}
\multiput(39.46,12.38)(0.16,-0.46){1}{\line(0,-1){0.46}}
\multiput(39.34,12.85)(0.12,-0.48){1}{\line(0,-1){0.48}}
\multiput(39.26,13.34)(0.08,-0.48){1}{\line(0,-1){0.48}}
\multiput(39.22,13.82)(0.04,-0.49){1}{\line(0,-1){0.49}}
\put(39.22,13.82){\line(0,1){0.49}}
\multiput(39.22,14.32)(0.04,0.49){1}{\line(0,1){0.49}}
\multiput(39.26,14.8)(0.08,0.48){1}{\line(0,1){0.48}}
\multiput(39.34,15.29)(0.12,0.48){1}{\line(0,1){0.48}}
\multiput(39.46,15.76)(0.16,0.46){1}{\line(0,1){0.46}}
\multiput(39.62,16.23)(0.1,0.22){2}{\line(0,1){0.22}}
\multiput(39.81,16.67)(0.11,0.22){2}{\line(0,1){0.22}}
\multiput(40.03,17.11)(0.13,0.21){2}{\line(0,1){0.21}}
\multiput(40.3,17.52)(0.15,0.19){2}{\line(0,1){0.19}}
\multiput(40.59,17.9)(0.11,0.12){3}{\line(0,1){0.12}}
\multiput(40.91,18.26)(0.12,0.11){3}{\line(1,0){0.12}}
\multiput(41.26,18.6)(0.13,0.1){3}{\line(1,0){0.13}}
\multiput(41.64,18.9)(0.2,0.13){2}{\line(1,0){0.2}}
\multiput(42.04,19.16)(0.21,0.12){2}{\line(1,0){0.21}}
\multiput(42.46,19.4)(0.22,0.1){2}{\line(1,0){0.22}}
\multiput(42.9,19.59)(0.45,0.16){1}{\line(1,0){0.45}}
\multiput(43.35,19.75)(0.46,0.12){1}{\line(1,0){0.46}}
\multiput(43.81,19.87)(0.47,0.08){1}{\line(1,0){0.47}}
\multiput(44.29,19.95)(0.48,0.04){1}{\line(1,0){0.48}}
\put(44.76,19.99){\line(1,0){0.48}}
\multiput(45.24,19.99)(0.48,-0.04){1}{\line(1,0){0.48}}
\multiput(45.71,19.95)(0.47,-0.08){1}{\line(1,0){0.47}}
\multiput(46.19,19.87)(0.46,-0.12){1}{\line(1,0){0.46}}
\multiput(46.65,19.75)(0.45,-0.16){1}{\line(1,0){0.45}}
\multiput(47.1,19.59)(0.22,-0.1){2}{\line(1,0){0.22}}
\multiput(47.54,19.4)(0.21,-0.12){2}{\line(1,0){0.21}}
\multiput(47.96,19.16)(0.2,-0.13){2}{\line(1,0){0.2}}
\multiput(48.36,18.9)(0.13,-0.1){3}{\line(1,0){0.13}}
\multiput(48.74,18.6)(0.12,-0.11){3}{\line(1,0){0.12}}
\multiput(49.09,18.26)(0.11,-0.12){3}{\line(0,-1){0.12}}
\multiput(49.41,17.9)(0.15,-0.19){2}{\line(0,-1){0.19}}
\multiput(49.7,17.52)(0.13,-0.21){2}{\line(0,-1){0.21}}
\multiput(49.97,17.11)(0.11,-0.22){2}{\line(0,-1){0.22}}
\multiput(50.19,16.67)(0.1,-0.22){2}{\line(0,-1){0.22}}
\multiput(50.38,16.23)(0.16,-0.46){1}{\line(0,-1){0.46}}
\multiput(50.54,15.76)(0.12,-0.48){1}{\line(0,-1){0.48}}
\multiput(50.66,15.29)(0.08,-0.48){1}{\line(0,-1){0.48}}
\multiput(50.74,14.8)(0.04,-0.49){1}{\line(0,-1){0.49}}

\end{picture}

%% file: ExMiddleActivated.tex
\ifx\JPicScale\undefined\def\JPicScale{0.5}\fi
\unitlength \JPicScale mm
\begin{picture}(75.78,20)(0,0)
\linethickness{0.3mm}
\put(30,15){\circle{10}}

\linethickness{0.3mm}
\put(50,15){\circle{10}}

\linethickness{0.3mm}
\put(70,14.07){\circle{10}}

\linethickness{0.3mm}
\put(35,14.07){\line(1,0){10}}
\linethickness{0.3mm}
\put(55,14.07){\line(1,0){10}}
\put(30,15){\makebox(0,0)[cc]{p}}

\put(70,14.07){\makebox(0,0)[cc]{r}}

\put(10,5){\makebox(0,0)[cc]{$a$}}

\put(30,5){\makebox(0,0)[cc]{$b$}}

\put(50,5){\makebox(0,0)[cc]{$c$}}

\linethickness{0.3mm}
\put(75.78,13.82){\line(0,1){0.49}}
\multiput(75.74,13.34)(0.04,0.49){1}{\line(0,1){0.49}}
\multiput(75.66,12.85)(0.08,0.48){1}{\line(0,1){0.48}}
\multiput(75.54,12.38)(0.12,0.48){1}{\line(0,1){0.48}}
\multiput(75.38,11.91)(0.16,0.46){1}{\line(0,1){0.46}}
\multiput(75.19,11.47)(0.1,0.22){2}{\line(0,1){0.22}}
\multiput(74.97,11.03)(0.11,0.22){2}{\line(0,1){0.22}}
\multiput(74.7,10.62)(0.13,0.21){2}{\line(0,1){0.21}}
\multiput(74.41,10.24)(0.15,0.19){2}{\line(0,1){0.19}}
\multiput(74.09,9.88)(0.11,0.12){3}{\line(0,1){0.12}}
\multiput(73.74,9.54)(0.12,0.11){3}{\line(1,0){0.12}}
\multiput(73.36,9.24)(0.13,0.1){3}{\line(1,0){0.13}}
\multiput(72.96,8.98)(0.2,0.13){2}{\line(1,0){0.2}}
\multiput(72.54,8.74)(0.21,0.12){2}{\line(1,0){0.21}}
\multiput(72.1,8.55)(0.22,0.1){2}{\line(1,0){0.22}}
\multiput(71.65,8.39)(0.45,0.16){1}{\line(1,0){0.45}}
\multiput(71.19,8.27)(0.46,0.12){1}{\line(1,0){0.46}}
\multiput(70.71,8.19)(0.47,0.08){1}{\line(1,0){0.47}}
\multiput(70.24,8.15)(0.48,0.04){1}{\line(1,0){0.48}}
\put(69.76,8.15){\line(1,0){0.48}}
\multiput(69.29,8.19)(0.48,-0.04){1}{\line(1,0){0.48}}
\multiput(68.81,8.27)(0.47,-0.08){1}{\line(1,0){0.47}}
\multiput(68.35,8.39)(0.46,-0.12){1}{\line(1,0){0.46}}
\multiput(67.9,8.55)(0.45,-0.16){1}{\line(1,0){0.45}}
\multiput(67.46,8.74)(0.22,-0.1){2}{\line(1,0){0.22}}
\multiput(67.04,8.98)(0.21,-0.12){2}{\line(1,0){0.21}}
\multiput(66.64,9.24)(0.2,-0.13){2}{\line(1,0){0.2}}
\multiput(66.26,9.54)(0.13,-0.1){3}{\line(1,0){0.13}}
\multiput(65.91,9.88)(0.12,-0.11){3}{\line(1,0){0.12}}
\multiput(65.59,10.24)(0.11,-0.12){3}{\line(0,-1){0.12}}
\multiput(65.3,10.62)(0.15,-0.19){2}{\line(0,-1){0.19}}
\multiput(65.03,11.03)(0.13,-0.21){2}{\line(0,-1){0.21}}
\multiput(64.81,11.47)(0.11,-0.22){2}{\line(0,-1){0.22}}
\multiput(64.62,11.91)(0.1,-0.22){2}{\line(0,-1){0.22}}
\multiput(64.46,12.38)(0.16,-0.46){1}{\line(0,-1){0.46}}
\multiput(64.34,12.85)(0.12,-0.48){1}{\line(0,-1){0.48}}
\multiput(64.26,13.34)(0.08,-0.48){1}{\line(0,-1){0.48}}
\multiput(64.22,13.82)(0.04,-0.49){1}{\line(0,-1){0.49}}
\put(64.22,13.82){\line(0,1){0.49}}
\multiput(64.22,14.32)(0.04,0.49){1}{\line(0,1){0.49}}
\multiput(64.26,14.8)(0.08,0.48){1}{\line(0,1){0.48}}
\multiput(64.34,15.29)(0.12,0.48){1}{\line(0,1){0.48}}
\multiput(64.46,15.76)(0.16,0.46){1}{\line(0,1){0.46}}
\multiput(64.62,16.23)(0.1,0.22){2}{\line(0,1){0.22}}
\multiput(64.81,16.67)(0.11,0.22){2}{\line(0,1){0.22}}
\multiput(65.03,17.11)(0.13,0.21){2}{\line(0,1){0.21}}
\multiput(65.3,17.52)(0.15,0.19){2}{\line(0,1){0.19}}
\multiput(65.59,17.9)(0.11,0.12){3}{\line(0,1){0.12}}
\multiput(65.91,18.26)(0.12,0.11){3}{\line(1,0){0.12}}
\multiput(66.26,18.6)(0.13,0.1){3}{\line(1,0){0.13}}
\multiput(66.64,18.9)(0.2,0.13){2}{\line(1,0){0.2}}
\multiput(67.04,19.16)(0.21,0.12){2}{\line(1,0){0.21}}
\multiput(67.46,19.4)(0.22,0.1){2}{\line(1,0){0.22}}
\multiput(67.9,19.59)(0.45,0.16){1}{\line(1,0){0.45}}
\multiput(68.35,19.75)(0.46,0.12){1}{\line(1,0){0.46}}
\multiput(68.81,19.87)(0.47,0.08){1}{\line(1,0){0.47}}
\multiput(69.29,19.95)(0.48,0.04){1}{\line(1,0){0.48}}
\put(69.76,19.99){\line(1,0){0.48}}
\multiput(70.24,19.99)(0.48,-0.04){1}{\line(1,0){0.48}}
\multiput(70.71,19.95)(0.47,-0.08){1}{\line(1,0){0.47}}
\multiput(71.19,19.87)(0.46,-0.12){1}{\line(1,0){0.46}}
\multiput(71.65,19.75)(0.45,-0.16){1}{\line(1,0){0.45}}
\multiput(72.1,19.59)(0.22,-0.1){2}{\line(1,0){0.22}}
\multiput(72.54,19.4)(0.21,-0.12){2}{\line(1,0){0.21}}
\multiput(72.96,19.16)(0.2,-0.13){2}{\line(1,0){0.2}}
\multiput(73.36,18.9)(0.13,-0.1){3}{\line(1,0){0.13}}
\multiput(73.74,18.6)(0.12,-0.11){3}{\line(1,0){0.12}}
\multiput(74.09,18.26)(0.11,-0.12){3}{\line(0,-1){0.12}}
\multiput(74.41,17.9)(0.15,-0.19){2}{\line(0,-1){0.19}}
\multiput(74.7,17.52)(0.13,-0.21){2}{\line(0,-1){0.21}}
\multiput(74.97,17.11)(0.11,-0.22){2}{\line(0,-1){0.22}}
\multiput(75.19,16.67)(0.1,-0.22){2}{\line(0,-1){0.22}}
\multiput(75.38,16.23)(0.16,-0.46){1}{\line(0,-1){0.46}}
\multiput(75.54,15.76)(0.12,-0.48){1}{\line(0,-1){0.48}}
\multiput(75.66,15.29)(0.08,-0.48){1}{\line(0,-1){0.48}}
\multiput(75.74,14.8)(0.04,-0.49){1}{\line(0,-1){0.49}}

\linethickness{0.3mm}
\put(10,15){\circle{10}}

\linethickness{0.3mm}
\put(15,15){\line(1,0){10}}
\put(10,15){\makebox(0,0)[cc]{l}}

\put(50,15){\makebox(0,0)[cc]{q}}

\put(70,5){\makebox(0,0)[cc]{$c-1$}}

\end{picture}

%% file: ExEndEnabled.tex
%%Created by jPicEdt 1.4.1_03: mixed JPIC-XML/LaTeX format
%%Fri Nov 20 09:39:23 CET 2009
%%Begin JPIC-XML
%<?xml version="1.0" standalone="yes"?>
%<jpic x-min="10" x-max="40.78" y-min="5" y-max="20.93" auto-bounding="true">
%<ellipse p3= "(20,10)"
%	 p2= "(20,20)"
%	 p1= "(10,20)"
%	 fill-style= "none"
%	 closure= "open"
%	 angle-end= "360"
%	 angle-start= "0"
%	 />
%<ellipse p3= "(40,10)"
%	 p2= "(40,20)"
%	 p1= "(30,20)"
%	 fill-style= "none"
%	 closure= "open"
%	 angle-end= "360"
%	 angle-start= "0"
%	 />
%<multicurve fill-style= "none"
%	 points= "(20,15);(20,15);(30,15);(30,15)"
%	 />
%<text text-vert-align= "center-v"
%	 fill-style= "none"
%	 anchor-point= "(15,15)"
%	 text-frame= "noframe"
%	 text-hor-align= "center-h"
%	 >
%p
%</text>
%<text text-vert-align= "center-v"
%	 fill-style= "none"
%	 anchor-point= "(35,15)"
%	 text-frame= "noframe"
%	 text-hor-align= "center-h"
%	 >
%r
%</text>
%<ellipse p3= "(40.78,9.07)"
%	 p2= "(40.78,20.93)"
%	 p1= "(29.22,20.93)"
%	 fill-style= "none"
%	 closure= "open"
%	 angle-end= "0"
%	 angle-start= "0"
%	 />
%<text text-vert-align= "center-v"
%	 fill-style= "none"
%	 anchor-point= "(15,5)"
%	 text-frame= "noframe"
%	 text-hor-align= "center-h"
%	 >
%$a$
%</text>
%<text text-vert-align= "center-v"
%	 fill-style= "none"
%	 anchor-point= "(35,5)"
%	 text-frame= "noframe"
%	 text-hor-align= "center-h"
%	 >
%$b$
%</text>
%</jpic>
%%End JPIC-XML
%LaTeX-picture environment using emulated lines and arcs
%You can rescale the whole picture (to 80% for instance) by using the command \def\JPicScale{0.8}
\ifx\JPicScale\undefined\def\JPicScale{0.5}\fi
\unitlength \JPicScale mm
\begin{picture}(40.78,20.93)(0,0)
\linethickness{0.3mm}
\put(15,15){\circle{10}}

\linethickness{0.3mm}
\put(35,15){\circle{10}}

\linethickness{0.3mm}
\put(20,15){\line(1,0){10}}
\put(15,15){\makebox(0,0)[cc]{p}}

\put(35,15){\makebox(0,0)[cc]{r}}

\linethickness{0.3mm}
\put(40.78,14.75){\line(0,1){0.49}}
\multiput(40.74,14.27)(0.04,0.49){1}{\line(0,1){0.49}}
\multiput(40.66,13.78)(0.08,0.48){1}{\line(0,1){0.48}}
\multiput(40.54,13.31)(0.12,0.48){1}{\line(0,1){0.48}}
\multiput(40.38,12.84)(0.16,0.46){1}{\line(0,1){0.46}}
\multiput(40.19,12.4)(0.1,0.22){2}{\line(0,1){0.22}}
\multiput(39.97,11.96)(0.11,0.22){2}{\line(0,1){0.22}}
\multiput(39.7,11.55)(0.13,0.21){2}{\line(0,1){0.21}}
\multiput(39.41,11.17)(0.15,0.19){2}{\line(0,1){0.19}}
\multiput(39.09,10.81)(0.11,0.12){3}{\line(0,1){0.12}}
\multiput(38.74,10.47)(0.12,0.11){3}{\line(1,0){0.12}}
\multiput(38.36,10.17)(0.13,0.1){3}{\line(1,0){0.13}}
\multiput(37.96,9.91)(0.2,0.13){2}{\line(1,0){0.2}}
\multiput(37.54,9.67)(0.21,0.12){2}{\line(1,0){0.21}}
\multiput(37.1,9.48)(0.22,0.1){2}{\line(1,0){0.22}}
\multiput(36.65,9.32)(0.45,0.16){1}{\line(1,0){0.45}}
\multiput(36.19,9.2)(0.46,0.12){1}{\line(1,0){0.46}}
\multiput(35.71,9.12)(0.47,0.08){1}{\line(1,0){0.47}}
\multiput(35.24,9.08)(0.48,0.04){1}{\line(1,0){0.48}}
\put(34.76,9.08){\line(1,0){0.48}}
\multiput(34.29,9.12)(0.48,-0.04){1}{\line(1,0){0.48}}
\multiput(33.81,9.2)(0.47,-0.08){1}{\line(1,0){0.47}}
\multiput(33.35,9.32)(0.46,-0.12){1}{\line(1,0){0.46}}
\multiput(32.9,9.48)(0.45,-0.16){1}{\line(1,0){0.45}}
\multiput(32.46,9.67)(0.22,-0.1){2}{\line(1,0){0.22}}
\multiput(32.04,9.91)(0.21,-0.12){2}{\line(1,0){0.21}}
\multiput(31.64,10.17)(0.2,-0.13){2}{\line(1,0){0.2}}
\multiput(31.26,10.47)(0.13,-0.1){3}{\line(1,0){0.13}}
\multiput(30.91,10.81)(0.12,-0.11){3}{\line(1,0){0.12}}
\multiput(30.59,11.17)(0.11,-0.12){3}{\line(0,-1){0.12}}
\multiput(30.3,11.55)(0.15,-0.19){2}{\line(0,-1){0.19}}
\multiput(30.03,11.96)(0.13,-0.21){2}{\line(0,-1){0.21}}
\multiput(29.81,12.4)(0.11,-0.22){2}{\line(0,-1){0.22}}
\multiput(29.62,12.84)(0.1,-0.22){2}{\line(0,-1){0.22}}
\multiput(29.46,13.31)(0.16,-0.46){1}{\line(0,-1){0.46}}
\multiput(29.34,13.78)(0.12,-0.48){1}{\line(0,-1){0.48}}
\multiput(29.26,14.27)(0.08,-0.48){1}{\line(0,-1){0.48}}
\multiput(29.22,14.75)(0.04,-0.49){1}{\line(0,-1){0.49}}
\put(29.22,14.75){\line(0,1){0.49}}
\multiput(29.22,15.25)(0.04,0.49){1}{\line(0,1){0.49}}
\multiput(29.26,15.73)(0.08,0.48){1}{\line(0,1){0.48}}
\multiput(29.34,16.22)(0.12,0.48){1}{\line(0,1){0.48}}
\multiput(29.46,16.69)(0.16,0.46){1}{\line(0,1){0.46}}
\multiput(29.62,17.16)(0.1,0.22){2}{\line(0,1){0.22}}
\multiput(29.81,17.6)(0.11,0.22){2}{\line(0,1){0.22}}
\multiput(30.03,18.04)(0.13,0.21){2}{\line(0,1){0.21}}
\multiput(30.3,18.45)(0.15,0.19){2}{\line(0,1){0.19}}
\multiput(30.59,18.83)(0.11,0.12){3}{\line(0,1){0.12}}
\multiput(30.91,19.19)(0.12,0.11){3}{\line(1,0){0.12}}
\multiput(31.26,19.53)(0.13,0.1){3}{\line(1,0){0.13}}
\multiput(31.64,19.83)(0.2,0.13){2}{\line(1,0){0.2}}
\multiput(32.04,20.09)(0.21,0.12){2}{\line(1,0){0.21}}
\multiput(32.46,20.33)(0.22,0.1){2}{\line(1,0){0.22}}
\multiput(32.9,20.52)(0.45,0.16){1}{\line(1,0){0.45}}
\multiput(33.35,20.68)(0.46,0.12){1}{\line(1,0){0.46}}
\multiput(33.81,20.8)(0.47,0.08){1}{\line(1,0){0.47}}
\multiput(34.29,20.88)(0.48,0.04){1}{\line(1,0){0.48}}
\put(34.76,20.92){\line(1,0){0.48}}
\multiput(35.24,20.92)(0.48,-0.04){1}{\line(1,0){0.48}}
\multiput(35.71,20.88)(0.47,-0.08){1}{\line(1,0){0.47}}
\multiput(36.19,20.8)(0.46,-0.12){1}{\line(1,0){0.46}}
\multiput(36.65,20.68)(0.45,-0.16){1}{\line(1,0){0.45}}
\multiput(37.1,20.52)(0.22,-0.1){2}{\line(1,0){0.22}}
\multiput(37.54,20.33)(0.21,-0.12){2}{\line(1,0){0.21}}
\multiput(37.96,20.09)(0.2,-0.13){2}{\line(1,0){0.2}}
\multiput(38.36,19.83)(0.13,-0.1){3}{\line(1,0){0.13}}
\multiput(38.74,19.53)(0.12,-0.11){3}{\line(1,0){0.12}}
\multiput(39.09,19.19)(0.11,-0.12){3}{\line(0,-1){0.12}}
\multiput(39.41,18.83)(0.15,-0.19){2}{\line(0,-1){0.19}}
\multiput(39.7,18.45)(0.13,-0.21){2}{\line(0,-1){0.21}}
\multiput(39.97,18.04)(0.11,-0.22){2}{\line(0,-1){0.22}}
\multiput(40.19,17.6)(0.1,-0.22){2}{\line(0,-1){0.22}}
\multiput(40.38,17.16)(0.16,-0.46){1}{\line(0,-1){0.46}}
\multiput(40.54,16.69)(0.12,-0.48){1}{\line(0,-1){0.48}}
\multiput(40.66,16.22)(0.08,-0.48){1}{\line(0,-1){0.48}}
\multiput(40.74,15.73)(0.04,-0.49){1}{\line(0,-1){0.49}}

\put(15,5){\makebox(0,0)[cc]{$a$}}

\put(35,5){\makebox(0,0)[cc]{$b$}}

\end{picture}

%% file: ExOmega.tex
\ifx\JPicScale\undefined\def\JPicScale{0.5}\fi
\unitlength \JPicScale mm
\begin{picture}(175,80.93)(0,0)
\linethickness{0.3mm}
\put(50,75){\circle{10}}

\linethickness{0.3mm}
\put(80,75){\circle{10}}

\linethickness{0.3mm}
\put(110,75){\circle{10}}

\linethickness{0.3mm}
\put(140,75){\circle{10}}

\linethickness{0.3mm}
\put(55,75){\line(1,0){20}}
\linethickness{0.3mm}
\put(85,75){\line(1,0){20}}
\linethickness{0.3mm}
\put(115,75){\line(1,0){20}}
\linethickness{0.3mm}
\put(50,45){\circle{10}}

\linethickness{0.3mm}
\put(80,45){\circle{10}}

\linethickness{0.3mm}
\put(110,45){\circle{10}}

\linethickness{0.3mm}
\put(140,45){\circle{10}}

\linethickness{0.3mm}
\put(55,45){\line(1,0){20}}
\linethickness{0.3mm}
\put(85,45){\line(1,0){20}}
\linethickness{0.3mm}
\put(115,45){\line(1,0){20}}
\linethickness{0.3mm}
\put(50,15){\circle{10}}

\linethickness{0.3mm}
\put(80,15){\circle{10}}

\linethickness{0.3mm}
\put(110,15){\circle{10}}

\linethickness{0.3mm}
\put(140,15){\circle{10}}

\linethickness{0.3mm}
\put(55,15){\line(1,0){20}}
\linethickness{0.3mm}
\put(85,15){\line(1,0){20}}
\linethickness{0.3mm}
\put(115,15){\line(1,0){20}}
\put(35,80){\makebox(0,0)[cc]{$s^0$}}

\put(35,50){\makebox(0,0)[cc]{$s^1$}}

\put(25,25){\makebox(0,0)[cc]{$s^i$ for}}

\put(25,15){\makebox(0,0)[cc]{$2\leq i\leq t$}}

\put(50,75){\makebox(0,0)[cc]{p}}

\put(50,45){\makebox(0,0)[cc]{}}

\put(50,15){\makebox(0,0)[cc]{}}

\put(80,75){\makebox(0,0)[cc]{q}}

\put(80,45){\makebox(0,0)[cc]{q}}

\put(80,15){\makebox(0,0)[cc]{q}}

\put(110,75){\makebox(0,0)[cc]{r}}

\put(110,45){\makebox(0,0)[cc]{}}

\put(110,15){\makebox(0,0)[cc]{r}}

\put(110,45){\makebox(0,0)[cc]{r}}

\put(50,45){\makebox(0,0)[cc]{p}}

\put(50,15){\makebox(0,0)[cc]{p}}

\put(140,75){\makebox(0,0)[cc]{s}}

\put(140,45){\makebox(0,0)[cc]{s}}

\put(140,15){\makebox(0,0)[cc]{s}}

\linethickness{0.3mm}
\put(145.78,74.75){\line(0,1){0.49}}
\multiput(145.74,74.27)(0.04,0.49){1}{\line(0,1){0.49}}
\multiput(145.66,73.78)(0.08,0.48){1}{\line(0,1){0.48}}
\multiput(145.54,73.31)(0.12,0.48){1}{\line(0,1){0.48}}
\multiput(145.38,72.84)(0.16,0.46){1}{\line(0,1){0.46}}
\multiput(145.19,72.4)(0.1,0.22){2}{\line(0,1){0.22}}
\multiput(144.97,71.96)(0.11,0.22){2}{\line(0,1){0.22}}
\multiput(144.7,71.55)(0.13,0.21){2}{\line(0,1){0.21}}
\multiput(144.41,71.17)(0.15,0.19){2}{\line(0,1){0.19}}
\multiput(144.09,70.81)(0.11,0.12){3}{\line(0,1){0.12}}
\multiput(143.74,70.47)(0.12,0.11){3}{\line(1,0){0.12}}
\multiput(143.36,70.17)(0.13,0.1){3}{\line(1,0){0.13}}
\multiput(142.96,69.91)(0.2,0.13){2}{\line(1,0){0.2}}
\multiput(142.54,69.67)(0.21,0.12){2}{\line(1,0){0.21}}
\multiput(142.1,69.48)(0.22,0.1){2}{\line(1,0){0.22}}
\multiput(141.65,69.32)(0.45,0.16){1}{\line(1,0){0.45}}
\multiput(141.19,69.2)(0.46,0.12){1}{\line(1,0){0.46}}
\multiput(140.71,69.12)(0.47,0.08){1}{\line(1,0){0.47}}
\multiput(140.24,69.08)(0.48,0.04){1}{\line(1,0){0.48}}
\put(139.76,69.08){\line(1,0){0.48}}
\multiput(139.29,69.12)(0.48,-0.04){1}{\line(1,0){0.48}}
\multiput(138.81,69.2)(0.47,-0.08){1}{\line(1,0){0.47}}
\multiput(138.35,69.32)(0.46,-0.12){1}{\line(1,0){0.46}}
\multiput(137.9,69.48)(0.45,-0.16){1}{\line(1,0){0.45}}
\multiput(137.46,69.67)(0.22,-0.1){2}{\line(1,0){0.22}}
\multiput(137.04,69.91)(0.21,-0.12){2}{\line(1,0){0.21}}
\multiput(136.64,70.17)(0.2,-0.13){2}{\line(1,0){0.2}}
\multiput(136.26,70.47)(0.13,-0.1){3}{\line(1,0){0.13}}
\multiput(135.91,70.81)(0.12,-0.11){3}{\line(1,0){0.12}}
\multiput(135.59,71.17)(0.11,-0.12){3}{\line(0,-1){0.12}}
\multiput(135.3,71.55)(0.15,-0.19){2}{\line(0,-1){0.19}}
\multiput(135.03,71.96)(0.13,-0.21){2}{\line(0,-1){0.21}}
\multiput(134.81,72.4)(0.11,-0.22){2}{\line(0,-1){0.22}}
\multiput(134.62,72.84)(0.1,-0.22){2}{\line(0,-1){0.22}}
\multiput(134.46,73.31)(0.16,-0.46){1}{\line(0,-1){0.46}}
\multiput(134.34,73.78)(0.12,-0.48){1}{\line(0,-1){0.48}}
\multiput(134.26,74.27)(0.08,-0.48){1}{\line(0,-1){0.48}}
\multiput(134.22,74.75)(0.04,-0.49){1}{\line(0,-1){0.49}}
\put(134.22,74.75){\line(0,1){0.49}}
\multiput(134.22,75.25)(0.04,0.49){1}{\line(0,1){0.49}}
\multiput(134.26,75.73)(0.08,0.48){1}{\line(0,1){0.48}}
\multiput(134.34,76.22)(0.12,0.48){1}{\line(0,1){0.48}}
\multiput(134.46,76.69)(0.16,0.46){1}{\line(0,1){0.46}}
\multiput(134.62,77.16)(0.1,0.22){2}{\line(0,1){0.22}}
\multiput(134.81,77.6)(0.11,0.22){2}{\line(0,1){0.22}}
\multiput(135.03,78.04)(0.13,0.21){2}{\line(0,1){0.21}}
\multiput(135.3,78.45)(0.15,0.19){2}{\line(0,1){0.19}}
\multiput(135.59,78.83)(0.11,0.12){3}{\line(0,1){0.12}}
\multiput(135.91,79.19)(0.12,0.11){3}{\line(1,0){0.12}}
\multiput(136.26,79.53)(0.13,0.1){3}{\line(1,0){0.13}}
\multiput(136.64,79.83)(0.2,0.13){2}{\line(1,0){0.2}}
\multiput(137.04,80.09)(0.21,0.12){2}{\line(1,0){0.21}}
\multiput(137.46,80.33)(0.22,0.1){2}{\line(1,0){0.22}}
\multiput(137.9,80.52)(0.45,0.16){1}{\line(1,0){0.45}}
\multiput(138.35,80.68)(0.46,0.12){1}{\line(1,0){0.46}}
\multiput(138.81,80.8)(0.47,0.08){1}{\line(1,0){0.47}}
\multiput(139.29,80.88)(0.48,0.04){1}{\line(1,0){0.48}}
\put(139.76,80.92){\line(1,0){0.48}}
\multiput(140.24,80.92)(0.48,-0.04){1}{\line(1,0){0.48}}
\multiput(140.71,80.88)(0.47,-0.08){1}{\line(1,0){0.47}}
\multiput(141.19,80.8)(0.46,-0.12){1}{\line(1,0){0.46}}
\multiput(141.65,80.68)(0.45,-0.16){1}{\line(1,0){0.45}}
\multiput(142.1,80.52)(0.22,-0.1){2}{\line(1,0){0.22}}
\multiput(142.54,80.33)(0.21,-0.12){2}{\line(1,0){0.21}}
\multiput(142.96,80.09)(0.2,-0.13){2}{\line(1,0){0.2}}
\multiput(143.36,79.83)(0.13,-0.1){3}{\line(1,0){0.13}}
\multiput(143.74,79.53)(0.12,-0.11){3}{\line(1,0){0.12}}
\multiput(144.09,79.19)(0.11,-0.12){3}{\line(0,-1){0.12}}
\multiput(144.41,78.83)(0.15,-0.19){2}{\line(0,-1){0.19}}
\multiput(144.7,78.45)(0.13,-0.21){2}{\line(0,-1){0.21}}
\multiput(144.97,78.04)(0.11,-0.22){2}{\line(0,-1){0.22}}
\multiput(145.19,77.6)(0.1,-0.22){2}{\line(0,-1){0.22}}
\multiput(145.38,77.16)(0.16,-0.46){1}{\line(0,-1){0.46}}
\multiput(145.54,76.69)(0.12,-0.48){1}{\line(0,-1){0.48}}
\multiput(145.66,76.22)(0.08,-0.48){1}{\line(0,-1){0.48}}
\multiput(145.74,75.73)(0.04,-0.49){1}{\line(0,-1){0.49}}

\linethickness{0.3mm}
\put(145.78,44.75){\line(0,1){0.49}}
\multiput(145.74,44.27)(0.04,0.49){1}{\line(0,1){0.49}}
\multiput(145.66,43.78)(0.08,0.48){1}{\line(0,1){0.48}}
\multiput(145.54,43.31)(0.12,0.48){1}{\line(0,1){0.48}}
\multiput(145.38,42.84)(0.16,0.46){1}{\line(0,1){0.46}}
\multiput(145.19,42.4)(0.1,0.22){2}{\line(0,1){0.22}}
\multiput(144.97,41.96)(0.11,0.22){2}{\line(0,1){0.22}}
\multiput(144.7,41.55)(0.13,0.21){2}{\line(0,1){0.21}}
\multiput(144.41,41.17)(0.15,0.19){2}{\line(0,1){0.19}}
\multiput(144.09,40.81)(0.11,0.12){3}{\line(0,1){0.12}}
\multiput(143.74,40.47)(0.12,0.11){3}{\line(1,0){0.12}}
\multiput(143.36,40.17)(0.13,0.1){3}{\line(1,0){0.13}}
\multiput(142.96,39.91)(0.2,0.13){2}{\line(1,0){0.2}}
\multiput(142.54,39.67)(0.21,0.12){2}{\line(1,0){0.21}}
\multiput(142.1,39.48)(0.22,0.1){2}{\line(1,0){0.22}}
\multiput(141.65,39.32)(0.45,0.16){1}{\line(1,0){0.45}}
\multiput(141.19,39.2)(0.46,0.12){1}{\line(1,0){0.46}}
\multiput(140.71,39.12)(0.47,0.08){1}{\line(1,0){0.47}}
\multiput(140.24,39.08)(0.48,0.04){1}{\line(1,0){0.48}}
\put(139.76,39.08){\line(1,0){0.48}}
\multiput(139.29,39.12)(0.48,-0.04){1}{\line(1,0){0.48}}
\multiput(138.81,39.2)(0.47,-0.08){1}{\line(1,0){0.47}}
\multiput(138.35,39.32)(0.46,-0.12){1}{\line(1,0){0.46}}
\multiput(137.9,39.48)(0.45,-0.16){1}{\line(1,0){0.45}}
\multiput(137.46,39.67)(0.22,-0.1){2}{\line(1,0){0.22}}
\multiput(137.04,39.91)(0.21,-0.12){2}{\line(1,0){0.21}}
\multiput(136.64,40.17)(0.2,-0.13){2}{\line(1,0){0.2}}
\multiput(136.26,40.47)(0.13,-0.1){3}{\line(1,0){0.13}}
\multiput(135.91,40.81)(0.12,-0.11){3}{\line(1,0){0.12}}
\multiput(135.59,41.17)(0.11,-0.12){3}{\line(0,-1){0.12}}
\multiput(135.3,41.55)(0.15,-0.19){2}{\line(0,-1){0.19}}
\multiput(135.03,41.96)(0.13,-0.21){2}{\line(0,-1){0.21}}
\multiput(134.81,42.4)(0.11,-0.22){2}{\line(0,-1){0.22}}
\multiput(134.62,42.84)(0.1,-0.22){2}{\line(0,-1){0.22}}
\multiput(134.46,43.31)(0.16,-0.46){1}{\line(0,-1){0.46}}
\multiput(134.34,43.78)(0.12,-0.48){1}{\line(0,-1){0.48}}
\multiput(134.26,44.27)(0.08,-0.48){1}{\line(0,-1){0.48}}
\multiput(134.22,44.75)(0.04,-0.49){1}{\line(0,-1){0.49}}
\put(134.22,44.75){\line(0,1){0.49}}
\multiput(134.22,45.25)(0.04,0.49){1}{\line(0,1){0.49}}
\multiput(134.26,45.73)(0.08,0.48){1}{\line(0,1){0.48}}
\multiput(134.34,46.22)(0.12,0.48){1}{\line(0,1){0.48}}
\multiput(134.46,46.69)(0.16,0.46){1}{\line(0,1){0.46}}
\multiput(134.62,47.16)(0.1,0.22){2}{\line(0,1){0.22}}
\multiput(134.81,47.6)(0.11,0.22){2}{\line(0,1){0.22}}
\multiput(135.03,48.04)(0.13,0.21){2}{\line(0,1){0.21}}
\multiput(135.3,48.45)(0.15,0.19){2}{\line(0,1){0.19}}
\multiput(135.59,48.83)(0.11,0.12){3}{\line(0,1){0.12}}
\multiput(135.91,49.19)(0.12,0.11){3}{\line(1,0){0.12}}
\multiput(136.26,49.53)(0.13,0.1){3}{\line(1,0){0.13}}
\multiput(136.64,49.83)(0.2,0.13){2}{\line(1,0){0.2}}
\multiput(137.04,50.09)(0.21,0.12){2}{\line(1,0){0.21}}
\multiput(137.46,50.33)(0.22,0.1){2}{\line(1,0){0.22}}
\multiput(137.9,50.52)(0.45,0.16){1}{\line(1,0){0.45}}
\multiput(138.35,50.68)(0.46,0.12){1}{\line(1,0){0.46}}
\multiput(138.81,50.8)(0.47,0.08){1}{\line(1,0){0.47}}
\multiput(139.29,50.88)(0.48,0.04){1}{\line(1,0){0.48}}
\put(139.76,50.92){\line(1,0){0.48}}
\multiput(140.24,50.92)(0.48,-0.04){1}{\line(1,0){0.48}}
\multiput(140.71,50.88)(0.47,-0.08){1}{\line(1,0){0.47}}
\multiput(141.19,50.8)(0.46,-0.12){1}{\line(1,0){0.46}}
\multiput(141.65,50.68)(0.45,-0.16){1}{\line(1,0){0.45}}
\multiput(142.1,50.52)(0.22,-0.1){2}{\line(1,0){0.22}}
\multiput(142.54,50.33)(0.21,-0.12){2}{\line(1,0){0.21}}
\multiput(142.96,50.09)(0.2,-0.13){2}{\line(1,0){0.2}}
\multiput(143.36,49.83)(0.13,-0.1){3}{\line(1,0){0.13}}
\multiput(143.74,49.53)(0.12,-0.11){3}{\line(1,0){0.12}}
\multiput(144.09,49.19)(0.11,-0.12){3}{\line(0,-1){0.12}}
\multiput(144.41,48.83)(0.15,-0.19){2}{\line(0,-1){0.19}}
\multiput(144.7,48.45)(0.13,-0.21){2}{\line(0,-1){0.21}}
\multiput(144.97,48.04)(0.11,-0.22){2}{\line(0,-1){0.22}}
\multiput(145.19,47.6)(0.1,-0.22){2}{\line(0,-1){0.22}}
\multiput(145.38,47.16)(0.16,-0.46){1}{\line(0,-1){0.46}}
\multiput(145.54,46.69)(0.12,-0.48){1}{\line(0,-1){0.48}}
\multiput(145.66,46.22)(0.08,-0.48){1}{\line(0,-1){0.48}}
\multiput(145.74,45.73)(0.04,-0.49){1}{\line(0,-1){0.49}}

\linethickness{0.3mm}
\put(145.78,14.75){\line(0,1){0.49}}
\multiput(145.74,14.27)(0.04,0.49){1}{\line(0,1){0.49}}
\multiput(145.66,13.78)(0.08,0.48){1}{\line(0,1){0.48}}
\multiput(145.54,13.31)(0.12,0.48){1}{\line(0,1){0.48}}
\multiput(145.38,12.84)(0.16,0.46){1}{\line(0,1){0.46}}
\multiput(145.19,12.4)(0.1,0.22){2}{\line(0,1){0.22}}
\multiput(144.97,11.96)(0.11,0.22){2}{\line(0,1){0.22}}
\multiput(144.7,11.55)(0.13,0.21){2}{\line(0,1){0.21}}
\multiput(144.41,11.17)(0.15,0.19){2}{\line(0,1){0.19}}
\multiput(144.09,10.81)(0.11,0.12){3}{\line(0,1){0.12}}
\multiput(143.74,10.47)(0.12,0.11){3}{\line(1,0){0.12}}
\multiput(143.36,10.17)(0.13,0.1){3}{\line(1,0){0.13}}
\multiput(142.96,9.91)(0.2,0.13){2}{\line(1,0){0.2}}
\multiput(142.54,9.67)(0.21,0.12){2}{\line(1,0){0.21}}
\multiput(142.1,9.48)(0.22,0.1){2}{\line(1,0){0.22}}
\multiput(141.65,9.32)(0.45,0.16){1}{\line(1,0){0.45}}
\multiput(141.19,9.2)(0.46,0.12){1}{\line(1,0){0.46}}
\multiput(140.71,9.12)(0.47,0.08){1}{\line(1,0){0.47}}
\multiput(140.24,9.08)(0.48,0.04){1}{\line(1,0){0.48}}
\put(139.76,9.08){\line(1,0){0.48}}
\multiput(139.29,9.12)(0.48,-0.04){1}{\line(1,0){0.48}}
\multiput(138.81,9.2)(0.47,-0.08){1}{\line(1,0){0.47}}
\multiput(138.35,9.32)(0.46,-0.12){1}{\line(1,0){0.46}}
\multiput(137.9,9.48)(0.45,-0.16){1}{\line(1,0){0.45}}
\multiput(137.46,9.67)(0.22,-0.1){2}{\line(1,0){0.22}}
\multiput(137.04,9.91)(0.21,-0.12){2}{\line(1,0){0.21}}
\multiput(136.64,10.17)(0.2,-0.13){2}{\line(1,0){0.2}}
\multiput(136.26,10.47)(0.13,-0.1){3}{\line(1,0){0.13}}
\multiput(135.91,10.81)(0.12,-0.11){3}{\line(1,0){0.12}}
\multiput(135.59,11.17)(0.11,-0.12){3}{\line(0,-1){0.12}}
\multiput(135.3,11.55)(0.15,-0.19){2}{\line(0,-1){0.19}}
\multiput(135.03,11.96)(0.13,-0.21){2}{\line(0,-1){0.21}}
\multiput(134.81,12.4)(0.11,-0.22){2}{\line(0,-1){0.22}}
\multiput(134.62,12.84)(0.1,-0.22){2}{\line(0,-1){0.22}}
\multiput(134.46,13.31)(0.16,-0.46){1}{\line(0,-1){0.46}}
\multiput(134.34,13.78)(0.12,-0.48){1}{\line(0,-1){0.48}}
\multiput(134.26,14.27)(0.08,-0.48){1}{\line(0,-1){0.48}}
\multiput(134.22,14.75)(0.04,-0.49){1}{\line(0,-1){0.49}}
\put(134.22,14.75){\line(0,1){0.49}}
\multiput(134.22,15.25)(0.04,0.49){1}{\line(0,1){0.49}}
\multiput(134.26,15.73)(0.08,0.48){1}{\line(0,1){0.48}}
\multiput(134.34,16.22)(0.12,0.48){1}{\line(0,1){0.48}}
\multiput(134.46,16.69)(0.16,0.46){1}{\line(0,1){0.46}}
\multiput(134.62,17.16)(0.1,0.22){2}{\line(0,1){0.22}}
\multiput(134.81,17.6)(0.11,0.22){2}{\line(0,1){0.22}}
\multiput(135.03,18.04)(0.13,0.21){2}{\line(0,1){0.21}}
\multiput(135.3,18.45)(0.15,0.19){2}{\line(0,1){0.19}}
\multiput(135.59,18.83)(0.11,0.12){3}{\line(0,1){0.12}}
\multiput(135.91,19.19)(0.12,0.11){3}{\line(1,0){0.12}}
\multiput(136.26,19.53)(0.13,0.1){3}{\line(1,0){0.13}}
\multiput(136.64,19.83)(0.2,0.13){2}{\line(1,0){0.2}}
\multiput(137.04,20.09)(0.21,0.12){2}{\line(1,0){0.21}}
\multiput(137.46,20.33)(0.22,0.1){2}{\line(1,0){0.22}}
\multiput(137.9,20.52)(0.45,0.16){1}{\line(1,0){0.45}}
\multiput(138.35,20.68)(0.46,0.12){1}{\line(1,0){0.46}}
\multiput(138.81,20.8)(0.47,0.08){1}{\line(1,0){0.47}}
\multiput(139.29,20.88)(0.48,0.04){1}{\line(1,0){0.48}}
\put(139.76,20.92){\line(1,0){0.48}}
\multiput(140.24,20.92)(0.48,-0.04){1}{\line(1,0){0.48}}
\multiput(140.71,20.88)(0.47,-0.08){1}{\line(1,0){0.47}}
\multiput(141.19,20.8)(0.46,-0.12){1}{\line(1,0){0.46}}
\multiput(141.65,20.68)(0.45,-0.16){1}{\line(1,0){0.45}}
\multiput(142.1,20.52)(0.22,-0.1){2}{\line(1,0){0.22}}
\multiput(142.54,20.33)(0.21,-0.12){2}{\line(1,0){0.21}}
\multiput(142.96,20.09)(0.2,-0.13){2}{\line(1,0){0.2}}
\multiput(143.36,19.83)(0.13,-0.1){3}{\line(1,0){0.13}}
\multiput(143.74,19.53)(0.12,-0.11){3}{\line(1,0){0.12}}
\multiput(144.09,19.19)(0.11,-0.12){3}{\line(0,-1){0.12}}
\multiput(144.41,18.83)(0.15,-0.19){2}{\line(0,-1){0.19}}
\multiput(144.7,18.45)(0.13,-0.21){2}{\line(0,-1){0.21}}
\multiput(144.97,18.04)(0.11,-0.22){2}{\line(0,-1){0.22}}
\multiput(145.19,17.6)(0.1,-0.22){2}{\line(0,-1){0.22}}
\multiput(145.38,17.16)(0.16,-0.46){1}{\line(0,-1){0.46}}
\multiput(145.54,16.69)(0.12,-0.48){1}{\line(0,-1){0.48}}
\multiput(145.66,16.22)(0.08,-0.48){1}{\line(0,-1){0.48}}
\multiput(145.74,15.73)(0.04,-0.49){1}{\line(0,-1){0.49}}

\put(50,65){\makebox(0,0)[cc]{$a+2t$}}

\put(80,65){\makebox(0,0)[cc]{$a+2t$}}

\put(110,65){\makebox(0,0)[cc]{$a$}}

\put(140,65){\makebox(0,0)[cc]{$a$}}

\put(50,35){\makebox(0,0)[cc]{$a+2t$}}

\put(80,35){\makebox(0,0)[cc]{$a+2t-1$}}

\put(110,35){\makebox(0,0)[cc]{$a+1$}}

\put(140,35){\makebox(0,0)[cc]{$a+1$}}

\put(40,5){\makebox(0,0)[cc]{$a+2t-i+1$}}

\put(85,5){\makebox(0,0)[cc]{$a+2t-i$}}

\put(115,5){\makebox(0,0)[cc]{$a+i$}}

\put(140,5){\makebox(0,0)[cc]{$a+i$}}

\end{picture}